\documentclass[11pt]{article}
\usepackage{amsmath,amsthm,amssymb,amscd}
\usepackage[longnamesfirst]{natbib}
\usepackage{color}
\usepackage{lipsum}
\usepackage{geometry}
\usepackage{enumerate}
\usepackage{threeparttable}
\usepackage[dvipsnames]{xcolor}
\usepackage{enumitem}
\usepackage{setspace}
\usepackage{authblk} 
\usepackage{hyperref}
\hypersetup{colorlinks=true, pdfauthor=author, allcolors=blue}

\usepackage[figuresright]{rotating}

\def\I{{\bf I}}

\def\S{{\bf S}}

\def\X{{\bf X}}

\def\Z{{\bf Z}}
\def\z{{\bf z}}

\def\calP{{\cal P}}
\def\calQ{{\cal Q}}

\def\ba{{\boldsymbol\alpha}}

\def\bt{{\boldsymbol\theta}}

\def\bb{{\boldsymbol\beta}}

\def\0{{\bf 0}}
\def\trans{^{\rm T}}
\def\pr{\hbox{pr}}
\def\wh{\widehat}
\def\wt{\widetilde}

\def\var{\hbox{var}}

\def\n_1{\nonumber}

\def\log{{\rm log}}

\def\squarebox#1{\hbox to #1{\hfill\vbox to #1{\vfill}}}
\def\povr{\buildrel p\over\longrightarrow}
\def\dovr{\buildrel d\over\longrightarrow}

\def\bse{\begin{eqnarray*}}
\def\ese{\end{eqnarray*}}
\def\be{\begin{eqnarray}}
\def\ee{\end{eqnarray}}
\def\bsq{\begin{equation*}}
\def\esq{\end{equation*}}
\def\bq{\begin{equation}}
\def\eq{\end{equation}}

\def\sumIP1{\sum_{i=1, i\in P_1}^n}

\def\boxit#1{\vbox{\hrule\hbox{\vrule\kern6pt\vbox{\kern6pt#1\kern6pt}\kern6pt\vrule}\hrule}}

\newtheorem{Lemma}{Lemma}[section] 
\newtheorem{Remark}{Remark}[section] 
\newtheorem{Theorem}{Theorem}[section] 
\newtheorem{Corollary}{Corollary}[section] 

\newtheorem*{Assumption*}{Assumption}
\newtheorem{Example}{Example}[section]

\makeatletter
\renewcommand\theequation{\thesection.\arabic{equation}}
\@addtoreset{equation}{section}

\title{Survival analysis under label shift}
\makeatother

\author[1]{Yuxiang Zong}
\author[2]{Yanyuan Ma}
\author[1]{Ingrid Van Keilegom}
\affil[1]{\small Research Centre for Operations Research and Statistics, KU Leuven, Naamsestraat 69,
3000 Leuven, Belgium}
\affil[2]{\small Department of Statistics, Pennsylvania State University, University Park, PA 16802, USA}\date{\today}

\date{ }
\usepackage{color}

\begin{document}

\maketitle

\begin{abstract}

Let $\calP$ represent the source population with complete data,
containing covariate $\Z$ and response $T$, and 
$\calQ$ the target population, where only the covariate $\Z$ is
available. We consider a setting with both label shift and label
censoring. Label shift assumes that the marginal distribution of $T$
differs between $\calP$ and $\calQ$, while the conditional
distribution of $\mathbf{Z}$ given $T$ remains the same. Label
censoring refers to the case where the response $T$ in $\calP$ is
subject to random censoring. Our goal is to leverage information from
the label-shifted and label-censored source population $\mathcal{P}$
to conduct statistical inference in the target population
$\mathcal{Q}$. We propose a parametric model for $T$ given $\Z$ in
$\calQ$ and estimate the model parameters by maximizing an
   approximate likelihood. This allows for statistical inference in $\calQ$ and
accommodates a range of classical survival models. Under the label
shift assumption, the likelihood depends not only on the unknown
parameters but also on the unknown distribution of $T$ in $\calP$ and
$ \Z$ in $\calQ$, which we estimate nonparametrically. The asymptotic
properties of the estimator are rigorously established and the 
effectiveness of the method is demonstrated through simulations and
a real data application. This work is the first to combine
survival analysis with label shift, offering a new research
  direction in this emerging topic. 

\end{abstract}

\section{Introduction}

Traditional supervised learning methods assume that the source
population $\mathcal{P}$ and the target population $\mathcal{Q}$ share
the same joint distribution of the response $T$ and covariate
$\mathbf{Z}$. Under this assumption, a model trained on data from
$\mathcal{P}$ can be directly applied to $\mathcal{Q}$ for inference
or prediction. However, in practice, the joint distribution of $T$ and
$\Z$ often differs between the two populations, a phenomenon known as
distributional shift.

Methods to address distributional shift can be broadly
classified into two categories based on the availability of target
data. The first is domain adaptation (DA), where target data is available,
either fully or partially~\citep{MORENOTORRES2012521,Kouw2021,qiu2024efficient}.
In handling DA, methods need to be tailored to
the specific characteristics of the target population.
The second is domain generalization (DG), where no target data is
accessible, and the goal is to develop models that generalize well to
unseen distributions~\citep{zhou2023}, i.e., robustness across
arbitrary potential distributional shifts. This paper belongs
to the domain
adaptation category, where typically the source population $\calP$ has fully
observed response and covariates, while in the target population
$\calQ$, only the covariate $\Z$ is available.

To take into account DA, appropriate assumptions
are essential to effectively leverage information from the source
population for analysis in the target population. There are three
primary types of distributional shift assumptions considered.
The first is concept shift, where the conditional distribution of the
response given the covariates differs between populations,
while the marginal distribution of the covariates remains the same,
i.e., $p_{T \mid \mathbf{Z}}(t, \mathbf{z}) \neq$ $q_{T \mid
  \mathbf{Z}}(t, \mathbf{z})$ but
$p_{\mathbf{Z}}(\mathbf{z})=q_{\mathbf{Z}}(\mathbf{z})$~\citep{Kouw2021}.
The second is covariate shift which assumes the conditional
distribution of the response given the covariates remains unchanged,
while the covariate distribution differs, i.e., 
$p_{T \mid \mathbf{Z}}(t, \mathbf{z})=q_{T \mid \mathbf{Z}}(t,
\mathbf{z})$ but $p_{\mathbf{Z}}(\mathbf{z}) \neq
q_{\mathbf{Z}}(\mathbf{z})$~\citep{JMLR:v8:sugiyama07a}.
The third is label shift,  the focus of this paper, which
assumes that the marginal distribution of the response differs between
populations, while the conditional distribution of covariates given
the response remains unchanged, i.e., 
$p_T(t) \neq q_T(t)$ but $p_{\mathbf{Z} \mid T}(\mathbf{z},
t)=q_{\mathbf{Z} \mid T}(\mathbf{z},
t)$~\citep{garg2020unified}. Label shift is commonly associated with
anti-causal problems, where the response variable is the cause of the
covariates. For instance, in a study of an infectious disease
outbreak, the source 
data might be collected at the onset of the outbreak, while the target
data is gathered after preventive measures have been implemented. The
response variable in both datasets is the diagnosis of the disease,
with covariates representing the biomedical factors of the
patients. As preventive efforts reduce the probability of infection,
the label distribution shifts~\citep{tian2023elsa}. However, the conditional distribution
of clinical behavior given a diagnosis remains consistent, making
label shift methods especially applicable in such
cases.

Label shift naturally arises in survival data as soon as the responses are
subject to random censoring. However, works addressing label shift
issues in the context of survival analysis are absent in
the literature so far. This paper is the first attempt in this line of
research. From the practical application point of view, our
motivation comes from the liver transplant dataset provided by the
Organ Procurement and Transplantation Network (OPTN). We
consider the lifespan of transplanted livers in adult recipients as
the response variable, measured from the date of transplantation to
either organ failure or patient death. The covariates include patient
and donor information, as well as clinical factors related to the
longevity of the transplanted organ. The response is subject to
censoring due to patients dropping out of study or losing contact. The two
populations consist of patients with private insurance and those
insured through Medicaid or Medicare, with the latter group showing
shorter survival time, indicating the possibility of a label shift model. To conduct a careful study of this dataset, we proceed
to propose a model framework in the context of label shift in survival
analysis, and further develop an estimation and inference procedure.

More specifically, we study the label shift problem when the
response is randomly right censored in the source data, while
only covariates are available in the target data. Since the goal is to
estimate the effects of covariate $\Z$ on the response $T$ in the
target population $\calQ$, based on the label-shifted and
label-censored data in the source population $\calP$, we
adopt a parametric model to describe the covariate effect to be
estimated in the target population. Starting from a fully parametric model under study, the resulting problem is far from parametric.
The challenge arises not only from the label shift phenomenon,
where the population of interest contributes no observed response at
all, but also from the incompleteness of the source data due to
censoring, which makes the information we want to borrow incomplete.
The parametric model assumed in $\calQ$ is
flexible. It encompasses various
  structures, such as parametric proportional hazards models,
  parametric accelerated failure time models or many other widely used
  survival models. 
We propose a nonparametric maximum likelihood based method to estimate
the unknown parameters, 
which allows us to assess the effect of covariates on the survival time in
the target population and perform statistical inference. The presence
of both label shift and censoring makes it harder to derive the
statistical properties of the estimator, which we address using
techniques from empirical process
theory~\citep{van1996weak,kosorok2007introduction}.  We develop the
asymptotic theory of the proposed estimator and provide explicit
formulas for its asymptotic variance, which can be efficiently
estimated from the data without relying on computationally intensive
bootstrap procedures. 

To the best of our knowledge, this is the first work in
survival analysis under label shift. We propose a novel framework that
enables statistical inference in a target population without
  observing any
labels, by using information from a label-shifted and label-censored
source population. The framework is flexible and has the potential to
be extended to other survival settings under label shift. We hope this
work will stimulate further research and inspire further works in
this emerging area. 

The paper is structured as follows. Section 2 provides a literature
review on label shift and survival analysis under distributional
shift. In Section 3, we introduce our model, describe the estimation
procedure, and establish the theoretical properties of our
estimator. Section 4 presents the simulation results. In Section 5, we
illustrate our method on a liver transplant dataset. Finally, Section
6 concludes the paper and discusses potential directions for future
research. Technical details, along with additional tables and figures,
are provided in the Supplement.

\section{Literature Review}
Most existing studies on label shift have focused exclusively on
either classification (discrete labels) or regression (continuous
labels). A central approach in the literature is to estimate the
density ratio $\rho(t)\equiv q_T(t)/p_T(t)$, which is referred to
as the importance weight in classification. Most proposed
methods hinge on the equality
\be
q_\Z(\z) = p_\Z(\z)\int p_{T\mid\Z}(t,\z)\rho(t) dt,
\label{eq1}
\ee  
and estimate the density ratio by minimizing the
discrepancy between an estimator $\wh
q_{\mathbf{Z}}(\mathbf{z})$ and the estimator
$\wt q_\Z(\z,\rho)\equiv\wh p_\Z(\z)\int \wh p_{T\mid\Z}(t,\z)\rho(t) dt$ , where $\wh p_\Z(\z)$ and $\wh p_{T\mid\Z}(t,\z)$ are appropriate estimators for marginal and conditional densities respectively.

In classification problems, estimating the importance weight reduces
to estimating a finite number of
parameters. \cite{lipton2018detecting} and
\cite{azizzadenesheli2019regularized} proposed methods based on the
idea that a model $f$ trained on population $\calP$ preserves the
equality $p_{\mathbf{Z} \mid T}(f(\mathbf{z}), t)=q_{\mathbf{Z} \mid
  T}(f(\mathbf{z}), t)$. \cite{lipton2018detecting} introduced Black Box Shift Estimation
(BBSE), a moment-matching method that estimates importance weights
using the inverse of the confusion matrix of a classifier trained on
population $\calP$. However, BBSE assumes that the confusion matrix is
non-singular, which is a strong requirement. To address this
limitation, \cite{azizzadenesheli2019regularized} proposed the
Regularized Learning Method under Label Shift (RLLS), which
incorporates a regularization term to protect against singularity issues. Despite these refinements, confusion matrix-based
approaches often perform poorly in high-dimensional settings~\citep{lipton2018detecting}.

To overcome these challenges, \cite{alexandari2020maximum} introduced
the Maximum Likelihood Label Shift (MLLS) method, which incorporates a
calibration step and outperforms both BBSE and RLLS. However, the
calibration process in MLLS is complicated, and its performance
greatly depends on the selected calibration method. To address these
issues, \cite{tian2023elsa} developed an approach based on
semi-parametric theory, estimating importance weights by constructing
an efficient influence function and solving an estimating
equation. Their method achieves $\sqrt{n}-$consistency and asymptotic
normality, demonstrating strong performance without requiring
calibration.

In regression problems, estimating the unknown density ratio function involves solving the integral equation (\ref{eq1}), making the estimation process more complex.
Most of the existing
literature has focused on characterizing various discrepancies between
$\wh{q}_{\Z}(\z)$ and $\wt{q}_{\Z}(\z,\rho)$, different ways of estimating $p_{T \mid \Z}(t, \z)$, $p_\Z(\z)$ and
$q_{\Z}(\z)$ in \eqref{eq1}, or simplifying the computation of the integral in
 (\ref{eq1}).
\cite{zhang2013domain} employed the kernel mean embedding method to
simplify the integral computation in (\ref{eq1}) and estimated the
density ratio $\rho(t)$ by minimizing the maximum mean discrepancy
between the embedded representations of $\wh q_\Z(\z)$ and $\wt
q_\Z(\z, \rho)$.
\cite{nguyen2016continuous} modelled the density ratio $\rho(t)$ by a
nonparameteric 
Gaussian kernel model $\rho^*(t,\ba)$ and estimated the parameter
$\ba$ by minimizing the $L^2$ distance 
between $\wh q_\Z(\z)$ and $\wt q_\Z\{\z,\rho^*(\cdot,\ba)\}$.
\cite{guo2020ltf} proposed a label transformation framework $f$
  that transfers $T$ in population $\calP$ such that $p_T\{f(t)\}$ can
  mimic $q_T(t)$ through minimizing the difference between $\wh
  q_\Z(\z)$ and 
  $\wt q_\Z[\z,p_T\{f(\cdot)\}/p_T(\cdot)]$. To
  model the invariant conditional distribution, \cite{guo2020ltf}
  utilized a conditional generator and employed Jensen-Shannon
  Divergence to quantify the difference between $\wh{q}_{\Z}(\z)$ and
  $\wt{q}_{\Z}[\z, p_T\{f(\cdot)\}/p_T(\cdot)]$. \cite{kim2024retasa} directly solved the
  integral equation (\ref{eq1}) for density ratio estimation. To
  mitigate the ill-posed nature of this problem, \cite{kim2024retasa}
  applied a nonparametric regularization method, using a generalized
  kernel function with regularization to stabilize the solution.  

\cite{lee2024doubly} took a completely different approach
    and proposed a doubly flexible approach applicable to both
  regression and classification problems based on semiparametric
  theory~\citep{tsiatis2007semiparametric}. Their method offers double
  flexibility by allowing misspecification in both the density ratio
  model and the regression model of $T$ given $\Z$ in population
  $\calP$, thereby avoiding  the most difficult aspects of the label
  shift problem-the direct estimation of the density ratio. 
However, all these methods
  assume that population $\calP$ has complete data, whereas the
  presence of censoring significantly complicates the estimation
  process, e.g. the integral computation or the estimation of various functions in \eqref{eq1}. Hence, none of the forementioned methods can
  be directly extended to the settings with censoring.

The integration of survival analysis and label shift remains largely
unexplored, although general DA and DG methods in survival
analysis have received some attention. \cite{jeanselme2022deepjoint}
addressed clinical presence shift, a combination of covariate and
concept shift, assuming an underlying stable distribution while
variations in hospitals and regions introduce shifts in the
observation process. They employed recurrent neural networks to model
the observation process and then recover the underlying stable
distribution. 
However, they do not provide any theoretical guarantees.
\cite{pmlr-v174-pfisterer22a} studied DG in survival analysis,
assuming no access to target data. They benchmarked multiple DG
methods and analyzed the impact of different types of shifts on model
performance. \cite{wang2022survm} leveraged information from multiple source populations to enhance estimation in the target population
under covariate shift and partial concept shift. They assumed that both source and target populations follow a Cox model with potentially
different baseline cumulative hazards and have shifted covariate
distributions. The method estimates the covariate effect in the target
population by maximizing the minimum reward across all source
populations. 
Despite these works, survival analysis under label shift with
unsupervised target data is entirely unexplored. 
Our work will fill this research gap by developing a 
method that effectively handles label shift while accounting for
censored data, and we hope to attract more research on survival analysis
under the label shift framework.

\section{Methodology}

Let $T$ denote the event time of interest, and $\Z$ denote the
  covariate vector. The covariate $\Z$ may include both continuous and categorical variables. We let $p$ and $q$ denote the probability density or mass functions, and $P$ and $Q$ denote probability distribution functions in populations $\calP$ and $\calQ$, respectively. The operators $E_p$, $\var_p$, $E_q$ and $\var_q$ represent expectations and variances under populations $\calP$ and $\calQ$. In follow-up studies, event times are often
  censored due to patient dropout or the end of the studies. To account
  for this, we introduce the censoring time $C$ in population $\calP$,
  assuming that $C$ is independent of $(T,\Z)$. The response we can
  observe in population $\calP$ is then $X \equiv \min (T, C)$ and a
  censoring indicator $\Delta \equiv I(T\le C)$. In the target
  population $\calQ$, event time $T$ is unavailable, but the
  covariate $\Z$ is observed. In summary, in population
  $\calP$, we have an i.i.d. sample $(X_i, \Delta_i, \Z_i), i=1,
  \ldots, n_1$, while in $\calQ$, we have a sample $\Z_i, i=n_1+1,
  \dots, n_1+n_2$. We stack these two samples into a full sample and introduce a binary indicator $R$ to distinguish
    which population an observation comes from. An observation
    can then be written as $( \Z_i,R_i, R_i X_i, R_i \Delta_i),  i=1,2,
  \dots, n\equiv n_1+n_2.$ When $R_i=1$, the observation belongs to $\calP$, and
  when $R_i=0$, it belongs to $\calQ$. We let $ \pi_n \equiv n_1
  /(n_1+n_2)$ and let $\pi$ denote its limiting value as both $n_1$ and $n_2$ tend to infinity. We allow $\pi$ to be 0 or 1, meaning that the sample sizes from $\calP$ and $\calQ$ can grow at different
  rates. This flexibility accommodates scenarios where data from one
  population grows faster than that from the other.  

To facilitate information transfer from the source data, we adopt the
label shift assumption, which states that while the marginal
distributions of $T$ in $\calP$ and $\calQ$ may differ, the
conditional distribution of $\Z$ given $T$ remains the same across
populations, specifically 
$$
p_T(t) \neq q_T(t), \quad p_{\Z \mid T}(\z, t) = q_{\Z \mid T}(\z, t) \equiv f_{\Z \mid T}(\z, t).
$$
Our goal is to estimate the effect of covariates on the event time $T$ in population $\calQ$.
We assume that in population $\calQ$, the conditional distribution of
$T$ given $\Z$ follows a prespecified model parameterized by an unknown $d$-dimensional vector 
$\bt$, thus the conditional density function can be written as
$q_{T \mid \Z}(t, \z;\bt)$. The parameter $\bt$ can include both
regression coefficients and additional parameters governing the shape
of the conditional density. Other than a valid conditional
  probability density function, we do not impose any other special
  structure on $q_{T \mid \Z}(t, \z;\bt)$. Thus, our framework
includes many commonly used survival
models such as parametric proportional hazards
models, parametric accelerated failure time models, parametric
proportional odds models, etc. Our objective is to estimate $\bt$, thereby enabling statistical inference on survival
outcomes in the target population $\calQ$.

To formalize our estimation approach clearly, we first write
  out the joint probability density function (pdf) of
$(\Z,R,RX,R\Delta)$, which has the form
\be
f(\z, r, rx, r\delta)\n_1
&=&\left[\pi\left\{
f_{\Z\mid T}(\z,x)p_T(x)\int_x^\infty p_{C}(c)dc\right\}^\delta \left\{p_{C}(x)
\int_x^{\infty}
f_{\Z\mid T}(\z,t)p_T(t) dt\right\}^{1-\delta}\right]^r \n_1\\
&& \times\left\{(1-\pi)\int f_{\Z\mid T}(\z,t)q_T(t) d
  t\right\}^{1-r}.
  \label{eq:model0}
\ee
Since our goal is to estimate the unknown parameter $\bt$ in $q_{T\mid \Z}(t,\z;\bt)$, we plug $q_{T\mid \Z}(t,\z;\bt)$ in the joint density (\ref{eq:model0}). Hence, the joint density (\ref{eq:model0}) can be further written as
\be
&&f(\z, r, rx, r\delta)\n_1\\
&=&\pi^r(1-\pi)^{1-r}\left\{ \frac{p_T(x)q_{T\mid \Z}( x,\z;\bt)}{
\int q_\Z(\z) q_{T\mid \Z}( x,\z;\bt)
d\z}\right\}^{r\delta} \left\{\int_x^{\infty}\frac{ p_T(t)
  q_{T\mid \Z}(t,\z;\bt)}{
\int q_\Z(\z)q_{T\mid \Z}(t,\z;\bt)
d\z}dt\right\}^{r(1-\delta)}\n_1 \\
&& \times \left\{\int_x^\infty p_C(c)dc \right\}^{r\delta} \left\{p_C(x)
\right\}^{r(1-\delta)}q_\Z(\z).
\label{eq:model1}
\ee

Before proceeding with the estimation procedure, it is essential to
discuss the identifiability of the proposed model. Given that we have
i.i.d. observations $(X_i, \Delta_i, \Z_i)$ from 
$\calP$, due to the independence between $T$ and $C$, $p_C(c)$,
$p_{T\mid \Z}(t,\z)$ and $p_\Z(\z)$ are all
identifiable. In population $\calQ$, we
observe a random sample of $\Z$, enabling the identifiability of
$q_\Z(\z)$. But unfortunately, since $T$ in population $\calQ$
  is not available, $q_{T\mid \Z}(t,\z;\bt)$ may not be
identified. The following example illustrates the case where $q_{T\mid
  \Z}(t,\z;\bt)$ is not identifiable.

\begin{Example}
\label{Example}
    Consider a setting where, given $Z=z$, the event time $T$ follows a
Log-Normal distribution in population $\calP$, i.e., 
\begin{equation*}
    p_{T\mid Z}(t,z;\beta) \equiv \frac{1}{ \sqrt{2\pi}t \beta z} \exp\left\{-\frac{(\log t)^2}{2(\beta z)^2}\right\}.
\end{equation*}
The covariate $Z$ takes discrete values with equal probability:
$p_Z(1)=\pr(Z=1) = 0.5$ and $p_Z(2)=\pr(Z=2) = 0.5.$
In population $\calQ$, we can verify that there are two different marginal distributions for $T$,
\bse
q_T(t;\mu,\sigma) = \frac{1}{\sqrt{2\pi} \sigma t} \exp\left\{-\frac{(\log t
        - \mu)^2}{2 \sigma^2}\right\},
\ese
and
\bse
q_T(t;-\mu,\sigma) = \frac{1}{\sqrt{2\pi} \sigma t} \exp\left\{-\frac{(\log t + \mu)^2}{2 \sigma^2}\right\},
\ese
where $\mu > 0$ and $\sigma>0$, such that
\bse
    q_Z(z) &=& \int f_{Z\mid T}(z,t)q_T(t;\mu,\sigma)dt\\
     &=& \int f_{Z\mid T}(z,t) q_T(t;-\mu,\sigma) dt,
\ese
where the proof is provided in Supplement~S3.1.

Thus, we have two conditional densities of $T$ given $Z$ in population $\calQ$:
\bse
 q_{T \mid Z}(t, z ; \bt) = \frac{p_{T \mid Z}(t, z ; \beta)  p_Z(z)q_T(t;\mu,\sigma)}{p_T(t) q_Z(z)},
\ese
and
\bse
q_{T \mid Z}(t, z ; \wt \bt) = \frac{p_{T \mid Z}(t, z ; \beta)  p_Z(z)q_T(t;-\mu,\sigma)}{p_T(t) q_Z(z)},
\ese
where $\bt=(\beta, \mu, \sigma)$ and $\wt \bt=(\beta, -\mu, \sigma)$ leading to the same joint density (\ref{eq:model1}).
\end{Example}

To ensure the identifiability of the proposed model, we establish the
following lemma. This lemma provides a sufficient condition under
which the model parameters can be uniquely determined from the
observed data. The proof is presented in Appendix~A. 

\begin{Lemma}\label{lemma3.1}
  All unknown components in (\ref{eq:model0}) are identifiable if
$q_{T\mid \Z}(t,\z;\bt)/q_{T\mid \Z}(t,\z;\wt \bt)$ depends on $\z$
whenever $\bt\ne\wt\bt$ and  $\var_q(\Z)\neq 0$.
\end{Lemma}

\begin{Remark}

Note that the condition required in Lemma~\ref{lemma3.1} is not satisfied in Example~\ref{Example} as $$\frac{q_{T\mid \Z}(t,\z;\bt)}{q_{T\mid \Z}(t,\z;\wt \bt)}=\frac{q_T(t;\mu,\sigma)}{q_T(t;-\mu,\sigma)},$$ where $\bt=(\beta, \mu, \sigma)$ and $\wt \bt=(\beta, -\mu, \sigma)$, which does not depend on $\z$.

\end{Remark}

The condition required in Lemma~\ref{lemma3.1} is relatively
mild. Example~\ref{exp:ph} illustrates that the proportional
hazards model with a Weibull baseline hazard function satisfies this identifiability condition. In Supplement~S1, we further demonstrate that commonly used models, e.g. (among others) the
accelerated failure time model with a Log-Normal baseline hazard
function, the proportional odds model with a Log-Logistic baseline
survival function, and the accelerated hazards model, all satisfy the
identifiability condition.

\begin{Example}[Proportional hazards model] \label{exp:ph}
For the proportional hazards model with
  $\bt=(\bb\trans,\gamma,\lambda)\trans$, and
\bse
    q_{T \mid \Z}(t,\z;\bt) = h(t;\gamma,\lambda)\exp(\z\trans\bb)\exp\left\{-H(t;\gamma,\lambda) \exp(\z\trans\bb)\right\},
    \ese
where $h(t;\gamma,\lambda)\equiv \lambda\gamma t^{\gamma-1}$, we have that $q_{T \mid \Z}(t,\z;\bt)/q_{T \mid
  \Z}(t,\z;\wt\bt)$ depends on $\z$ whenever
$\bt \neq \wt\bt$. Hence by Lemma~\ref{lemma3.1}, the model is
identifiable. The proof is provided in Supplement~S3.1.
\end{Example}

Given the established identifiability, we now derive
methods for parameter estimation. Up to a constant additive
term that does not involve $\bt$, the 
loglikelihood of one single observation is given by
\bse
\ell(x, \z, \delta, r; \bt) &\equiv&
 r\delta \log \{q_{T\mid \Z}(x,\z;\bt)\}
-r\delta
\log\int q_\Z(\z) q_{T\mid \Z}(x,\z;\bt)dz
\\&&+ r
(1-\delta)\log
\int_x^{\infty}\frac{ p_T(t)q_{T\mid \Z}(t,\z;\bt)}{
\int q_\Z(\z)q_{T\mid \Z}(t,\z;\bt)d\z}dt.
\ese
Notably, in this log-likelihood, both $p_T(t)$ and $q_\Z(\z)$
are unknown. Fortunately, since we have access to $X$ in
population $\calP$ and $\Z$ in population $\calQ$, we can estimate
these distributions directly. Rather than imposing parametric
structures on $p_T(t)$ and $q_\Z(\z)$, we adopt nonparametric
approaches to maintain flexibility. Specifically, we use the
Kaplan-Meier estimator for $P_T(t)$, where $\wh P_T(t)=1-\prod_{j:X_{(j)}\leq t}\{y_{(j)}-d_{(j)}\}/y_{(j)}$, with $X_{(j)}$ denoting the ordered event times, $d_{(j)}$ the number of events at time $X_{(j)}$ and $y_{(j)}$ the size of the risk set at $X_{(j)}$ which effectively accounts for
censoring in survival data. For $Q_\Z(\z)$, we employ the empirical estimator $\wh Q_\Z(\z) \equiv n_2^{-1}\sum_{i=1}^n (1-R_i) I(\Z_i \leq \z)$, where the inequality is understood componentwise. With these estimates, the log-likelihood function can be
approximated as
\be
\label{eq:llh}
\ell_n(\bt)&=& n_1^{-1}
\sum_{i=1}^{n_1}
\left\{ \delta_i \log\{q_{T\mid \Z}(x_i,\z_i;\bt)\}
-\delta_i
\log\left(n_2^{-1}\sum_{j=1}^n (1-r_j)q_{T\mid \Z}(x_i,\z_j;\bt)\right)\right.\n_1\\
&&\left.+
(1-\delta_i)\log
\left(\sum_{k=1}^n\frac{ r_k \delta_k I(t_k>x_i)\{\wh S_p(t_k-)-\wh S_{p}(t_k)\}
  q_{T\mid \Z}(t_k,\z_i;\bt)}{
n_2^{-1}\sum_{j=1}^n (1-r_j) q_{T\mid \Z}(t_k,\z_j;\bt)}\right)
\right\},
\ee
where $\wh S_p(t)=1-\wh P_T(t)$ is the Kaplan-Meier estimator of $S_p(t)$, the marginal
survival function in population $\calP$. We subsequently obtain
the estimator $\wh \bt$ by maximizing the approximated log-likelihood
function in \eqref{eq:llh} with respect to $\bt$.

However, the presence of both label shift and censoring introduces
significant challenges in the 
theoretical analysis of our model. In particular, the third term in
the approximated log-likelihood function (\ref{eq:llh}) accounting for
censored observations has a nested structure where the estimated
$\wh Q_\Z(\z)$ appears within the Kaplan-Meier integral based on $\wh P_T(t)$. This
interdependence makes it difficult to establish the statistical
properties of the estimator. To address these technical challenges, we
employ techniques from empirical process
theory~\citep{kosorok2007introduction,van1996weak}. We establish the
asymptotic properties of $\wh \bt$ in the following two theorems.  Let
$n_0 \equiv \min(n_1,n_2)$ and let $\bt_0$ denote the true value of
the parameter $\bt$. The conditions required for the following lemmas
and theorems are presented in Supplement~S3.2.2.

\begin{Theorem}[Consistency]
\label{Th:3.1}
Under Conditions A1--A6, presented in Supplement~{S3.2.2},
 we have that $\wh \bt \povr \bt_0$  as $n_0 \rightarrow \infty$.
\end{Theorem}

With consistency established, we proceed to show the asymptotic
normality of $\wh \bt$. Since the estimator $\wh \bt$ is obtained by
maximizing the approximated loglikelihood function, under some
regularity conditions, $\wh \bt$ satisfies the estimating equation
$\Psi_n(\bt)=0$, where $\Psi_n(\bt)$ is given by \bse 
\Psi_{n}(\bt)&\equiv& n^{-1} \sum_{i=1}^{n} \frac{R_i}{\pi} \psi_{n}(X_i,\Z_i,\Delta_i;\bt),
\ese 
and $\tfrac{r}{\pi}\psi_n(\cdot)$ is the approximated score
function. As in standard asymptotic theory, the first step is to study
the asymptotic behavior of $\Psi_n(\bt_0)$. However, technical
challenges arise because $\tfrac{r}{\pi}\psi_n(\cdot)$ involves
nonparametrically estimated components $\wh P_T(t)$ and $\wh
Q_\Z(\z)$, which precludes the direct application of the classical
central limit theorem. 

To address this, we decompose the approximated score function $\tfrac{r}{\pi}\psi_n(\cdot)$ into three components:
\bse
\frac{r}{\pi}\psi_n(x,\z,\delta;\bt) 
\equiv \psi_{1}(x,\z,\delta,r;\bt) - \psi_{2n}(x,\delta,r;\bt) + \psi_{3n}(x,\z,\delta,r;\bt),
\ese
where \bse
\psi_{1}(x,\z,\delta,r;\bt) &\equiv& \frac{r \delta}{\pi} \frac{\frac{\partial}{\partial \bt} q_{T \mid \Z}(x, \z ; \bt)}{q_{T \mid \Z}(x, \z ; \bt)}, \\
\psi_{2n}(x,\delta,r;\bt) &\equiv & \frac{r \delta}{\pi} \frac{\int \frac{\partial}{\partial \bt}  q_{T\mid\Z}(x,\z;\bt)d\wh Q_\Z(\z)}{\int q_{T\mid\Z}(x,\z;\bt)d\wh Q_\Z(\z)}, \\ \psi_{3n}(x,\z,\delta,r;\bt) &\equiv & 
\frac{r (1-\delta)}{\pi} \frac{\int I(t>x)\left[ \frac{\frac{\partial}{\partial \bt} q_{T \mid Z}(t, \z ; \bt)}{\int q_{T \mid \Z}(t, \z ; \bt) d \wh Q_\Z(\z)}-\frac{q_{T \mid \Z}(t, \z ; \bt) \int \frac{\partial}{\partial \bt} q_{T\mid \Z}(t, \z ; \bt) d \wh Q_\Z(\z)}{\{\int q_{T \mid \Z}(t, \z ; \bt) d \wh Q_\Z(\z)\}^2}\right] d \wh P_T(t)}{\int  \frac{I(t>x) q_{T \mid \Z}(t, \z ; \bt)}{\int q_{T \mid \Z}(t, \z ; \bt) d \wh Q_\Z(\z)} d \wh P_T(t)}.\\
\ese
Similarly, the score function based on the true distributions $P_T(t)$ and $Q_\Z(\z)$ can also be decomposed as:
\bse
\frac{r}{\pi}\psi(x,\z,\delta;\bt) 
& \equiv & \psi_{1}(x,\z,\delta,r;\bt) - \psi_{2}(x,\delta,r;\bt) + \psi_{3}(x,\z,\delta,r;\bt),
\ese
where 
\bse
\psi_{2}(x,\delta,r;\bt) &\equiv & \frac{r \delta}{\pi} \frac{\int \frac{\partial}{\partial \bt}  q_{T\mid\Z}(x,\z;\bt)d Q_\Z(\z)}{\int q_{T\mid\Z}(x,\z;\bt)d Q_\Z(\z)}, \\ \psi_{3}(x,\z,\delta,r;\bt) &\equiv & 
\frac{r (1-\delta)}{\pi} \frac{\int I(t>x)\left[ \frac{\frac{\partial}{\partial \bt} q_{T \mid Z}(t, \z ; \bt)}{\int q_{T \mid \Z}(t, \z ; \bt) d  Q_\Z(\z)}-\frac{q_{T \mid \Z}(t, \z ; \bt) \int \frac{\partial}{\partial \bt} q_{T\mid \Z}(t, \z ; \bt) d  Q_\Z(\z)}{\{\int q_{T \mid \Z}(t, \z ; \bt) d  Q_\Z(\z)\}^2}\right] d  P_T(t)}{\int  \frac{I(t>x) q_{T \mid \Z}(t, \z ; \bt)}{\int q_{T \mid \Z}(t, \z ; \bt) d Q_\Z(\z)} d P_T(t)}.\\
\ese
Among these, the terms $\psi_{2n}$ and $\psi_{3n}$ involve the
nonparametric estimators $\wh P_T(t)$ and $\wh Q_\Z(\z)$. We first
present two lemmas that simplify the expressions of $\psi_{2n}$ and
$\psi_{3n}$, which allows us to apply the central limit theorem
later. These results are then combined in Theorem~\ref{Th:3.2}
to establish 
the asymptotic normality of $\wh \bt$. 

\begin{Lemma}
\label{lm:3.2}
	Under Conditions A2 and B2, presented in Supplement~{S3.2.2}, we have
	\bse
	&&n^{-1}\sum_{i=1}^n \psi_{2n}(X_i,\Delta_i,R_i;\bt_0)\\
	&=& n^{-1}\sum_{i=1}^n \left(\psi_{2}(X_i,\Delta_i,R_i;\bt_0)\right.\\&&\left. + \frac{1-R_i}{1-\pi} E_p \left[ \left. \Delta \frac{\frac{\partial}{\partial \bt}  q_{T\mid\Z}(X,\Z_i;\bt_0)q_T(X)- q_T^*(X) q_{T\mid\Z}(X,\Z_i;\bt_0) }{\{q_T(X)\}^2}\right| \Z_i \right]\right)\\
	&&+o_p(n_1^{-1/2}) +o_p(n_2^{-1/2}),
	\ese 
	where $q_T$ and $q_T^*$
are defined in Supplement~S3.2.1.
\end{Lemma}

\begin{Remark}
	In the result of Lemma~\ref{lm:3.2}, the first term comes from the score function $\psi(\cdot)$, while the second term arises from estimating $Q_\Z(\z)$ in $\psi_{2n}(\cdot)$ and therefore involves only data from population $\calQ$.
\end{Remark}

The third term $\psi_{3 n}$ presents the main difficulty, as several components involving $\wh Q_\Z(\z)$, which is estimated from population $\calQ$, appear inside the Kaplan-Meier integral based on $\wh P_T(t)$, which is estimated from population $\calP$. In other words, function classes indexed by the multivariate empirical distribution function $\wh Q_\Z(\z)$ enter the Kaplan-Meier integral based on $\wh P_T(t)$, complicating the analysis.

 To handle this, we apply the uniform properties for Kaplan-Meier integrals from \cite{sellero2005uniform}. However, one key assumption in that work is that the function class must be VC-subgraph. In our setting, the VC-subgraph assumption is difficult to verify since the structure of $q_{T\mid \Z}(t,\z;\bt)$ is unspecified and $\Z$ is multivariate. To address this challenge, we propose an alternative in Lemma~\ref{lm:c1} in Appendix C, which imposes conditions on the bracketing and covering numbers of the function class, making verification more tractable.

\begin{Lemma}
\label{lm:3.3}
	Under Conditions A1, B1, B2 and B3, presented in Supplement~{S3.2.2}, we have
	\bse
	&&n^{-1}\sum_{i=1}^n \psi_{3n}(X_i,\Z_i,\Delta_i,R_i;\bt_0)\\
	&=& n^{-1}\sum_{i=1}^n\left( \psi_{3}(X_i,\Z_i,\Delta_i,R_i;\bt_0)\right.\\
	&&\left. -\frac{R_i}{\pi}  E_p \left[ (1-\Delta) \frac{S_{1}(X,\Z, q_T)-S_{2}(X,\Z, q_T,q_T^*)}{\{S_0(X,\Z,q_T)\}^2}\eta_{0p}(X_i,\Delta_i;X,\Z)\right| X_i,\Delta_i  \right]\\
&& +\frac{1-R_i}{1-\pi}E_p\left.\left\{(1-\Delta)\frac{\eta_{1q}(\Z_i;X,\Z)-\eta_{2q}( \Z_i,\X,\Z)}{S_0(X,\Z,q_T)} \right | \Z_i \right\}
\\&&  - \left. \frac{1-R_i}{1-\pi}E_p\left.\left[ (1-\Delta)\frac{S_{1}(X,\Z, q_T)-S_{2}(X,\Z, q_T,q_T^*)}{\{S_0(X,\Z,q_T)\}^2}\eta_{0q}(\Z_i;X,\Z) \right| \Z_i \right] \right)\notag\\
&& - E_p\left[-(1-\Delta)\frac{S_{1}(X,\Z, q_T)-S_{2}(X,\Z, q_T,q_T^*)}{\{S_0(X,\Z,q_T)\}^2}\eta_{0p}(X,\Delta;X,\Z)\right]\\
	&&+ o_p(n_1^{-1/2}) +o_p(n_2^{-1/2}),
	\ese 
where $q_T$, $q_T^*$, $S_i$, $\eta_{ip}$ and $\eta_{iq}$ for $i=0,1,2$
are defined in Supplement~S3.2.1 and formulas~(3)-(11) in Supplement~S3.2.4 respectively.
\end{Lemma}

\begin{Remark}
	In the result of Lemma~\ref{lm:3.3}, the first term is contributed by the score function $\psi(\cdot)$. The term involving $R_i$ corresponds to the estimation of $P_T(t)$ and thus relies solely on data from population $\calP$, whereas the terms involving $1 - R_i$ correspond to the estimation of $Q_\Z(\z)$ and rely solely on data from population $\calQ$.
\end{Remark}

The proofs of Lemma~\ref{lm:3.2} and Lemma~\ref{lm:3.3} are presented in Supplement~S3.2.6 and Supplement~S3.2.7 respectively. By combining Lemmas~3.2 and 3.3, we obtain the asymptotic normality of $\wh \bt$ in the following theorem.

\begin{Theorem}[Asymptotic Normality]
\label{Th:3.2}
 Under Conditions~A1--A6, B1--B3, presented in Supplement~{S3.2.2}, we have that
    $$\sqrt{n_0}(\wh \bt - \bt_0) \dovr N(\0,\bf \Sigma),$$
    as $n_0 \rightarrow \infty$, where 
\bse
\bf \Sigma &\equiv & E_p\{\tfrac{\partial}{\partial \bt}\psi(X,\Z,\Delta;\bt_0)\}^{-1}{\bf \Sigma}_{\psi}E_p\{\tfrac{\partial}{\partial \bt} \psi(X,\Z,\Delta;\bt_0)\}^{-\rm T},\\
\bf \Sigma_{\psi} &\equiv & \min(\tfrac{1}{\pi},\tfrac{1}{1-\pi})[(1-\pi)\var_p\{\psi(X,\Z,\Delta;\bt_0)+\psi_{p_T}(X,\Delta)\}+\pi \var_q\{\psi_{q_\Z}(\Z)\}],\\
\psi(x,\z,\delta;\bt_0) &\equiv & \delta \frac{\frac{\partial}{\partial \bt} q_{T \mid \Z}(x, \z ; \bt_0)}{q_{T \mid \Z}(x, \z ; \bt_0)} - \delta \frac{q_T^*(x)}{q_T(x)} + 
(1-\delta) \frac{S_1(x,\z,q_T)-S_2(x,\z,q_T,q^{*}_T)}{S_0(x,\z,q_T)},\\ \psi_{p_T}(x,\delta)&\equiv& 
\left. -E_p \left[ (1-\Delta) \frac{S_{1}(X,\Z, q_T)-S_{2}(X,\Z, q_T,q_T^*)}{\{S_0(X,\Z,q_T)\}^2}\eta_{0p}(x,\delta;X,\Z)\right| x,\delta  \right],\\
\psi_{q_\Z}(\z)&\equiv&  - E_p\left[ \left. \Delta \frac{\frac{\partial}{\partial \bt}  q_{T\mid\Z}(X,\z;\bt_0)q_T(X)- q_T^*(X) q_{T\mid\Z}(X,\z;\bt_0) }{\{q_T(X)\}^2}\right| \z \right]  \notag \\ &&  + E_p\left.\left\{(1-\Delta)\frac{\eta_{1q}(\z;X,\Z)-\eta_{2q}(\z,\X,\Z)}{S_0(X,\Z,q_T)} \right | \z \right\} \notag \\&&  - E_p\left.\left[ (1-\Delta)\frac{S_{1}(X,\Z, q_T)-S_{2}(X,\Z, q_T,q_T^*)}{\{S_0(X,\Z,q_T)\}^2}\eta_{0q}(\z;X,\Z) \right| \z \right],
\ese
and $q_T$, $q_T^*$ $S_i$, $\eta_{ip}$ and $\eta_{iq}$ for $i=0,1,2$
are defined in Supplement~S3.2.1 and Supplement~S3.2.4, formulas~(3)-(11) respectively.
\end{Theorem}

\begin{Remark}
The expression of $\bf \Sigma_{\psi}$ involves three components: $\psi$,
$\psi_{p_T}$ and $\psi_{q_\Z}$, where $\psi$ is the score function, $\psi_{p_T}$ is related to the
estimation of $P_T(\cdot)$, and $\psi_{q_\Z}$
arises from the estimation of $Q_\Z(\z)$. Hence, if the
true functions $Q_\Z(\z)$ and $P_T(t)$ are known, the terms $\psi_{p_T}$
and $\psi_{q_\Z}$ would vanish, leaving only the contribution from $\psi$. In that case, ${\bf\Sigma}=E_p\{\tfrac{\partial}{\partial \bt} \psi(X,\Z,\Delta;\bt_0)\}^{-1}$. Furthermore, since we allow the sample sizes from populations $\calP$ and $\calQ$ to grow at different rates, the limit proportion $\pi$ can take any value in $[0,1]$, including the boundary cases $\pi = 0$ or $\pi = 1$. Accordingly, we summarize the expression of $\bf \Sigma_{\psi}$ for specific values of $\pi$ below:
	$$
	{\bf \Sigma_{\psi}}=
	\begin{cases}
	\var_p\{\psi(X,\Z,\Delta;\bt_0)+\psi_{p_T}(X,\Delta)\} & \text{if } \pi = 0, \\
		\var_p\{\psi(X,\Z,\Delta;\bt_0)+\psi_{p_T}(X,\Delta)\} + \frac{\pi}{1-\pi} \var_q\{\psi_{q_\Z}(\Z)\} & \text{if } 0 < \pi < \frac{1}{2}, \\
		\var_p\{\psi(X,\Z,\Delta;\bt_0)+\psi_{p_T}(X,\Delta)\} + \var_q\{\psi_{q_\Z}(\Z)\} & \text{if } \pi = \frac{1}{2}, \\
		\frac{1-\pi}{\pi} \var_p\{\psi(X,\Z,\Delta;\bt_0)+\psi_{p_T}(X,\Delta)\} + \var_q\{\psi_{q_\Z}(\Z)\} & \text{if } \frac{1}{2} < \pi < 1, \\
		\var_q\{\psi_{q_\Z}(\Z)\} & \text{if } \pi = 1.
	\end{cases}
	$$
\end{Remark}

The proofs of Lemmas~\ref{lm:3.2} and~\ref{lm:3.3}, as well as Theorems~\ref{Th:3.1} and \ref{Th:3.2} are lengthy and are technically challenging and innovative. A concise summary of the proof is presented in the Appendix B, while the complete technical details can be found in Supplement~S3.2. The
technical conditions are also listed in Supplement~S3.2 and can be
classified into three main groups. The first key assumption (Condition
A1) is that the support of $q_T(t)$ in the target population $\calQ$
is contained within the support of $p_T(t)$ in the source population
$\calP$. This assumption is common in the label shift
literature~\citep{lee2024doubly,tian2023elsa}, since leveraging
information from the source requires that the source distribution
adequately covers the target distribution.
The second group of assumptions involves the smoothness and boundedness
of $q_{T\mid\Z}(t,\z;\bt)$, $\frac{\partial}{\partial
  \bt}q_{T\mid\Z}(t,\z;\bt_0)$, and related functions (Conditions
A2--A5, B1--B2, and B4), which are mild. The third group consists
  of assumptions related to the Kaplan-Meier integrals (Conditions A6
  and B3). These conditions are established based on standard
  assumptions introduced by \cite{sellero2005uniform}. However,
  \cite{sellero2005uniform} assumes the class of functions within the
  Kaplan-Meier integrals to be VC-subgraph. We adapt these
  conditions further by imposing constraints on the covering and
  bracketing numbers defined in empirical process theory, facilitating
  easier verification. More details can be found in  Lemma~\ref{lm:c1} in Appendix C. 
Additionally, we assume that the support of $T$ is compact. Although
this assumption may appear restrictive in the context of survival
analysis, it primarily serves for technical purposes. Our simulations
demonstrate that our method performs well even when $T$ is unbounded,
suggesting this assumption may not be strictly necessary in practice.

In Theorem~\ref{Th:3.2}, we prove that our estimator is
$\sqrt{n_0}$-consistent, where $n_0 \equiv \min \left(n_1,
  n_2\right)$. Thus, the convergence rate is determined by the smaller sample
size of $n_1$ and $n_2$. Intuitively, this is because the estimator is constructed 
using data from both populations. While a larger dataset provides more
information, it cannot improve the convergence rate if the other
dataset remains limited, as both contribute for estimation. This
result is consistent with \cite{lee2024doubly}, which allows $n_2$ to
grow faster than $n_1$ and proves that the estimator remains
$\sqrt{n_1}$-consistent.

One of the main contributions of this paper is the development of a
theoretical framework for the asymptotic properties of survival
analysis under label shift. We provide an explicit formula for the
asymptotic variance, which allows direct and easy estimation without
relying on time-consuming bootstrap procedures. The corresponding
estimation procedures are 
provided in Supplement~S3.2.10. 

The estimation of the unknown parameter $\bt$ also enables the
estimation of functionals of the conditional survival function in the
target population $\calQ$ as a by-product. We use $\zeta(\z) \equiv
E_q\{g(T)\mid \z\}$ to represent a functional of the conditional
survival function, where $g(t)$ is a given function of $t$. For
example, when $g(t)=t$, this corresponds to the conditional mean
survival time in population $\calQ$. The estimator of $\zeta(\z)$ is
obtained by plugging in the estimator $\wh \bt$, that is, $\wh
\zeta(\z)=\int g(t) q_{T\mid \Z}(t,\z;\wh \bt)dt$. The following
Corollary~\ref{col:3.1} establishes the asymptotic properties of $\wh
\zeta(\z)$. 

\begin{Corollary}
\label{col:3.1}
Under Conditions $A1-A6$ and $B1-B3$, for any given $\Z=\z$, as $n_0\rightarrow \infty$, we have $\wh \zeta(\z) \povr \zeta(\z)$ and $\sqrt{n_0}\{\wh \zeta(\z) -  \zeta(\z)\} \dovr N(0,\bf \Sigma_{\zeta(\z)}),$
where
$\Gamma_{\zeta(\z)} \equiv \tfrac{\partial}{\partial \bt}\int g(t) q_{T\mid \Z}(t,\z; \bt_0)dt$, and
${\bf \Sigma}_{\zeta(\z)} \equiv \Gamma_{\zeta(\z)} \Sigma \Gamma_{\zeta(\z)}\trans.$

\end{Corollary}
The proof of Corollary~\ref{col:3.1} is presented in Supplement~S3.2.9.

\section{Simulations}

First, we describe the general data generation process used in our
simulations. We begin by specifying four key components of the
setting: $q_{T\mid\Z}(t,\z;\bt)$, $p_T(t)$, $p_C(c)$, and
$q_\Z(\z)$. These fully determine the conditional distribution
$f_{\Z\mid T}(\z,t)$.  
In population $\calP$, we generate random values for $T$ and $C$
independently based on $p_T(t)$ and $p_C(c)$.  Next, we generate the
covariate $\Z$ in $\calP$ from $f_{\Z\mid T}(\z,t)$.
  This completes the data generation for population $\calP$.
For population $\calQ$, the covariate $\Z$ are drawn directly from
the specified marginal distribution $q_\Z(\z)$. With these steps, we
generate the full dataset for both populations.
  
Specifically, in all the simulations, in population $\calP$, we set $p_T(t)$ to be an
$\text{Exp}(1)$ density and $p_C(c)$ to be an
$\text{Exp}(2.5)$ density. In population $\calQ$, we consider two cases for
$q_\Z(\z)$. The first case is a bivariate normal distribution with mean vector $(0,1)\trans$ and identity covariance matrix $\I$. In the second
case, we further added two independent covariates, 
one has an exponential distribution $\text{Exp}(1)$ and the other has a binomial distribution $B(1,0.5)$.
We considered four survival models for $q_{T\mid\Z}(t,\z)$: the
proportional hazards model, the accelerated failure time model, the
proportional odds model, and the accelerated hazards
model. Specifically, 
\begin{enumerate}[label=(\alph*)]
\item Proportional Hazards Model: 
  $$q_{T \mid \Z}(t,\z;\bt) \equiv h(t;\gamma,\lambda)\exp(\z\trans\bb)\exp\left\{-H(t;\gamma,\lambda) \exp(\z\trans\bb)\right\},$$
  where $\bt\equiv (\bb\trans,\gamma,\lambda)\trans$, and $h(t;\gamma,\lambda)$ is the Weibull baseline hazard function, specified as
  $h(t;\gamma,\lambda) = \lambda \gamma t^{\gamma-1}$. The
  results under this model are presented in
  Table~\ref{Tab1} and Supplementary Table~S4.

\item Proportional Odds Model:
  $$q_{T \mid \Z}(t,\z;\bt) \equiv \frac{h(t;\mu,\sigma)\exp(\z\trans
    \bb)}{\{1+H(t;\mu,\sigma)\exp(\z\trans \bb)\}^2},$$ 
where $\bt\equiv (\bb\trans,\mu,\sigma)\trans$,  $h(t;\mu,\sigma) \equiv \exp(-\mu/\sigma)/\sigma
t^{1/\sigma-1} $ and $H(t;\mu,\sigma)\equiv \exp(-\mu/\sigma)
t^{1/\sigma}$. The results are shown in Supplementary Table~S1. 

\item Accelerated Failure Time Model: $$
    q_{T \mid \Z}(t,\z;\bt) = \frac{1}{\sigma t}f_\varepsilon
    \left(\frac{\log t - \mu - \z\trans \bb}{\sigma}\right),$$ where $\bt\equiv (\bb\trans,\mu,\sigma)\trans$, and
    $f_\varepsilon$ is the density of a standard normal
    distribution. The results are presented in Supplementary Table~S2.

\item Accelerated Hazards Model: $$
q_{T \mid \Z}(t,\z;\bt) \equiv h\{t \exp (\z\trans \bb);\lambda,\gamma\} \exp \left[-\int_0^t h\{s\exp (\z\trans \bb);\lambda,\gamma\} d s\right],
$$
where $\bt\equiv (\bb\trans,\lambda,\gamma)\trans$, and $h(t;\gamma,\lambda)=\gamma\lambda t^{\gamma-1}$. The results are reported in Supplementary Table~S3.
\end{enumerate}

In all the model settings, we generated 500 datasets and implemented our proposed
estimator. The simulation results, including empirical mean squared
error (MSE), empirical standard error (SE), mean of the estimated
standard error ($\widehat{\text{SE}}$), and the empirical
  coverage probability (CP) 
of the estimated 95\% confidence intervals, are provided in Table
\ref{Tab1} and the 
Supplementary Tables~S1--S4. We varied the sample
sizes $n_1$ and $n_2$ for data from populations $\calP$ and $\calQ$
and the values we chose are presented in the tables as well.

Overall, our estimator performs consistently well across all
  simulation scenarios, demonstrating low MSE, SE,
  $\widehat{\text{SE}}$, and bias. The empirical results support the
  consistency of our estimator, as MSE, SE, and bias decrease with
  increasing sample size. Moreover, the estimated standard errors, derived from the methods introduced in Supplement~S3.2.10,
  closely align with the empirical standard errors, especially
    for large sample sizes.
  Furthermore, in most cases, the
  empirical coverage probabilities of the confidence intervals are
  close to the nominal 95\% level, providing empirical
  support for our theoretical results. In the few cases where
    the coverage rate slightly deviates from the nominal value, these
  differences decrease as the sample size increases. Compared to the bivariate covariate case,
    the results under a four-dimensional covariate including a
    categorical component show slightly more deviation
    from 95\% in terms of empirical
    coverage rate, if everything else is held the same. This is not a
    surprise since the higher dimension increases the numerical
    complexity. Furthermore,
these deviations are still modest and tend to decrease when sample
sizes increase. When one of $n_1$ and $n_2$ is fixed, increasing the
other improves the performance. Moreover, increasing both $n_1$ and
$n_2$ simultaneously leads to a more substantial improvement in the
results.  

When comparing results across different survival models, our
  method consistently exhibits good performance, maintaining low
  MSE, bias, and standard errors. This suggests that our method is
  sufficiently flexible to accommodate various survival models. Among
  these, the proportional hazards model with a Weibull baseline hazard
  function achieves the best overall performance. 

\begin{table}[h]
\centering
\renewcommand\arraystretch{1.2}
\setlength{\tabcolsep}{12pt}
\caption{Simulation results for the proportional hazards model with Weibull baseline hazard function and with $\beta_1=1$, $\beta_2=1$, $\lambda=1$ and $\gamma=1.5$. }
\label{Tab1}
\begin{tabular}{cccccccc}
  \hline \hline
$n_1$ & $n_2$ & & MSE & Bias & SE & $\wh{\text{SE}}$ & CP \\ 
  \hline
250 & 500 & $\beta_1$ & 0.0105 & 0.0070 & 0.1025 & 0.1035 & 0.9517 \\ 
   &  & $\beta_2$ & 0.0094 & 0.0069 & 0.0971 & 0.1033 & 0.9618 \\ 
   &  & $\lambda$ & 0.0168 & 0.0263 & 0.1269 & 0.1288 & 0.9557 \\ 
   &  & $\gamma$ & 0.0184 & 0.0267 & 0.1330 & 0.1318 & 0.9356 \\ \hline
  500 & 250 & $\beta_1$ & 0.0060 & 0.0030 & 0.0772 & 0.0786 & 0.9504 \\ 
   &  & $\beta_2$ & 0.0061 & 0.0077 & 0.0780 & 0.0794 & 0.9566 \\ 
   &  & $\lambda$ & 0.0106 & 0.0192 & 0.1014 & 0.0991 & 0.9442 \\ 
   &  & $\gamma$ & 0.0111 & 0.0090 & 0.1050 & 0.1098 & 0.9731 \\ \hline
  500 & 500 & $\beta_1$ & 0.0047 & -0.0009 & 0.0689 & 0.0737 & 0.9558 \\ 
   &  & $\beta_2$ & 0.0055 & 0.0029 & 0.0740 & 0.0741 & 0.9418 \\ 
   &  & $\lambda$ & 0.0087 & 0.0176 & 0.0918 & 0.0928 & 0.9518 \\ 
   &  & $\gamma$ & 0.0109 & 0.0184 & 0.1027 & 0.0986 & 0.9478 \\ \hline
  500 & 750 & $\beta_1$ & 0.0050 & 0.0041 & 0.0705 & 0.0726 & 0.9560 \\ 
   &  & $\beta_2$ & 0.0052 & 0.0055 & 0.0720 & 0.0724 & 0.9500 \\ 
   &  & $\lambda$ & 0.0083 & 0.0169 & 0.0895 & 0.0897 & 0.9480 \\ 
   &  & $\gamma$ & 0.0102 & 0.0140 & 0.1000 & 0.0945 & 0.9340 \\ \hline
  750 & 500 & $\beta_1$ & 0.0039 & 0.0031 & 0.0626 & 0.0616 & 0.9478 \\ 
   &  & $\beta_2$ & 0.0040 & 0.0006 & 0.0635 & 0.0618 & 0.9398 \\ 
   &  & $\lambda$ & 0.0060 & 0.0187 & 0.0752 & 0.0777 & 0.9498 \\ 
   &  & $\gamma$ & 0.0073 & 0.0159 & 0.0841 & 0.0852 & 0.9618 \\ \hline
  500 & 1000 & $\beta_1$ & 0.0047 & 0.0007 & 0.0685 & 0.0717 & 0.9660 \\ 
   &  & $\beta_2$ & 0.0050 & 0.0034 & 0.0707 & 0.0717 & 0.9440 \\ 
   &  & $\lambda$ & 0.0082 & 0.0169 & 0.0890 & 0.0884 & 0.9560 \\ 
   &  & $\gamma$ & 0.0089 & 0.0164 & 0.0931 & 0.0921 & 0.9520 \\ \hline
  1000 & 500 & $\beta_1$ & 0.0033 & -0.0022 & 0.0571 & 0.0552 & 0.9336 \\ 
   &  & $\beta_2$ & 0.0027 & -0.0045 & 0.0518 & 0.0551 & 0.9577 \\ 
   &  & $\lambda$ & 0.0051 & 0.0101 & 0.0706 & 0.0691 & 0.9497 \\ 
   &  & $\gamma$ & 0.0060 & 0.0195 & 0.0750 & 0.0783 & 0.9618 \\ \hline
  1000 & 1000 & $\beta_1$ & 0.0026 & -0.0007 & 0.0506 & 0.0515 & 0.9540 \\ 
   &  & $\beta_2$ & 0.0028 & 0.0000 & 0.0527 & 0.0515 & 0.9520 \\ 
   &  & $\lambda$ & 0.0040 & 0.0127 & 0.0621 & 0.0641 & 0.9520 \\ 
   &  & $\gamma$ & 0.0049 & 0.0102 & 0.0694 & 0.0690 & 0.9620 \\ 
   \hline
\hline
\end{tabular}     
\begin{tablenotes}
    \item $n_1$: sample size in $\calP$; $n_2$: sample size in
      $\calQ$; MSE: empirical mean squared error; SE: empirical standard error; 
       $\wh{\text{SE}}$: mean of estimated standard errors; CP:
      empirical coverage probability of 95\% confidence intervals.
\end{tablenotes}
\end{table}

\section{Data Application}

We now apply our method to a
  liver transplantation dataset obtained from the Organ Procurement
  and Transplantation Network (http://optn.transplant.hrsa.gov). We include adult white patients
  who underwent only one transplantation procedure between 2010 and
  2013 in our analysis. We set the event time $T$ to be the
  lifespan of the transplanted 
  liver, measured from the transplantation date until organ failure or
  patient death. The covariate $\Z$ considered in the analysis
  include the MELD score recorded when the patient entered the organ
  waiting list~(INIT\_MELD), the medical condition~(MED) prior to
  transplantation (categorized as not hospitalized, hospitalized, or
  ICU), and the length of hospital stay following
  transplantation~(LOS). Moreover, the original dataset contains two
  variables related to age: the age of the donor and the age of the
  recipient. To better reflect the medical importance of the ages, we
  transform these variables
  into two new features: the sum of their ages (SUM\_AGE) and the age
  difference (DIFF\_AGE), where ``DIFF\_AGE'' is defined as the donor's
  age minus the recipient's age. This transformation allows us to
  better understand how both the ages of donors and recipients and
  their age difference influence survival time.

We define population $\calP$ as patients whose primary payment
  method is private insurance ($R=1$, $n_1=4509$), while population
  $\calQ$ consists of patients whose primary payment method was
  Medicaid or Medicare ($R=0$, $n_2=3131$). Event times in both
  populations are subject to right censoring due to patient dropout or
  loss of contact, with a censoring rate of 32.6\% in population
  $\calP$. 
In the original dataset, the response in population $\calQ$ is
  available. Thus, before applying our method, we conduct a
  preliminary analysis using the complete dataset. 
  
  We plot the
  survival curves based on the Kaplan-Meier estimator for populations $\calP$ and $\calQ$, as shown in Figure~\ref{fig:liver}. The separation of the survival
  curves suggests the possibility of label shift in the dataset. Specifically, patients covered by Medicaid or Medicare exhibit shorter survival times for their transplanted organs. This observation further justifies our assumption that the support of $q_T(t)$ in the target population $\calQ$ is included in that of $p_T(t)$ in population $\calP$.
To formally assess the label shift assumption $q_{\Z\mid T}(\z,t)
=p_{\Z\mid T}(\z,t)$, we perform a hypothesis test. However, because
the event times in both populations are censored, we 
cannot directly convert the problem into a test of conditional independence between the variable $R$ (indicating population membership) and the vector of covariate $\Z$ given $T$, and we are unable to
use existing methods. Hence, we construct a new bootstrap based method to test the label shift assumption.
Details of this test statistic are provided in Supplement~S2. The test yields a $p$-value of 0.1040, indicating that we do not reject the null hypothesis, supporting the validity of the label shift assumption in this
dataset.
\begin{figure}[h]
    \centering
    \includegraphics[width=14cm]{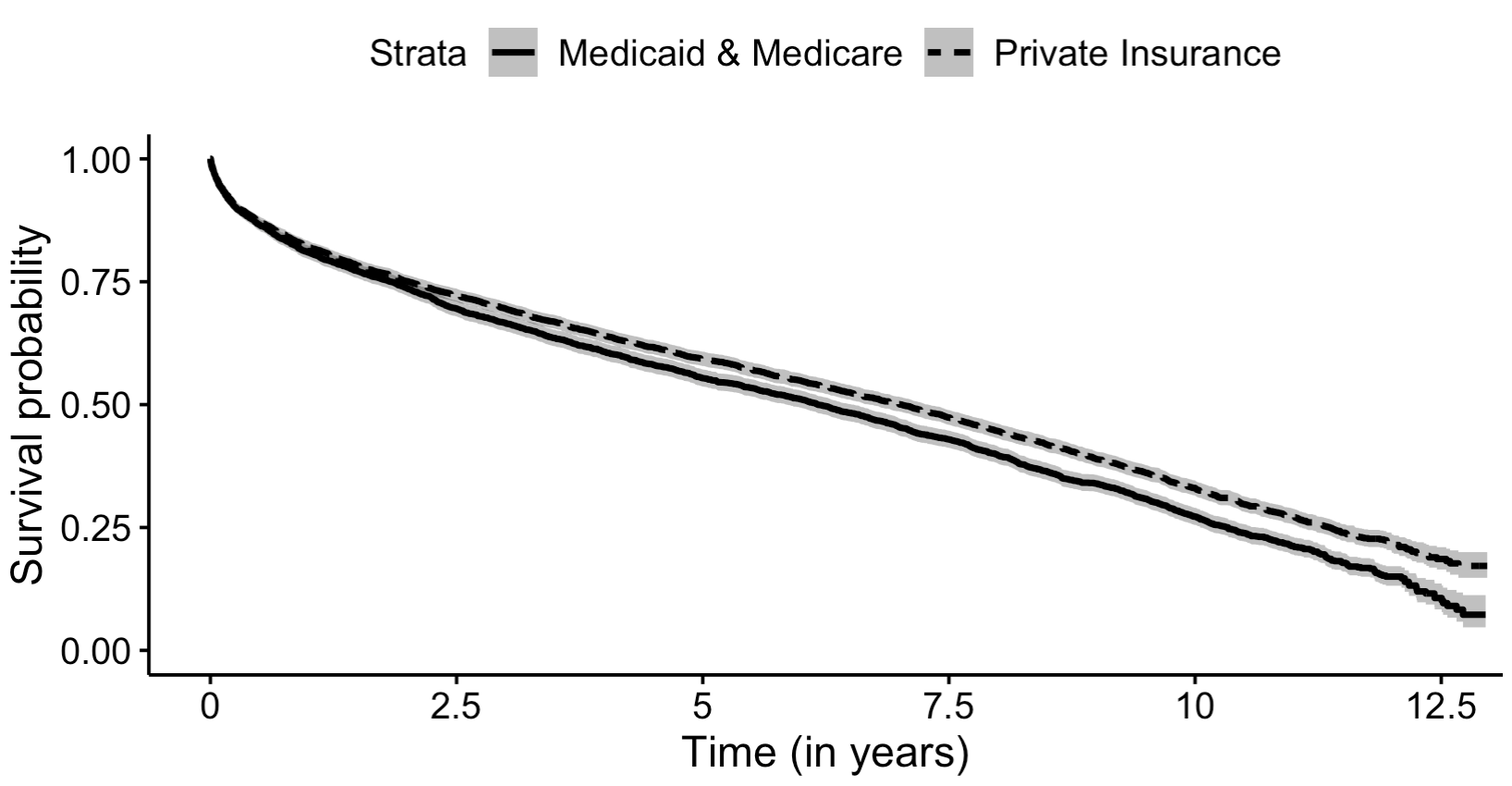}
    \caption{Estimated survival curves for transplanted livers stratified by patients' insurance types.}
    \label{fig:liver}
\end{figure}

We then mask the observed survival times from population
  $\calQ$. To identify a suitable model for
  $q_{T\mid\Z}(t,\z;\bt)$, we fit several candidate models and use the
  Bayesian Information Criterion (BIC) to select the best one. In the implementation of BIC, although our likelihood function involves two nonparametric components, our goal is only to
  select a suitable $q_{T\mid\Z}(t,\z;\bt)$, which is fully
  parametric. In addition, the nonparametric components
  $p_T(\cdot)$ and $q_\Z(\cdot)$ enter the likelihood in the
  same pattern under different $q_{T\mid\Z}(t,\z;\bt)$ models,
hence the model  complexity or flexibility contributed by $p_T(\cdot)$ and
$q_\Z(\cdot)$ remains the same. 
Therefore, we use the relative values of
$-2\ell_n(\bt)+\log(n)d_{\bt}$ under different $q_{T\mid\Z}(t,\z;\bt)$
models 
to perform model selection, where
$\ell_n(\bt)$ is given by (\ref{eq:llh}) and $d_{\bt}$ denotes the
dimension of $\bt$. Specifically, 
we randomly split the data from populations $\calP$ and $\calQ$
separately, use 20\% of the data to compute the criteria and select
the smallest value as our chosen model, and use the remaining 80\% for
estimation and inference.  
The models we considered and their corresponding criterion values are
presented in Table~\ref{BIC}. The accelerated failure time model with an exponential baseline hazard function yields the
lowest criterion value. Therefore, we perform further analysis
under this model. The results are
presented in Table~\ref{RealData}. In the table, the rows labeled
``Shift'' display the estimates obtained from our method considering
the label shift, along with their estimated standard errors and
confidence intervals. The row labeled ``$\calP$'' contains estimates
derived from fitting the proportional hazards model with a Weibull
baseline hazard function only to the complete source population data. This implies that we assume the same survival models apply to
  both $\calP$ and $\calQ$ and we ignore the possibility of a label shift. The row labeled ``$\calQ$ (Oracle)'' presents results from the complete
target population data. Specifically, we fit the selected parametric model to the full $\calQ$ dataset, including censored outcomes. These results are used as an oracle benchmark for evaluation and represent an ideal reference that is not attainable in real applications. Additionally, we report the differences between the
estimates and the oracle. We also present in Figure~\ref{fig:CI} the
estimates and confidence intervals for each parameter with different
methods, where the estimates are marked by a star.

\begin{table}[h]
\centering
\caption{Model comparison results for the transplanted liver dataset}
\label{BIC}
\renewcommand\arraystretch{1.3}
\begin{tabular}{lrr}
\hline
Model &  $d_\bt$ & BIC \\
\hline
Proportional hazards model with a Weibull baseline hazard function & 8 & 12457 \\ 
Accelerated failure time model with a Log-Normal baseline hazard function & 8 & 12729 \\ 
Proportional odds model with a Log-Logistic baseline survival function & 8 & 12577 \\ 
Accelerated failure time model with an Exponential baseline hazard function & 7 & 12433 \\ 
Accelerated hazards model with a Weibull baseline hazard function & 8 & 13241 \\
\hline
\end{tabular}
\end{table}

The estimates obtained using our method closely align with those
derived from the complete $\calQ$ data. Compared to results based
solely on the $\calP$ data, our method significantly reduces the
discrepancy from the oracle. However, when examining the estimated
standard errors, our method yields higher standard errors than both
the oracle and the result based on the $\calP$ data, which is expected due
to the absence of observed responses in population $\calQ$, leading to
greater variation in estimation.

In Figure~\ref{fig:CI}, we observe that the confidence intervals from
our method either encompass or highly overlap with those from the
oracle results, demonstrating its effectiveness. Our method also
outperforms the estimates solely based on the $\calP$ data. Notably, for
the variable ``DIFF\_AGE'', the oracle results indicate
non-significance, as its confidence interval includes zero. Our method
correctly captures this, with its confidence interval also containing
zero. In contrast, the results from $\calP$ suggest the variable
``DIFF\_AGE'' is significant, potentially leading to a misleading
conclusion. 

Moreover, for the variable ``INIT\_MELD", the oracle results suggest
that the coefficient is significantly greater than zero, whereas the
results based solely on the $\calP$ data suggest non-significance. Our
method correctly identifies the significance, consistent with the
oracle results, as the estimated confidence interval from ``Shift'' does
not include zero. These findings highlight the risk of ignoring label
shift, as 
directly applying inference results from the source population $\calP$
to the target population $\calQ$ without accounting for distribution
shift can lead to incorrect conclusions. 

Notably, the estimator for the parameter $\lambda$  in the baseline
hazard function exhibits a high standard error, although the
confidence interval still encompasses that from the oracle result. One
possible explanation for this relatively large standard error is the
complete absence of response data in population $\calQ$, making it
challenging to accurately capture the underlying baseline hazard function using
only the label-shifted and label-censored source data.

\begin{table}[h]
\centering
\renewcommand\arraystretch{1.5}
\setlength{\tabcolsep}{8pt}
\caption{Results for the transplanted liver dataset}
\label{RealData}
   \begin{tabular}{l|lrrrr}
        \hline
        \hline
        Var & Type      & Est.     & $\wh{\text{SE}}$      & CI                  & Diff.   \\
        \hline
        $\beta_1$ (SUM\_AGE)  & $\calQ$ (Oracle)     & -0.1173 & 0.0248 & [-0.1659, -0.0687] & -      \\
                  & Shift & -0.1190 & 0.0406 & [-0.1986, -0.0394] & -0.0017 \\
                  & $\calP$     & -0.1369 & 0.0230 & [-0.1820, -0.0918] & -0.0196 \\ \hline
        $\beta_2$ (DIFF\_AGE) & $\calQ$ (Oracle)     &  0.0320 & 0.0240 & [-0.0150,  0.0790] & -      \\
                  & Shift &  0.0378 & 0.0382 & [-0.0371,  0.1127] & 0.0058 \\
                  & $\calP$     &  0.0609 & 0.0226 & [0.0166,  0.1052] & 0.0289 \\ \hline
        $\beta_3$ (INIT\_MELD) & $\calQ$ (Oracle)     &  0.0706 & 0.0245 & [0.0226,  0.1186] & -      \\
                   & Shift &  0.0873 & 0.0369 & [0.0150,  0.1596] & 0.0167 \\
                   & $\calP$     &  0.0435 & 0.0224 & [-0.0004,  0.0874] & -0.0271 \\ \hline
        $\beta_4$ (LOS)    & $\calQ$ (Oracle)     & -0.2277 & 0.0203 & [-0.2675, -0.1879] & -      \\
                  & Shift & -0.2107 & 0.0249 & [-0.2595, -0.1619] & 0.0170 \\
                  & $\calP$     & -0.2255 & 0.0175 & [-0.2598, -0.1912] & 0.0022 \\ \hline
        $\beta_5$ (MED: hospital no ICU) & $\calQ$ (Oracle)     &  0.1978 & 0.0758 & [ 0.0492,  0.3464] & -      \\
                  & Shift &  0.2284 & 0.0902 & [0.0516,  0.4052] & 0.0306 \\
                  & $\calP$     &  0.1867 & 0.0643 & [0.0607,  0.3127] & -0.0111 \\ \hline
        $\beta_6$ (MED: no hospital)   & $\calQ$ (Oracle)     &  0.2726 & 0.0735 & [0.1285,  0.4167] & -      \\
                  & Shift &  0.3183 & 0.0946 & [0.1329,  0.5037] & 0.0457 \\
                  & $\calP$     &  0.1999 & 0.0627 & [0.0770,  0.3228] & -0.0727 \\ \hline
        $\lambda$    & $\calQ$ (Oracle)     &  1.7996 & 0.0645 & [1.6732,  1.9260] & -      \\
                  & Shift &  1.7120 & 0.1086 & [1.4992,  1.9249] & -0.0876 \\
                  & $\calP$     &  2.0020 & 0.0541 & [1.8960,  2.1080] & 0.2024 \\
        \hline
        \hline
    \end{tabular}
    \begin{tablenotes}
\small
    \item Est.: estimates; $\wh{\text{SE}}$: estimated standard error; CI: 95\% confidence interval; Diff: the difference with oracle
\end{tablenotes}
\end{table}

\begin{figure}[h!]
    \centering
    \includegraphics[width=14cm]{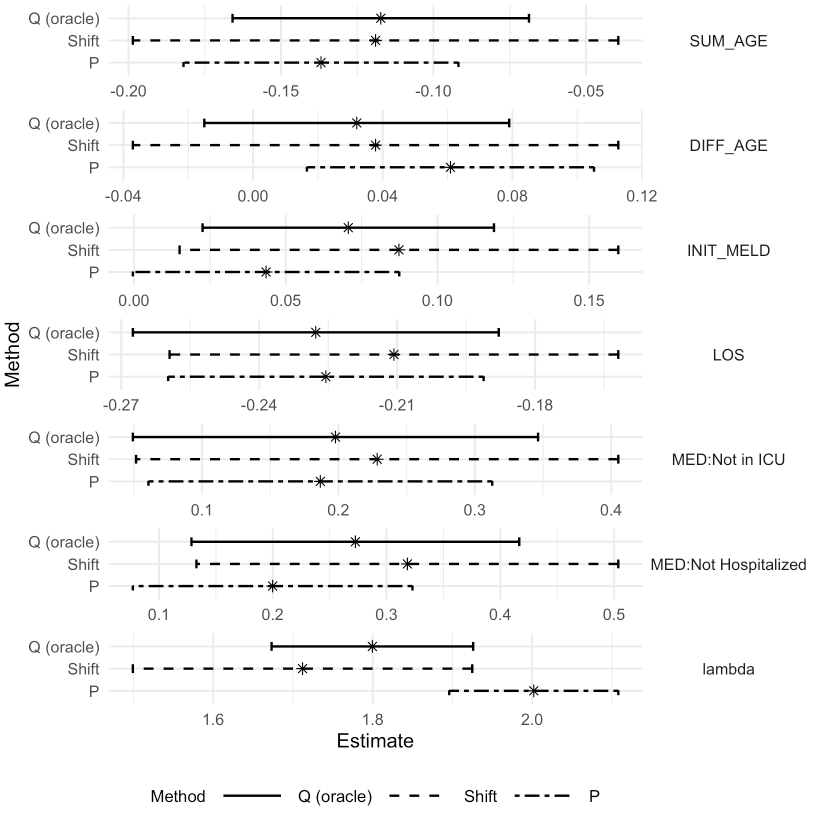}
    \caption{Confidence intervals for the variables in the transplanted liver dataset using various estimation methods. Estimated coefficients are indicated by stars.}
    \label{fig:CI}
\end{figure}

\section{Discussion}

In this paper, we propose a domain adaptation method for survival
analysis under the label shift assumption, which, to the best of our
knowledge, is the first consideration and first solution in this
setting. Our approach
estimates the covariate effects in the target population, which
  is the main population of research interest, but in which no
  event times are observed.  We use data from a source population in
  which censored survival times are observed, and which has a label
  shift relation to the target population, i.e., it shares a common
  covariate distribution given the survival time, to help achieve the analysis in the target population.
  We establish
the asymptotic consistency and normality of the proposed estimator and
validate its performance through extensive simulation studies and a
real data application. Because of the difficult nature of the problem,
caused by the complete absence of response information in the target
population and the weak assumption of a common but unspecified
covariate conditional distribution, somewhat stringent assumptions are
required from the perspective of classical survival
analysis.  Nonetheless, the proposed framework may serve as a
foundation for further methodological developments in such
settings. Currently, we are working on extending the parametric model
applied to the target population to a semiparametric model and on
developing weaker independence assumptions. While we assume that the event time and censoring time are independent, we plan to extend this to allow for conditional independence given covariates. Similarly, incorporating more flexible survival models, such as competing risks and dependent censoring, is also of interest. Additionally, our current
work focuses on scenarios involving only a 
single source population. Sometimes,  multiple source populations are 
available, presenting an interesting direction
for further exploration. Furthermore, while our method assumes that the
support of event times in the target population is included in
that of the source population, real-world applications often encounter
observations outside this support. Developing methods capable of 
handling these out-of-support observations would be a valuable
advancement, although we expect that stronger assumptions will be
needed to enable such an expansion.  In conclusion, we hope this work
will serve as a useful starting point for future research on survival
analysis under label shift and help draw greater attention to this
underexplored but important setting. 

\begin{appendix}
\section{Proof of Lemma 3.1}

\begin{proof}
 We show that when the ratio $q_{T\mid \Z}(t,\z;\bt)/q_{T\mid
  \Z}(t,\z;\wt \bt)$ depends on $\z$ whenever $\bt \neq \wt \bt$, then
all the unknown components in (2) can be
identified. It
suffices to establish the identifiability of $f_{\Z\mid T}(t,\z)$,
$p_T(t)$, $p_C(c)$, and $q_T(t)$. Given that  we have i.i.d. observations
      $(X_i, \Delta_i, \Z_i)$ from
    $\calP$, due to the independence
    between $(T,\Z)$ and $C$, $p_C(c)$,
    $p_{T\mid \Z}(t,\z)$ and $p_\Z(\z)$ are all
identifiable. Hence, $f_{\Z \mid T}(\z,t)$ and $p_T(t)$ can be identified. In population $\calQ$, we
observe a random sample of $\Z$, enabling the identifiability of
$q_\Z(\z)$. Suppose we have two different $\bt$ and $\wt \bt$ such that $$f_{\Z \mid T}(\z,t)=\frac{q_{T\mid \Z}(t,\z;\bt) q_{\Z}(\z)}{\int q_{T\mid \Z}(t,\z;\bt) q_{\Z}(\z) d\z} = \frac{q_{T\mid \Z}(t,\z;\wt \bt) q_{\Z}(\z)}{\int q_{T\mid \Z}(t,\z;\wt \bt) q_{\Z}(\z) d\z},$$
which further implies that the ratio $q_{T\mid \Z}(t,\z;\bt)/q_{T\mid
  \Z}(t,\z;\wt \bt)$ depends only on $t$, which contradicts the
assumption.
Thus, $\bt$ is unique and hence $q_{T\mid\Z}(t,\z)$ is unique. This
  directly implies that $q_T(t)$ is identifiable.
\end{proof}

\section{Sketch of proofs of Lemmas~\ref{lm:3.2} and~\ref{lm:3.3} and Theorems~\ref{Th:3.1} and~\ref{Th:3.2}}

 A detailed proof is provided in Supplement~S3.2, and we summarize the key steps below.

For consistency, we first show in Lemma~{S3.2} that under Conditions A1--A5, $$\sup_{\bt \in \Theta}|\ell_n (\bt)-E\{\ell(X, \Z, \Delta, R ; \bt)\}| \povr 0.$$ This result is obtained by proving in Lemma~S3.1 that the relevant classes of functions involved in the approximated likelihood $\ell_n(\bt)$ are Glivenko-Cantelli classes. This leads to uniform convergence in probability and, consequently, the consistency of the estimator.

For asymptotic normality, the main technical challenge lies in the presence of the nonparametrically estimated components $\wh P_T(t)$ and $\wh Q_\Z(\z)$. We decompose the derived score function into three parts, which we call: $\psi_1, \psi_{2 n}$, and $\psi_{3 n}$. The first part $\psi_1$ is handled using the standard central limit theorem. 

In Lemma~\ref{lm:3.2}, we provide the simplified expression of $\psi_{2n}$, which involves $\wh Q_\Z(\z)$. We transform $\psi_{2n}$ into a U-statistic via a Taylor expansion and simplify by using properties of U-statistics. This simplification is necessary to apply the central limit theorem.

In Lemma~\ref{lm:3.3}, we simplify $\psi_{3n}$, which contains terms involving $\wh Q_\Z(\z)$ inside a Kaplan-Meier integral based on $\wh P_T(t)$. To facilitate this simplification, we propose  Lemma~\ref{lm:c1} in Appendix C, which establishes the uniform properties of the Kaplan-Meier integral.  Lemma~\ref{lm:c1} extends the conditions of \cite{sellero2005uniform} by imposing assumptions on the bracketing and covering numbers of the relevant function class, thereby making the verification process easier. In Lemmas S3.3--S3.5, we show that the function classes appearing in the Kaplan-Meier integrals in our case satisfy these conditions in  Lemma~\ref{lm:c1}, using tools from empirical process theory. Finally, in Lemma~\ref{lm:3.3}, we combine the results from Lemmas S3.3 to S3.5 to obtain a simplified representation of $\psi_{3n}$.

Finally, in Theorem~\ref{Th:3.2}, we plug these results back into the score function and establish asymptotic normality via the central limit theorem.

\section{Technical Lemma}

Let $P$ be the probability measure of $(X,\Delta,\Z,R)$  where $\Z$ is
$d_\z$-dimensional. Let
$\mathcal{S}_{\Z p}$ and $\mathcal{S}_{\Z q}$ denote the support 
of $p_\Z(\z)$ and $q_\Z(\z)$ respectively. Let $\Theta$ denote
the parameter space of $\bt$. Let  $\|\cdot\|_k$ denote
the vector $\ell_k$-norm. For $\varepsilon>0$, we use $\mathcal{N}_{[]}(\varepsilon,
\mathcal{F},\|\cdot\|)$ and $\mathcal{N}(\varepsilon,
\mathcal{F},\|\cdot\|)$ to denote the $\varepsilon$-bracketing number
and the $\varepsilon$-covering number of any metric space
$(\mathcal{F},\|\cdot\|)$ respectively~\citep{van1996weak}.

\begin{Lemma}

\label{lm:c1}
Under Conditions A1, if the family of functions $\{\varphi\}$ satisfies:
\begin{enumerate}[label=(\alph*)]
    \item the family of functions $\{\varphi\}$ has an integrable envelope function $\Phi(t)$, such that  \bse
&& E[\Phi(T)\{1-P_C(T)\}^{-2}\{1-P_T(T)\}^{-5}] < \infty, \\
&& E[\Phi^2(T)\{1-P_C(T)\}^{-2}\{1-P_T(T)\}^{-3}] < \infty,\\
&& E[\{1-P_C(T)\}^{-2}] < \infty;
\ese

\item either of the following holds:
\begin{enumerate}
    \item[(b1)] the family of functions $\{\varphi\}$ is a measurable VC-subgraph class of functions;

    \item[(b2)] the family of functions $\{\varphi\}$ satisfies 
$$
\mathcal{N}_{[]}\left(\varepsilon, \{\varphi\}, L_1(P)\right)<\infty, \quad \text { for every } \varepsilon>0,
$$
and has a measurable envelope function $\Phi(t)$ such that $$
\int_0^{\infty} \sup_{\wt P} \sqrt{\log \mathcal{N}\left(\varepsilon\| \Phi\|_{L_2({\wt P})}, \{\varphi\}, L_2({\wt P})\right)} d \varepsilon<\infty,
$$
where the $\sup$ is taken over all probability measures $\wt P$ on $\mathcal{S}_T$ such that $\|\Phi\|_{L_2(\wt P)}<\infty$;
\end{enumerate} 

\end{enumerate}
then we have
\begin{enumerate}[label=(\arabic*),ref=(\arabic*)]
    \item \bse
\int \varphi(t) d \wh P_T(t)&=& n_1^{-1} \sum_{i=1}^{n_1}\{\varphi(X_i) \gamma_0(X_i)\Delta_i+\gamma^\varphi_1(X_i)(1-\Delta_i)- \gamma^\varphi_2(X_i)\}+R_{n_1}(\varphi)\\
&\equiv& n_1^{-1} \sum_{i=1}^{n_1} \eta^\varphi(X_i)+R_{n_1}(\varphi),
\ese
where $\sup_{\varphi}|R_{n_1}(\varphi)|=O({\log^3 n_1}/{n_1})$ almost surely and 
\bse
\wt{H}^0(x)&\equiv&\int_{-\infty}^x(1-P_T(y)) P_C(d y), \\ \wt{H}^1(x)&\equiv&\int_{-\infty}^x(1-P_C(y-)) P_T(d y), \\
\gamma_0(x) &\equiv& \exp \left\{\int_{-\infty}^{x-} \frac{\wt{H}^0(d z)}{\{1-P_T(z)\}\{1-P_C(z)\}}\right\}, \\
\gamma^\varphi_1(x) &\equiv& \frac{1}{\{1-P_T(x)\}\{1-P_C(x)\}} \int I{(x<w)} \varphi(w) \gamma_0(w) \wt{H}^1(d w), \\
\gamma^\varphi_2(x) &\equiv& \iint \frac{I{(v<x, v<w)} \varphi(w) \gamma_0(w)}{[\{1-P_T(v)\}\{1-P_C(v)\}]^2} \wt{H}^0(d v) \wt{H}^1(d w),
\ese

\item $$\sup_{\varphi}\left| \int \varphi(t) d \{\wh P_T(t) - P_T(t) \} \right| = O_p(n_1^{-1/2}),$$

\item for every $\varepsilon>0$,
$$
\lim _{\alpha \rightarrow 0} \limsup_{n_1 \rightarrow \infty} P\left(\sup_{\operatorname{var}_p(\eta^{\varphi-\wt \varphi})<\alpha} \left|\sqrt{n_1}\int \{\varphi(t)- \wt \varphi(t) \}d \{\wh P_T(t) - P_T(t) \}\right|>\varepsilon \right)=0.
$$
\end{enumerate}

\end{Lemma}

\begin{Remark}

 In Lemma~\ref{lm:c1}, the uniform expression of the Kaplan-Meier
 integral follows the expression provided by \cite{stute1995central}
 rather than that in \cite{sellero2005uniform}. This adjustment is
 made because \cite{sellero2005uniform} consider the presence of
 truncation. It can be easily verified that, in the absence of
 truncation as in our case, the main term in \cite{sellero2005uniform}
 reduces to the expression in \cite{stute1995central}.  
 
\end{Remark}

\end{appendix}


\bibliographystyle{apalike}
\bibliography{Label_Shift_Literature}

\setcounter{section}{0}
\setcounter{equation}{0}\renewcommand{\theequation}{S\arabic{equation}}
\setcounter{section}{0}\renewcommand{\thesection}{S\arabic{section}}
\setcounter{table}{0}\renewcommand{\thetable}{S\arabic{table}}
\setcounter{figure}{0}\renewcommand{\thefigure}{S\arabic{figure}}

\newpage
\setcounter{page}{1}

\section{Examples}
\label{Sup:Example}
\begin{proof}[Proof of Example~3.2]
To verify that $q_{T \mid \Z}(t,\z;\bt)/q_{T \mid
  \Z}(t,\z;\wt \bt)$ depends on $\z$, we compute
\bse
\frac{\partial}{\partial \z}\log \frac{q_{T \mid \Z}(t,\z;\bt)}{q_{T \mid \Z}(t,\z;\wt\bt)}
= (\bb - \wt \bb) - \{\bb H(t;\gamma,\lambda) \exp(\z\trans\bb) - \wt \bb H(t;\wt \gamma,\wt \lambda) \exp(\z\trans\wt \bb)\}
\ese
and verify that it is not the zero function.  In fact, if it were the zero
function, setting $t=0$ would lead to $\bb=\wt\bb$,
  which would further imply $\gamma=\wt\gamma$ and
  $\lambda=\wt\lambda$,
  which contradicts the assumption $\bt \neq \wt \bt$. 
\end{proof}

\begin{Example}[Log linear model]
For the log linear model with $\bt \equiv (\bb\trans,\mu,\sigma)\trans$ and
    $$\log T = \mu + \z\trans \bb + \sigma \varepsilon,$$
    where $\varepsilon$ follows a standard Normal distribution,
    we have that $q_{T \mid \Z}(t,\z;\bt)/q_{T \mid
      \Z}(t,\z;\wt\bt)$ depends on $\z$ whenever $\bt \neq \wt\bt$. Hence by Lemma 3.1, the model
    is identifiable.
\end{Example}

\begin{proof}
  Assume the contrary where
  for some $\bt \neq \wt \bt$, $q_{T \mid \Z}(t,\z;\bt)/q_{T \mid
    \Z}(t,\z;\wt\bt)$ does not depend on $\z$ and  $\var_q(\Z)\neq
  \0$. This implies that for every $t$ and $\z$,
\bse
\frac{\partial}{\partial \z}\log \frac{q_{T \mid \Z}(t,\z;\bt)}{q_{T \mid \Z}(t,\z;\wt\bt)}
&=& \frac{(\log(t)-\mu-\z\trans\bb)\bb}{\sigma^2} - \frac{(\log(t)-\wt\mu-\z\trans\wt\bb)\wt\bb}{\wt \sigma^2} \\&=& \0,
\ese
which leads to $\bt=\wt\bt$ contradicting the assumption $\bt \neq \wt
\bt$. 
\end{proof}

\begin{Example}[Proportional odds model]
For the proportional odds model with LogLogistic baseline  survival function with $\bt=(\bb\trans, \mu,\sigma)\trans$, the conditional density is defined as$$q_{T \mid \Z}(t,\z;\bt) \equiv \frac{h(t;\mu,\sigma)\exp(\z\trans \bb)}{\{1+H(t;\mu,\sigma)\exp(\z\trans \bb)\}^2},$$ 
where $h(t;\mu,\sigma) \equiv \exp(-\mu/\sigma)/\sigma
t^{1/\sigma-1} $ and $H(t;\mu,\sigma)\equiv \exp(-\mu/\sigma)
t^{1/\sigma}$. If $\bt \neq \wt\bt$, $q_{T \mid \Z}(t,\z;\bt)/q_{T
  \mid \Z}(t,\z;\wt\bt)$ depends on $\z$. Hence by Lemma 3.1, the
model is identifiable.
\begin{proof}
Assume the contrary that for some $\bt \neq \wt \bt$, $q_{T \mid \Z}(t,\z;\bt)/q_{T \mid
  \Z}(t,\z;\wt \bt)$ does not depend on $\z$, which implies that
\bse
\frac{\partial}{\partial \z}\log \frac{q_{T \mid \Z}(t,\z;\bt)}{q_{T \mid \Z}(t,\z;\wt \bt)} = (\bb - \wt\bb)- 2\frac{H(t;\mu,\sigma)\exp(\z\trans \bb)\bb}{1+H(t;\mu,\sigma)\exp(\z\trans \bb)}+2\frac{H(t;\wt\mu,\wt\sigma)\exp(\z\trans \wt\bb)\wt\bb}{1+H(t;\wt\mu,\wt\sigma)\exp(\z\trans \wt\bb)}= \0
\ese
for all $t$ and $\z$.
Taking  $t=0$ leads to $\bb=\wt\bb$, which further leads to
  $\mu=\wt\mu$ and $\sigma=\wt\sigma$, hence
$\bt=\wt \bt$,  which is a contradiction.
\end{proof}
\end{Example}

\begin{Example}[Accelerated hazards model]
For the accelerated hazards model with Weibull baseline hazard
function, with $\bt=(\bb\trans,\gamma,\lambda)\trans$, the conditional density is given by
$$
q_{T \mid \Z}(t,\z;\bt) \equiv h\{t \exp (\z\trans \bb);\lambda,\gamma\} \exp \left[-\int_0^t h\{s\exp (\z\trans \bb);\lambda,\gamma\} d s\right],
$$
where $h(t;\lambda, \gamma)\equiv\gamma\lambda t^{\gamma-1}, \gamma\ne1$.
We have that
$q_{T \mid \Z}(t,\z;\bt)/q_{T \mid
      \Z}(t,\z;\wt\bt)$ depends on $\z$ whenever $\bt \neq \wt \bt$. Hence by Lemma 3.1, the
model is identifiable.
\end{Example}
\begin{proof}
Assume the contrary that for some $\bt\neq \wt
\bt$, $$\frac{\partial}{\partial \z} \log \frac{q_{T \mid
    \Z}(t,\z;\bt)}{q_{T \mid \Z}(t,\z;\wt \bt)}  = (\gamma-1)\bb +
\lambda (1-\gamma) t^{\gamma}\exp\{(\gamma-1)\z\trans\bb\}\bb -(\wt
\gamma-1)\wt \bb - \wt \lambda (1-\wt \gamma) t^{\wt \gamma}\exp\{(\wt
\gamma-1)\z\trans\wt\bb\}\wt\bb=\0$$
for all $t,\z$.
Setting $t=0$ leads to $(\gamma-1)\bb=(\wt
\gamma-1)\wt \bb$. Then, setting $t=1$ and $\z=\0$ leads to
$\lambda (1-\gamma)\bb=\wt \lambda (1-\wt \gamma)\wt\bb$.  Subsequently, setting $\z=0$ implies $\gamma=\wt\gamma$. Hence we obtain $\bt=\wt\bt$, which is a contradiction.
\end{proof}

\section{Hypothesis Test}
\label{sup:HT}

We want to test
$$H_0: \forall(\z, t) \ p_{\Z \mid T}(\z, t)=q_{\Z \mid T}(z, t)
\text{ vs. } H_1: \exists (\z, t) \ p_{\Z \mid T}(\z, t)\neq q_{\Z \mid T}(z, t) \
.$$
Let $P_{\Z,T}(\z,t)$ and $Q_{\Z,T}(\z,t)$ denote the joint
distribution functions of $\Z$ and $T$ and let $P_{T}(t)$ and
$Q_{T}(t)$ denote the marginal distribution functions of $T$,  in
populations $\calP$ and $\calQ$ respectively. The null hypothesis is
equivalent to
\bse 
\frac{p_{\Z, T}(\z, t)}{p_T(t)}=\frac{q_{\Z, T}(\z, t)}{q_T(t)} 
\ \forall(\z, t) \Leftrightarrow \frac{\frac{d}{d t} P_{\Z, T}(\z, t)}{\frac{d}{d t} P_T(t)}=\frac{\frac{d}{d t} Q_{\Z, T}(\z, t)}{\frac{d}{d t} Q_T(t)} \ \forall(\z, t).
\ese
Define $$R_p(\z,t)\equiv \frac{\frac{d}{d t} P_{\Z, T}(\z, t)}{\frac{d}{d t} P_T(t)}, \quad R_q(\z,t)\equiv \frac{\frac{d}{d t} Q_{\Z, T}(\z, t)}{\frac{d}{d t} Q_T(t)}.$$

We estimate $P_{\Z, T}(\z, t)$ and $Q_{\Z, T}(\z, t)$ based on the estimator of the distribution function proposed by \cite{stute1993consistent} and $P_T(t)$ and $Q_T(t)$ are estimated using the Kaplan-Meier estimator. The derivatives $\frac{d}{d t} P_{\Z, T}(\z, t)$, $\frac{d}{d t} Q_{\Z, T}(\z, t)$, $\frac{d}{d t} P_T(t)$ and $\frac{d}{d t} Q_T(t)$ are estimated using kernel smoothing method~\citep{mielniczuk1986some}. Then we obtain $\wh R_p(\z,t)$ and $\wh R_q(\z,t)$. The test statistic is given by $$
T_n=\iint \{\wh{R}_p(\z, t)-\wh{R}_q(\z, t)\}^2 d\z d t.
$$

To determine the critical value, we apply a bootstrap procedure. Specifically, we resample data $K$ times from populations $\calP$ and $\calQ$, re-estimate each time $\wh{R}_p(\z,t)$ and $\wh{R}_q(\z,t)$, and obtain $K$ realizations of the bootstrapped statistics $T_{n}^*$. The critical value for the nominal level $(1-\alpha)$ is then given by the $(1-\alpha)$th quantile of the distribution of $
T_{n}^*-T_n.$

\section{Technical Details}
\label{Sup:DT}

Let $P$ be the probability measure of $(X,\Delta,\Z,R)$  where $\Z$ is
$d_\z$-dimensional. Let
$\mathcal{S}_{\Z p}$ and $\mathcal{S}_{\Z q}$ denote the support 
of $p_\Z(\z)$ and $q_\Z(\z)$ respectively. We assume the support of
$q_T(t)$ is included in the support of $p_T(t)$ and let
$\mathcal{S}_T$ denote the support of $p_T(t)$. Let $\Theta$ denote
the parameter space of $\bt$. Let $E_p$ and $E_q$ denote the
expectation with respect to the population $\calP$ and $\calQ$
respectively, and let $E(\cdot |x)$ denote the conditional expectation given $X=x$. We use the notation $\int q_{T \mid \Z}(t, \z ; \bt)
d Q_\Z(\z) \equiv \int q_{T \mid \Z}(t, \z ; \bt) q_\Z(\z) d \z$
and $\int q_{T \mid \Z}(t, \z ; \bt) d \wh Q_\Z(\z) \equiv
  n_2^{-1}\sum_{j=1}^n (1-r_j) q_{T\mid \Z}(t,\z_j;\bt)$.
Let  $\|\cdot\|_k$ denote
the vector $\ell_k$-norm. For $\varepsilon>0$, we use $\mathcal{N}_{[]}(\varepsilon,
\mathcal{F},\|\cdot\|)$ and $\mathcal{N}(\varepsilon,
\mathcal{F},\|\cdot\|)$ to denote the $\varepsilon$-bracketing number
and the $\varepsilon$-covering number of any metric space
$(\mathcal{F},\|\cdot\|)$ respectively~\citep{van1996weak}.

\subsection{Proof of Example 3.1}
\label{ssup:example3.1}
Given that \begin{equation*}
    p_{T\mid Z}(t,z;\beta) \equiv \frac{1}{ \sqrt{2\pi} t \beta z} \exp\left\{-\frac{(\log t)^2}{2(\beta z)^2}\right\},
\end{equation*}
and that the covariate $Z$ takes discrete values with equal probability:
$\pr(Z=1) = 0.5$ and $\pr(Z=2) = 0.5,$
the conditional density $f_{Z\mid T}(z,t)$ is given by 
\bse
 f_{Z\mid T}(z,t)
 &=& \frac{p_{T\mid Z}(t,z;\beta)\pr(Z=z)}{\sum_{z=1,2} \{\pr(Z=z)p_{T\mid Z}(t,z;\beta)\}}\\
 &=&\frac{ \exp\{-{(\log t)^2}/{2(\beta z)^2}\}/{\beta
     z}}{\sum_{z=1,2} [
\exp\{-{(\log t)^2}/{2(\beta z)^2}\}/{\beta z}]},
\ese

such that the following equation holds \bse
&& \int f_{Z\mid T}(z,t)q_T(t;\mu,\sigma)dt\\
&=& \int \frac{ \exp\{-{(\log t)^2}/{2(\beta z)^2}\}/{\beta z}}{\sum_{z=1,2} [\exp\{-{(\log t)^2}/{2(\beta z)^2}\}/{\beta z}]} \frac{1}{ \sqrt{2\pi} \sigma t} \exp\left\{-\frac{(\log t
        - \mu)^2}{2 \sigma^2}\right\}dt\\
&=&  \int \frac{ \exp\{-{(\log t)^2}/{2(\beta z)^2}\}/{\beta z}}{\sum_{z=1,2} [\exp\{-{(\log t)^2}/{2(\beta z)^2}\}/{\beta z}]} \frac{1}{\sqrt{2\pi} \sigma } \exp\left\{-\frac{(\log t
        - \mu)^2}{2 \sigma^2}\right\}d \log t\\
&=& \int \frac{ \exp\{-{(\log t)^2}/{2(\beta z)^2}\}/{\beta z}}{\sum_{z=1,2} [\exp\{-{(\log t)^2}/{2(\beta z)^2}\}/{\beta z}]} \frac{1}{ \sqrt{2\pi} \sigma t} \exp\left\{-\frac{(\log t
        + \mu)^2}{2 \sigma^2}\right\}dt\\
&=& \int f_{Z\mid T}(z,t)q_T(t;-\mu,\sigma)dt.
\ese
Hence, we have that $q_T(t;\mu,\sigma)$ and $q_T(t;-\mu,\sigma)$ satisfy \bse
    q_Z(z) = \int f_{Z\mid T}(z,t)q_T(t;\mu,\sigma)dt = \int f_{Z\mid T}(z,t) q_T(t;-\mu,\sigma) dt.
\ese

\subsection{Proof of Theorem 3.1 and Theorem 3.2}
\label{Sup:Thproof}

\subsubsection{Notations}
\label{Sup:Notations}

To simplify the writing, we define 
\bse
 q_T(t) &\equiv& \int q_{T\mid \Z}(t,\z;\bt_0)d Q_\Z(\z),\quad
 \wh q_T(t) \equiv\int q_{T\mid \Z}(t,\z;\bt_0)d \wh Q_\Z(\z),\\
q_T^*(t) &\equiv& \tfrac{\partial}{\partial \bt} \int q_{T\mid \Z}(t,\z;\bt_0)d Q_\Z(\z),\quad
\wh{q}_T^*(t) \equiv \tfrac{\partial}{\partial \bt} \int q_{T\mid \Z}(t,\z;\bt_0)d \wh Q_\Z(\z),\\
 S_0(x,\z,q_T) &\equiv& \int \frac{I(t>x) q_{T \mid \Z}(t,\z;\bt_0)}{\int q_{T \mid \Z}(t,\z;\bt_0)d Q_\Z(\z)}d P_T(t)=\frac{\int_x^\infty p_{T,\Z}(t,\z)dt}{q_\Z(\z)},\\
 S_{0n}(x,\z,\wh q_T) &\equiv& \int \frac{I(t>x)q_{T \mid \Z}(t,\z;\bt_0)}{\int q_{T \mid \Z}(t,\z;\bt_0)d \wh Q_\Z(\z)}d \wh P_T(t),\\
 S_1(x,\z,q_T) &\equiv& \int \frac{I(t>x)\tfrac{\partial}{\partial \bt}  q_{T \mid \Z}(t,\z;\bt_0)}{\int q_{T \mid \Z}(t,\z;\bt_0)d Q_\Z(\z)}d P_T(t),\quad
 S_{1n}(x,\z,\wh q_T) \equiv \int \frac{I(t>x)\tfrac{\partial}{\partial \bt} q_{T \mid \Z}(t,\z;\bt_0)}{\int q_{T \mid \Z}(t,\z;\bt_0)d \wh Q_\Z(\z)}d \wh P_T(t),\\
 S_2(x,\z,q_T,q^{*}_T) &\equiv& \int \frac{I(t>x) q_{T \mid \Z}(t, \z ; \bt_0) \int \tfrac{\partial}{\partial \bt} q_{T\mid \Z}(t, \z ; \bt_0) d  Q_\Z(\z)}{\{\int q_{T \mid \Z}(t, \z ; \bt_0) d  Q_\Z(\z)\}^2} d P_T(t),\\
 S_{2n}(x,\z,\wh q_T,\wh q^{*}_T) &\equiv& \int \frac{I(t>x) q_{T \mid \Z}(t, \z ; \bt_0) \int \tfrac{\partial}{\partial \bt} q_{T\mid \Z}(t, \z ; \bt_0) d  \wh Q_\Z(\z)}{\{\int q_{T \mid \Z}(t, \z ; \bt_0) d  \wh Q_\Z(\z)\}^2} d \wh P_T(t),\\
 \varphi_0( t;x, \z) &\equiv& \frac{I(t>x) q_{T \mid \Z}(t,\z;\bt_0)}{\int q_{T \mid \Z}(t,\z;\bt_0)d Q_\Z(\z)}=\frac{I(t>x)q_{T \mid \Z}(t,\z;\bt_0)}{q_T(t)}=\frac{I(t>x) f_{\Z \mid T}(\z,t)}{q_\Z(\z)},\\
 \varphi_1(t;x, \z) &\equiv& \frac{I(t>x)\tfrac{\partial}{\partial \bt}  q_{T \mid \Z}(t,\z;\bt_0)}{\int q_{T \mid \Z}(t,\z;\bt_0)d Q_\Z(\z)}=\varphi_0(t;x, \z)\frac{\tfrac{\partial}{\partial \bt}  q_{T \mid \Z}(t,\z;\bt_0)}{q_{T \mid \Z}(t,\z;\bt_0)},\\
 \varphi_2(t;x, \z) &\equiv& \frac{I(t>x) q_{T \mid \Z}(t, \z ; \bt_0) \int \tfrac{\partial}{\partial \bt} q_{T\mid \Z}(t, \z ; \bt_0) d  Q_\Z(\z)}{\{\int q_{T \mid \Z}(t, \z ; \bt_0) d  Q_\Z(\z)\}^2}=\varphi_0(t;x, \z)\int \frac{\tfrac{\partial}{\partial \bt}  q_{T \mid \Z}(t,\z;\bt_0)}{q_{T \mid \Z}(t,\z;\bt_0)}f_{\Z\mid T}(\z,t)d\z,\\
\eta_{0q}(\Z;x,\z)&\equiv& -\int  \frac{I(t>x) q_{T \mid \Z}(t, \z ; \bt_0)\{q_{T \mid \Z}(t,\Z;\bt_0)-q_T(t)\}}{\{\int q_{T \mid \Z}(t, \z ; \bt_0) d Q_\Z(\z)\}^2}d P_T(t),\\
\eta_{1q}(\Z;x,\z)&\equiv& -\int  \frac{I(t>x) \tfrac{\partial}{\partial \bt} q_{T \mid \Z}(t, \z ; \bt_0)\{q_{T \mid \Z}(t,\Z;\bt_0)-q_T(t)\}}{\{\int q_{T \mid \Z}(t, \z ; \bt_0) d Q_\Z(\z)\}^2}d P_T(t),\\
\eta_{2q}(\Z;x,\z)&\equiv&  \int  \frac{I(t>x) q_{T \mid \Z}(t, \z ; \bt_0)\{\tfrac{\partial}{\partial \bt} q_{T \mid \Z}(t,\Z;\bt_0)-q^*_T(t)\}}{\{\int q_{T \mid \Z}(t, \z ; \bt_0) d Q_\Z(\z)\}^2}d P_T(t) \notag \\
 &&- \int  \frac{2 I(t>x) q_{T \mid \Z}(t, \z ; \bt_0)q^*_T(t) \{q_{T \mid \Z}(t,\Z;\bt_0)-q_T(t)\}}{\{\int q_{T \mid \Z}(t, \z ; \bt_0) d Q_\Z(\z)\}^3}d P_T(t),\\
\eta_{ip}(X,\Delta; x,\z) &\equiv& \varphi_i(X; x,\z) \gamma_0(X)\Delta + \gamma_1^{\varphi_i}(X)(1 - \Delta) - \gamma_2^{\varphi_i}(X), \quad i = 0,1,2,\ese
 where
\bse
\wt{H}^0(x)&\equiv&\int_{-\infty}^x(1-P_T(y)) P_C(d y), \\ \wt{H}^1(x)&\equiv&\int_{-\infty}^x(1-P_C(y-)) P_T(d y), \\
\gamma_0(X) &\equiv& \exp \left\{\int_{-\infty}^{X-} \frac{\wt{H}^0(d z)}{\{1-P_T(z)\}\{1-P_C(z)\}}\right\}, \\
\gamma_1^{\varphi_i}(X) &\equiv& \frac{1}{\{1-P_T(X)\}\{1-P_C(X)\}} \int I{(X<w)} \varphi_i(w;x,\z) \gamma_0(w) \wt{H}^1(d w), \\
\gamma_2^{\varphi_i}(X) &\equiv& \iint \frac{I{(v<X, v<w)} \varphi_i(w;x,\z) \gamma_0(w)}{[\{1-P_T(v)\}\{1-P_C(v)\}]^2} \wt{H}^0(d v) \wt{H}^1(d w).
\ese
In the following theoretical developments, we will encounter terms of the form
$n_1^{-1} \sum_{i=1}^{n_1} g_1(X_i, \Delta_i, \Z_i) \quad$ and $n_2^{-1} \sum_{i=1}^{n_2} g_2(\Z_i)$
for some functions $g_1$ and $g_2$. We can rewrite them as follows:
\bse
n_1^{-1} \sum_{i=1}^{n_1} g_1(X_i, \Delta_i, \Z_i)
&=& n^{-1} \sum_{i=1}^{n} \frac{R_i}{\pi_n} g_1(X_i, \Delta_i, \Z_i) \\
&=& n^{-1} \sum_{i=1}^{n} \frac{R_i}{\pi} \frac{\pi}{\pi_n} g_1(X_i, \Delta_i, \Z_i) \\
&=& n^{-1} \sum_{i=1}^{n} \frac{R_i}{\pi} g_1(X_i, \Delta_i, \Z_i) + o(1),
\ese
where $\pi/\pi_n \rightarrow 1$. Similarly, for $g_2$, we have
\bse
n_2^{-1} \sum_{i=1}^{n_2} g_2(\Z_i) = n^{-1} \sum_{i=1}^{n} \frac{1-R_i}{1-\pi} g_2(\Z_i) + o(1).
\ese

\subsubsection{Conditions}
\label{Sup:Conditions}

\begin{enumerate}[label=A\arabic*,ref=A\arabic*]
    \item \label{A1}  The support of
$q_T(t)$ is included in the support of $p_T(t)$. The parameter space $\Theta$, and the supports $\mathcal{S}_{T}$, $\mathcal{S}_{\Z p}$
      and $\mathcal{S}_{\Z q}$ are all compact. The censoring time $C$ is independent of $(T,\Z)$. Let $\tau_T \equiv \inf\{t: P_T(t)=1\}$ and $\tau_C \equiv \inf\{t: P_C(t)=1\}$. Then $\tau_T \leq \tau_C$, and for every $\z$, $\pr(\Delta=1\mid\z)>0$.

      \item \label{A2} There exist functions $m_1(t,\z)$ and $M_1(t,\z)$
      such that for every $t \in \mathcal{S}_T$, $\bt \in \Theta$
      and $\z \in \mathcal{S}_{\Z p}$, $0\leq m_1(t,\z)\leq
      q_{T \mid \Z}(t, \z ; \bt) \leq M_1(t,\z)$ and
      $E_p[|\log \{m_1(T,\Z)\}|], E_p[|\log \{M_1(T,\Z)\}|]$ and
      $E_p\{M_1^2(T,\Z)\}$ are finite. 

    There exists a function $c_1: \mathcal{S}_T \times \mathcal{S}_{\Z
      p} \rightarrow \mathbb{R}_{\geq 0}$ such that for every $(t,\z)
    \in \mathcal{S}_T \times \mathcal{S}_{\Z p}$ and every $\bt,\wt
    \bt \in \Theta$, 
    $$| q_{T \mid \Z}(t,\z; \bt)-q_{T \mid \Z}(t,\z; \wt \bt) | \leq
     \|\bt - \wt \bt\|_1c_1(t,\z),$$
    with $0<E_p\{|c_1(T,\Z)|\}<+\infty$.

    There exists a function $c_2: \mathcal{S}_{\Z q} \rightarrow
    \mathbb{R}_{\geq 0}$ such that for every $\z \in \mathcal{S}_{\Z
      q}$ and every $(t,\bt\trans)\trans,(\wt t,\wt \bt\trans)\trans
    \in \mathcal{S}_T \times \Theta$, 
    $$\left | q_{T \mid \Z}(t,\z; \bt)-q_{T \mid \Z}(\wt t,\z; \wt
      \bt)\right | \leq \left \|(t,\bt\trans)\trans - (\wt t,\wt
      \bt\trans)\trans \right\|_1c_2(\z),$$
    with $0< E_q[\{c_2(\Z)\}^2]<+\infty$.

      \item \label{A3} There exist functions $m_2(t)$ and $M_2(t)$ such
      that for every $t \in \mathcal{S}_T$, and $\bt \in \Theta$, $0
      \leq m_2(t)\leq \int q_{T \mid \Z}(t,\z; \bt)
      dQ_\Z(\z) \leq  M_2(t)$ and $E_p[|\log \{m_2(T)\}|],
      E_p[|\log\{M_2(T)\}|]$, $E_p\{M_2^2(T)\}$ and
       $E_p\{1/m_2(T)\}$ are finite.  
    
    There exists a function $c_3: \mathcal{S}_{T} \rightarrow
    \mathbb{R}_{\geq 0}$ such that for every $\bt,\wt \bt \in
    \Theta$, 
    $$\left | \int q_{T \mid \Z}(t,\z; \bt) dQ_\Z(\z)-\int q_{T \mid
        \Z}(t,\z; \wt \bt)dQ_\Z(\z)\right | \leq \|\bt - \wt \bt\|_1c_3(t),$$
    with $0<E_p\{|c_3(T)|\}<+\infty$.

    \item \label{A4} There exist functions $m_3(x,\z)$ and $M_3(x,\z)$
      such that for every  $x \in \mathcal{S}_T$, $\z \in \mathcal{S}_{\Z p}$ and $\bt \in
      \Theta$, $$0 \leq  m_3(x,\z)\leq\int \frac{I(t > x)q_{T
          \mid \Z}(t, \z ; \bt)}{\int q_{T \mid \Z}(t, \z ; \bt) d
        Q_\Z(\z)} d P_T(t)\leq M_3(x, \z).$$ 

    Let $$c_4(\z) \equiv \int \frac{M_1(t,\z) c_3(t)}{m_2^2(t)}+
    \frac{c_1(t,\z)}{m_2(t)} dP_T(t),$$
    with $0<E_p\{|c_4(\Z)|\}<+\infty$. $E_p[|\log \{m_3(C,\Z)\}|],
   E_p[|\log\{M_3(C,\Z)\}|]$, $E_p\{M_3^2(C,\Z)\}$ and
   $E_p\{1/m_3(C,\Z)\}$ are finite.

   \item \label{A5}
     The class
     $\{t \mapsto \frac{   q_{T\mid \Z}(t,\z;\bt)}{
\{\int q_{T\mid \Z}(t,\z;\bt)d Q_\Z(\z)\}^2}: \bt \in \Theta, \z \in
\mathcal{S}_{\Z p} \}$ has an envelope function $M_4(t)$ with
$E_p\{M_4(T)\}<+\infty$.

\item \label{A6}  $P_T$ and $P_C$ do not have jumps in
      common.  The class $\{t \mapsto \varphi(t; \bt,x,\z) \equiv \frac{ I(t>x)
  q_{T\mid \Z}(t,\z;\bt)}{
\int q_{T\mid \Z}(t,\z;\bt)d Q_\Z(\z)}: \bt \in \Theta, x \in \mathcal{S}_T, \z \in \mathcal{S}_{\Z p} \}$ has an integrable envelope function $\Phi(t)$, such that 
\bse
&& E[\Phi(T)\{1-P_C(T)\}^{-2}\{1-P_T(T)\}^{-5}] < \infty, \\
&& E[\Phi^2(T)\{1-P_C(T)\}^{-2}\{1-P_T(T)\}^{-3}] < \infty,\\
&& E[\{1-P_C(T)\}^{-2}] < \infty,
\ese
and 
either of the following holds:
\begin{enumerate}
    \item[(1)] the family of functions $\{\varphi\}$ is a measurable VC-subgraph class of functions;
    
    \item[(2)] the family of functions $\{\varphi\}$ satisfies 
$$
\mathcal{N}_{[]}\left(\varepsilon, \{\varphi\}, L_1(P)\right)<\infty, \quad \text { for every } \varepsilon>0,
$$
and has a measurable envelope function $\Phi(t)$ such that $$
\int_0^{\infty} \sup_{\wt P} \sqrt{\log \mathcal{N}\left(\varepsilon\|\Phi\|_{L_2(\wt P)}, \{\varphi\}, L_2(\wt P)\right)} d \varepsilon<\infty,
$$
where the $\sup$ is taken over all probability measures $\wt P$ on $\mathcal{S}_T$ such that $\|\Phi\|_{L_2(\wt P)}<\infty$.
\end{enumerate} 
    
\end{enumerate}

\begin{Remark}\label{rem:c4}

Condition \ref{A4} implies that the function $\int \frac{I(t > x) q_{T
    \mid \Z}(t, \z ; \bt)}{\int q_{T \mid \Z}(t, \z ; \bt) d Q_\Z(\z)}
d P_T(t)$ is Lipschitz continuous with respect to $\bt$.
To see this,
for any $\bt, \wt \bt\in \Theta$, 
\bse
&& \left| \int \frac{I(t> x) q_{T \mid \Z}(t, \z ; \bt)}{\int q_{T \mid \Z}(t, \z ; \bt) d Q_\Z(\z)} d P_T(t) - \int \frac{I(t> x) q_{T \mid \Z}(t, \z ; \wt \bt)}{\int q_{T \mid \Z}(t, \z ; \wt \bt) d Q_\Z(\z)} d P_T(t) \right |\\
&\leq& \int \left|\frac{q_{T \mid \Z}(t, \z ; \bt)}{\int q_{T \mid \Z}(t, \z ; \bt) d Q_\Z(\z)}  - \frac{q_{T \mid \Z}(t, \z ; \wt \bt)}{\int q_{T \mid \Z}(t, \z ; \wt \bt) d Q_\Z(\z)} \right | d P_T(t) \\
&=&  \int \left|\frac{q_{T \mid \Z}(t, \z ; \bt)\int q_{T \mid \Z}(t, \z ; \wt \bt) d Q_\Z(\z)  - \int q_{T \mid \Z}(t, \z ; \bt) d Q_\Z(\z)q_{T \mid \Z}(t, \z ; \wt \bt)}{\int q_{T \mid \Z}(t, \z ; \bt) d Q_\Z(\z) \int q_{T \mid \Z}(t, \z ; \wt \bt) d Q_\Z(\z)} \right | d P_T(t) \\
&\leq&  \int \left|\frac{q_{T \mid \Z}(t, \z ; \bt) \left\{\int q_{T \mid \Z}(t, \z ; \wt \bt) d Q_\Z(\z)-\int q_{T \mid \Z}(t, \z ; \bt) d Q_\Z(\z) \right\}}{\int q_{T \mid \Z}(t, \z ; \bt) d Q_\Z(\z) \int q_{T \mid \Z}(t, \z ; \wt \bt) d Q_\Z(\z)} \right |\\
&&  + \left|\frac{q_{T \mid \Z}(t, \z ; \bt) - q_{T \mid \Z}(t, \z ; \wt \bt)}{\int q_{T \mid \Z}(t, \z ; \wt \bt) d Q_\Z(\z)} \right |  d P_T(t) \\
&\leq&  \int \left\{ \frac{M_1(t,\z)c_3(t)}{m_2^2(t)}+
  \frac{c_1(t,\z)}{m_2(t)} \right\} dP_T(t)   \|\bt - \wt \bt\|_1\\
&=& c_4(\z)\|\bt - \wt \bt\|_1,
\ese
 where $E_p\{|c_4(\Z)|\}<\infty$ under Condition A4.
\end{Remark}

\begin{Remark}
    Condition \ref{A6}  is imposed to ensure that the class $\{\varphi\}$ satisfies the conditions in Lemma~C.1.
\end{Remark}

\begin{enumerate}[label=B\arabic*,ref=B\arabic*]

\item \label{B1} 
There exists a function $c_5:\mathcal{S}_T \mapsto \mathbb{R}_{\geq 0}$ such that for every $(\wt \z, \z) \in  (\mathcal{S}_{\Z p}\cup\mathcal{S}_{\Z q})^2$
$$|q_{T\mid\Z}(t,\z;\bt_0) - q_{T\mid\Z}(t,\wt \z;\bt_0)| \leq \|\z - \wt \z\|_1 c_5(t),$$
    with $0<E\{c_5(T)/q_T(T)\}< +\infty$.
Let $M_5(t) \equiv (\operatorname{diam} \mathcal{S}_{\Z
  p})c_5(t)/q_T(t)+|q_{T \mid \Z}(t, \z_0 ; \bt_0)/q_T(t)|$
for some fixed $\z_0 \in \mathcal{S}_{\Z p}$  where
$\operatorname{diam} \mathcal{S}_{\Z 
  p} \equiv \sup_{\z,\wt\z \in \mathcal{S}_{\Z p}}\|\z - \wt
\z\|_1$  and $E\{|q_{T \mid \Z}(T, \z_0 ; \bt_0)/q_T(T)|\}<\infty$. Then for every $t \in \mathcal{S}_T$ and $\z \in
\mathcal{S}_{\Z p}$, $|q_{T\mid\Z}(t,\z;\bt_0)/q_T(t)| \leq
M_5(t)$. 

\item \label{B2} 
There exists a function $c_6:\mathcal{S}_T \mapsto \mathbb{R}_{\geq 0}$ such that for every $(\wt \z, \z) \in (\mathcal{S}_{\Z p}\cup\mathcal{S}_{\Z q})^2$
$$\left|\tfrac{\partial}{\partial \bt}q_{T\mid\Z}(t,\z;\bt_0) - \tfrac{\partial}{\partial \bt}q_{T\mid\Z}(t,\wt \z;\bt_0)\right| \leq \|\z - \wt \z\|_1 c_6(t),$$
with $0<E\{c_6(T)/q_T(T)\}< +\infty$. Let
$M_6(t) \equiv (\operatorname{diam} \mathcal{S}_{\Z
  p})c_6(t)/q_T(t)+|\tfrac{\partial}{\partial \bt}q_{T \mid \Z}(t, \z_0 ;
\bt_0)/q_T(t)|$ for some fixed $\z_0 \in \mathcal{S}_{\Z p}$ where
$E\{|\tfrac{\partial}{\partial \bt}q_{T \mid \Z}(T, \z_0 ;
\bt_0)/q_T(T)|\}<\infty$. Then for every 
$t \in \mathcal{S}_T$ and $\z \in \mathcal{S}_{\Z p}$,
$|\tfrac{\partial}{\partial \bt} q_{T\mid\Z}(t,\z;\bt_0)/q_T(t)| \leq
M_6(t)$. 

There exists a function $c_7:\mathcal{S}_{\Z q} \mapsto \mathbb{R}_{\geq 0}$ such that for every $(\wt t, t) \in \mathcal{S}_T^2$
$$\left|\tfrac{\partial}{\partial \bt}q_{T\mid\Z}(t,\z;\bt_0) - \tfrac{\partial}{\partial \bt}q_{T\mid\Z}(\wt t, \z;\bt_0)\right| \leq \|t - \wt t\|_1 c_7(\z),$$
with $0<E[\{(1-R) c_7(\Z)\}^2]< +\infty$.

\item \label{B3} Let $\Phi_1(t) \equiv  M_5(t)$, $\Phi_2(t) \equiv
  M_6(t)$, $\Phi_3(t) \equiv  M_5(t)c_5(t)/q_T(t)$, $\Phi_4(t) \equiv
  M_5(T)c_6(t)/q_T(t)$ and $\Phi_5(t) \equiv  M_6(t)c_5(t)/q_T(t)$. It
  follows that
  \bse
&& E[\Phi_i(T)\{1-P_C(T)\}^{-2}\{1-P_T(T)\}^{-5}] < \infty, \\
&& E[\Phi_i^2(T)\{1-P_C(T)\}^{-2}\{1-P_T(T)\}^{-3}] < \infty,\\
&& E[\{1-P_C(T)\}^{-2}] < \infty;
\ese
for $i=1,2,3,4,5$. 

\item \label{B4} The functions $q_{T\mid\Z}(t,\z;\bt_0)$, $\frac{\partial}{\partial
    \bt}q_{T\mid\Z}(t,\z;\bt_0)$, $q_{T\mid\Z}(t,\z;\bt_0)/q_T(t)$ and
  $\tfrac{\partial}{\partial \bt}q_{T\mid\Z}(t,\z;\bt_0)/q_T(t)$ are
  bounded on $\mathcal{S}_T \times (\mathcal{S}_{\Z_p} \cup
  \mathcal{S}_{\Z q})$.  

\item \label{B5} The map $\bt \mapsto E\{\ell(X, \Z, \Delta, R; \bt)\}$ is twice differentiable at $\bt_0$ with a non-singular second derivative matrix.
\end{enumerate}

\begin{Remark}
Condition \ref{B3} guarantees that the classes defined in the proof of
Lemma \ref{lm:2.5} have envelop functions satisfying the Condition (a)
in Lemma~C.1. 
\end{Remark}

\subsubsection{Proof of Lemma~C.1}

\begin{proof}
Conditions A1, (a) and (b1) ensures that the class $\{\varphi\}$
    satisfies the conditions in Theorem 1 of
    \cite{sellero2005uniform}. We introduce an alternative option,
    Condition (b2), to replace Condition (b1). In the proof of Theorem
    1 in \cite{sellero2005uniform}, Condition (b1) was used to apply
    Corollary 5.2.3 and Theorem 5.4.1 from
    \cite{victor1999decoupling}. However, based on Corollary 5.2.9 and
    Remark 5.3.9 from \cite{victor1999decoupling}, Condition (b2) can
    also lead to the same results as in Theorem 1 of
    \cite{sellero2005uniform}.

 The expansion of $\int\varphi(t)d\wh P_T(t)$ in (1) is a
    direct application of Theorem 1 of \cite{stute1995central}.
    Under Conditions (a) and (b), the class $\{\varphi\}$ is Donsker
    by \cite{van1996weak}.  For $\gamma^\varphi_1(x)$, it follows that 
$$\int I{(x<w)} \varphi(w) \gamma_0(w)
    \wt{H}^1(d w) = \int I{(x<w)} \varphi^+(w) \gamma_0(w)
    \wt{H}^1(d w) - \int I{(x<w)} \varphi^-(w) \gamma_0(w)
    \wt{H}^1(d w),$$
where $\varphi^+(w)$ and $\varphi^-(w)$ denote the positive part and
the negative part of $\varphi(w)$ respectively. Hence, the class $\{x
\mapsto  \int I{(x<w)} \varphi(w) \gamma_0(w) 
    \wt{H}^1(d w)\}$ is the difference of two monotone, bounded
    functions by Condition (a). By \cite{van1996weak}[Theorem 2.7.5],
    for every $\varepsilon > 0$,  the $L_2(P)$-bracketing numbers for
    the class $\{x \mapsto  \int I{(x<w)} \varphi(w) \gamma_0(w) 
    \wt{H}^1(d w)\}$ is bounded by $\exp(K\varepsilon^{-1})$, where
    $K$ denotes some constant. Since the class $\{x \mapsto  \int
    I{(x<w)} \varphi(w) \gamma_0(w) 
    \wt{H}^1(d w)\}$ is uniformly bounded, the $L_2(P)$-bracketing
    integral defined in \cite{van1996weak}[Section 2.5.2] is finite,
    which implies the class $\{x \mapsto  \int I{(x<w)} \varphi(w)
    \gamma_0(w)\wt{H}^1(d w)\}$ is Donsker. 
   Multiplying the function ${1}/[{\{1-P_T(x)\}\{1-P_C(x)\}}]$
   preserves the Donsker property because the function is fixed and
   boundedness is maintained by Condition (a). Hence, the class
   $\{\gamma_1^\varphi\}$ is Donsker. 
    
Similarly, $\{\gamma_2^\varphi\}$ can also be expressed as the
difference of two monotone bounded functions, that is, 
$$\gamma^\varphi_2(x) =  \iint \frac{I{(v<x, v<w)} \varphi^+(w)
  \gamma_0(w)}{[\{1-P_T(v)\}\{1-P_C(v)\}]^2} \wt{H}^0(d v) \wt{H}^1(d
w) - \iint \frac{I{(v<x, v<w)} \varphi^-(w)
  \gamma_0(w)}{[\{1-P_T(v)\}\{1-P_C(v)\}]^2} \wt{H}^0(d v) \wt{H}^1(d
w).$$ 
Applying \cite{van1996weak}[Theorem 2.7.5] and by Condition (a),
$$\int_0^{\infty} \sqrt{\log
  \mathcal{N}_{[]}(\varepsilon,\{\gamma_2^\varphi\}, L_2(P))} d
\varepsilon<\infty.$$ Hence, $\{\gamma_2^\varphi\}$ is
Donsker. 
As $\gamma_0$ is a fixed function, the class $\{\eta^\varphi\}$ is
Donsker. Consequently, Results (2) and (3) follow from the properties
of the Donsker class and Corollary 2.3.12 in \cite{van1996weak}.
\end{proof}

\subsubsection{Technical Lemmas}

\begin{Lemma}\label{lm:pgc}
    \begin{enumerate}[label=(\arabic*)]
        Under Conditions \ref{A1}-\ref{A4}, the following holds:
        
        \item the class $\mathcal{F}_1 \equiv \left\{(x, \z, \delta) \mapsto q_{T \mid \Z}(x, \z ; \bt)^{ \delta}: \bt \in \Theta \right\}$ is P-Glivenko-Cantelli,

        \item the class $\mathcal{F}_2 \equiv \left\{(\z,r) \mapsto (1-r)q_{T \mid \Z}(t, \z ; \bt): \bt \in \Theta, t \in \mathcal{S}_T\right\}$ is Q-Glivenko-Cantelli,

        \item the class $\mathcal{F}_3 \equiv \left\{(x, \delta) \mapsto \{\int q_{T \mid \Z}(x, \z ; \bt) d Q_\Z(\z)\}^\delta : \bt \in \Theta \right\}$ is P-Glivenko-Cantelli,

        \item the class $\mathcal{F}_4 \equiv \left\{(x, \z,\delta) \mapsto \delta \log q_{T \mid \Z}(x, \z ; \bt): \bt \in \Theta\right\}$ is P-Glivenko-Cantelli,

        \item the class $\mathcal{F}_5 \equiv \left\{(x,\delta) \mapsto \delta \log \int q_{T \mid \Z}(x, \z ; \bt)dQ_\Z(\z): \bt \in \Theta\right\}$ is P-Glivenko-Cantelli,

        \item the class $\mathcal{F}_6 \equiv \left\{(x, \z, \delta) \mapsto \left \{ \int \frac{I(s>x) q_{T \mid \Z}(s, \z ; \bt)}{\int q_{T \mid \Z}(s, \z ; \bt) d Q_\Z(\z)} d P_T(s) \right \}^{1- \delta }: \bt \in \Theta\right\}$ is P-Glivenko-Cantelli,

        \item the class $\mathcal{F}_7 \equiv \left\{(x, \z, \delta) \mapsto (1-\delta) \log \int \frac{I(s>x) q_{T \mid \Z}(s, \z ; \bt)}{\int q_{T \mid \Z}(s, \z ; \bt) d Q_\Z(\z)} d P_T(s): \bt \in \Theta\right\}$ is P-Glivenko-Cantelli,

        \item the class $\mathcal{F}_8 \equiv \left\{(x, \delta) \mapsto {1}/\left\{\int  q_{T \mid \Z}(x, \z ; \bt) d Q_\Z(\z)\right\}^\delta: \bt \in \Theta  \right\}$ is P-Glivenko-Cantelli,

        \item the class $\mathcal{F}_9 \equiv \left\{(x,\z, \delta) \mapsto {1}/{\left\{\int \frac{I(s>x) q_{T \mid \Z}(s, \z ; \bt)}{\int q_{T \mid \Z}(s, \z ; \bt) d Q_\Z(\z)} d P_T(s)\right\}^{1- \delta } }: \bt \in \Theta \right\}$ is P-Glivenko-Cantelli.

    \end{enumerate}
\end{Lemma}

\begin{proof}

Let $\mathcal{N}_{[]}(\varepsilon, \mathcal{F},\|\cdot\|)$ and
$\mathcal{N}(\varepsilon, \mathcal{F},\|\cdot\|)$ denote the
$\varepsilon$-bracketing number and $\varepsilon$-covering number of
some metric space $(\mathcal{F},\|\cdot\|)$. 

 Result (1):
We first show that the class $\mathcal{F}_1$ is
Glivenko-Cantelli.  Let $\Theta_0$ denote a fixed big $\ell_1$-ball
that contains $\Theta$, i.e.
 $\Theta\subset\Theta_0$. Since the $\varepsilon$-covering number of
 $\Theta_0$ is $O(\varepsilon^{-d})$, i.e., $\mathcal{N}(\varepsilon,
 \Theta_0,\|\cdot\|_1)=O(\varepsilon^{-d})$, hence
 $\mathcal{N}(\varepsilon,
 \Theta,\|\cdot\|_1)=O(\varepsilon^{-d})$.
By \cite{van1996weak}[Theorem 2.7.11],
  for the $L_1(P)$-norm, we have $$
\mathcal{N}_{[]}(2 \varepsilon\|{c_1}\|_{L_1(P)}, \mathcal{F}_1, L_1(P))
\leq \mathcal{N}(\varepsilon, \Theta,\|\cdot\|_1)\leq K_1
\varepsilon^{-d}, 
$$
for some $K_1>0$.
Since $0<\|c_1\|_{L_1(P)}<+\infty$ by Condition \ref{A1},
$\mathcal{N}_{[]}\left(\varepsilon, \mathcal{F}_1,
  L_1(P)\right)<+\infty$ for every fixed $\varepsilon>0$. Then,
class $\mathcal{F}_1$ is Glivenko-Cantelli~\citep[{Theorem
  2.4.1}]{van1996weak}.

Result (4):
We next prove the class $\mathcal{F}_4$ is Glivenko-Cantelli.
Note that $\log (\cdot)$ is a continuous mapping. For every $x \in \mathcal{S}_T$, $\bt \in \Theta$, $\delta|\log q_{T \mid \Z}(x, \z ; \bt)|\le\delta|\log\{m_1(x,\z)\}|+\delta|\log\{
M_1(x,\z)\}|$ and $E_p[|\log\{m_1(T,\Z)\}|+|\log
\{M_1(T,\Z)\}|]$ is finite by Condition A2.
Thus
$\delta|\log\{m_1(x,\z)\}|+\delta|\log\{
M_1(x,\z)\}|$ is an envelope function of
${\cal F}_4$ and it is integrable.  Then applying \cite{van2000preservation}[Theorem 3] and Result (1),
${\cal F}_4$ is P-Glivenko-Cantelli.

Result (2):
The class $\mathcal{F}_2$ is indexed by $(\bt\trans,t)\trans \in \Theta\times\mathcal{S}_T \subset \Theta_0 \times \mathcal{S}_{T0}$ where $\Theta_0 \times \mathcal{S}_{T0}$ denotes a fixed big $\ell_1$-ball that contains $\Theta\times\mathcal{S}_T$. Since the $\varepsilon$-covering number of $\Theta_0 \times \mathcal{S}_{T0}$ is $O(\varepsilon^{-d-1})$, $\mathcal{N}(\varepsilon, \Theta \times \mathcal{S}_{T},\|\cdot\|_1)=O(\varepsilon^{-d-1})$. Similar to Result (1), we have $$
\mathcal{N}_{[]}(2 \varepsilon\|c_2\|_{L_1(Q)}, \mathcal{F}_2, L_1(Q)) \leq \mathcal{N}(\varepsilon, \Theta \times \mathcal{S}_{T},\|\cdot\|_1)\leq K_2 \varepsilon^{-d-1},
$$
for some $K_2>0$. Since $0<\|c_2 \|_{L_1(Q)}<+\infty$ by Condition A2, $\mathcal{N}_{[]}\left(\varepsilon, \mathcal{F}_1, L_1(Q)\right)<+\infty$ for every fixed $\varepsilon>0$. Hence, $\mathcal{F}_2$ is Q-Glivenko-Cantelli.

Result (3):
As $\mathcal{F}_3$ is indexed by $\bt \in \Theta$, applying
the same argument as in Result (1) we have $$
\mathcal{N}_{[]}(2 \varepsilon\|c_3\|_{L_1(P)}, \mathcal{F}_3, L_1(P)) \leq \mathcal{N}(\varepsilon, \Theta,\|\cdot\|_1)\leq K_1 \varepsilon^{-d}.
$$
 Since $0<\|c_3\|_{L_1(P)}<+\infty$ by Condition A3,
 $\mathcal{N}_{[]}\left(\varepsilon, \mathcal{F}_3,
   L_1(P)\right)<+\infty$ for every fixed $\varepsilon>0$. Hence $\mathcal{F}_3$ is Glivenko-Cantelli.

Result (5):
The proof follows the same line as that for Result (4).
First, $\log(\cdot)$ is a continuous mapping.
For every $x\in \mathcal{S}_T$ and $\bt \in \Theta$, $\delta | \log \{\int
q_{T|\Z}(x,\z;\bt)d Q_\Z(\z)\}| \leq \delta| \log\{ M_2(x)\} | + \delta| \log
\{m_2(x)\} |$ and $E_p[\Delta | \log\{ M_2(X)\} | +\Delta  | \log \{m_2(X)\} |]$ is
finite by Condition A3. Thus the class $\mathcal{F}_5$ has
integrable envelope function $\delta| \log\{ M_2(x)\} | + \delta| \log
\{m_2(x)\} |$. Again by \cite{van2000preservation}[Theorem 3] 
  and Result (3), $\mathcal{F}_5$ is also
P-Glivenko-Cantelli. 

Result (8): The proof follows the same line as that for Result (5).
The class $\mathcal{F}_8$ is a continuous transformation of
$\mathcal{F}_3$. For every $x \in \mathcal{S}_T$ and $\bt \in \Theta$,
$| 1/ \{\int q_{T|\Z}(x,\z;\bt)d Q_\Z(\z)\}^\delta| \leq 1/m_2(x)^\delta$ and
$E_p\{1/m_2(T)\}$ is finite by Condition A3. Thus $1/m_2(x)^\delta$ is an
integrable envelope function of the class $\mathcal{F}_8$. Then the
class $\mathcal{F}_8$ is P-Glivenko-Cantelli. 

Result (6): 
The class $\mathcal{F}_6$ is indexed by $\bt \in \Theta$. Then
applying the same derivation as in
Result (1), using the Lipschitz condition derived in Remark \ref{rem:c4} and \cite{van1996weak}[Theorem 2.7.11],
  for the $L_1(P)$-norm, we have $$
\mathcal{N}_{[]}(2 \varepsilon\|c_4\|_{L_1(P)}, \mathcal{F}_6, L_1(P))
\leq \mathcal{N}(\varepsilon, \Theta,\|\cdot\|_1)\leq K_1
\varepsilon^{-d}.
$$
By Condition A4, since $0<\|c_4\|_{L_1(P)}<+\infty$, $\mathcal{N}_{[]}\left(\varepsilon, \mathcal{F}_6,
  L_1(P)\right)<+\infty$ for every fixed $\varepsilon>0$. Hence
$\mathcal{F}_6$ is Glivenko-Cantelli.  

Result (7):
The proof follows the same line as that for Result (5). 
 For every $x \in \mathcal{S}_T$, $\z \in
\mathcal{S}_{\Z p}$ and $\bt \in \Theta$, $(1-\delta)\left| \log \int \frac{I(s>
    x) q_{T \mid \Z}(s, \z ; \bt)}{\int q_{T \mid \Z}(s, \z ; \bt) d
    Q_\Z(\z)} d P_T(s) \right| \leq
(1-\delta)|\log\{m_3(x,\z)\}|+(1-\delta)|\log\{M_3(x,\z)\}|$ and
$E_p[|\log\{m_3(C,\Z)\}|+|\log\{M_3(C,\Z)\}|]$ is finite by Condition
A4. Hence, $(1-\delta)|\log\{m_3(x,\z)\}|+(1-\delta)|\log\{M_3(x,\z)\}|$ is an integrable
envelope function of the class $\mathcal{F}_7$. The class
$\mathcal{F}_7$ is Glivenko-Cantelli by
\cite{van2000preservation}[Theorem 3]. 

Result (9): 
The proof follows the same line as that for Result (8).
The class $\mathcal{F}_9$ is a continuous transformation of the class
$\mathcal{F}_6$. For every $x \in \mathcal{S}_T$, $\z \in
\mathcal{S}_{\Z p}$ and $\bt \in \Theta$, $\left|{1}/{\int \frac{I(s>
x) q_{T \mid \Z}(s, \z ; \bt)}{\int q_{T \mid \Z}(s, \z ; \bt) d
      Q_\Z(\z)} d P_T(s)}\right| \leq 1/m_3(x,\z)$ and $E_p\{
1/m_3(C,\Z)\}$ is finite from Condition A4. Then $1/m_3(x,\z)^{(1-\delta)}$ is an
integrable envelope function of the class $\mathcal{F}_9$. Hence,
$\mathcal{F}_9$ is Glivenko-Cantelli~\citep[Theorem
3]{van2000preservation}.

\end{proof}

\begin{Lemma}\label{lem2.2}
Under A1-A5,  we have
     $$\sup_{\bt \in \Theta}|\ell_n (\bt)-E\{\ell(X, \Z, \Delta, R ; \bt)\}| \povr 0.$$
\end{Lemma}

\begin{proof}
First, we divide $\ell_n(\cdot)$ and $\ell(\cdot)$ into three parts
for ease of proof, using the following notation: 
    \bse
    \ell(x,\z, \delta,r; \bt)
    &=& \ell_1(x, \z ,\delta,r ; \bt)-\ell_2(x ,\delta,r; \bt)+\ell_3(x, \z ,\delta,r ; \bt)\\
&=&r\delta \log \{q_{T\mid \Z}(x,\z;\bt)\}
-r\delta
\log\int q_\Z(\z) q_{T\mid \Z}(x,\z;\bt)d\z
+
r(1-\delta)\log
\int_x^{\infty}\frac{ p_T(t)q_{T\mid \Z}(t,\z;\bt)}{
\int q_\Z(\z)q_{T\mid \Z}(t,\z;\bt)d\z}dt,
\ese
and
\bse
\ell_n(\bt) &=&  \ell_{1n}(\bt) - \ell_{2n}(\bt) + \ell_{3n}(\bt)\\
&=&
n_1^{-1}\sum_{i=1}^{n_1}
r_i\delta_i \log\{q_{T\mid \Z}(x_i,\z_i;\bt)\}
-n_1^{-1}\sum_{i=1}^{n_1} r_i \delta_i
\log \int q_{T\mid \Z}(x_i,\z;\bt)d \wh Q_{\Z}(\z)\\
&&+n_1^{-1}\sum_{i=1}^{n_1} r_i (1-\delta_i)\log
\left(\int \frac{ I(t>x_i)
  q_{T\mid \Z}(t,\z_i;\bt)}{
\int q_{T\mid \Z}(t,\z;\bt)d \wh Q_\Z(\z)}d \wh P_T(t)\right),
\ese
where $\wh P_T(t)$ denotes the Kaplan-Meier estimator of
$P_T(t)$.  
Then the objective function can be separated into three parts
\bse
&&\sup_{\bt \in \Theta}|\ell_n(\bt)-E\{\ell(X, \Z, \Delta, R ; \bt)\}| \\ & \leq & \sup _{\bt \in \Theta}[|\ell_{1n}(\bt)-E\{\ell_1(X, \Z, \Delta,R ; \bt)\}|+|\ell_{2n}(\bt)-E\{\ell_2(X, \Delta,R ; \bt)\}|+|\ell_{3n}(\bt)-E\{\ell_3(X,\Z,\Delta,R ; \bt)\}|].
\ese
To simplify the proof, we introduce the  notation $q_T(t;\bt) \equiv
\int q_\Z(\z)q_{T\mid \Z}(t,\z;\bt)d\z$ and $\wh q_T(t;\bt) \equiv
\int q_{T\mid \Z}(t,\z;\bt)d \wh Q_\Z(\z)$. 
From Lemma \ref{lm:pgc} Result (4), 
$\sup_{\bt \in \Theta}\left|\ell_{1n}(\bt)-E\left\{\ell_1(X, \Z, \Delta,R ; \bt)\right\}\right| \povr 0$.

Regarding $\ell_{2n}(\bt)$, it follows that 
\bse
&& \sup_{\bt \in \Theta}\left|\ell_{2n}(\bt)-E\left\{\ell_2(X, \Delta, R ; \bt)\right\}\right| \\
&\leq& \sup _{\bt \in \Theta}\left[\left|\ell_{2n}(\bt)-n_1^{-1}\sum_{i=1}^{n_1} \ell_2(X_i, \Delta_i, R_i ; \bt)\right|+ \left | n_1^{-1}\sum_{i=1}^{n_1} \ell_2(X_i, \Delta_i, R_i; \bt)-E\left\{\ell_2(X, \Delta, R ; \bt)\right\}\right|\right].
\ese
According to Result (5) in Lemma \ref{lm:pgc}, the second term goes to
zero, i.e.,
$$
\sup _{\bt \in \Theta}\left|n_1^{-1}\sum_{i=1}^{n_1} \ell_2(X_i, \Delta_i,
  R_i ; \bt)-E\left\{\ell_2(X, \Delta, R ; \bt)\right\}\right| \povr
0. 
$$
For the first term, we have that
\bse
&&\sup _{\bt \in \Theta}\left|\ell_{2n}(\bt)-n_1^{-1}\sum_{i=1}^{n_1}
  \ell_2(X_i, \Delta_i, R_i ; \bt)\right| \\ 
& = & \sup _{\bt \in \Theta} \left| n_1^{-1}\sum_{i=1}^{n_1} R_i
  \Delta_i\left\{\log \wh q_T(X_i;\bt)-\log q_T(X_i;\bt)\right\}
\right| \\ 
 & = & \sup_{\bt \in \Theta} \left| n_1^{-1}\sum_{i=1}^{n_1} \frac{\wh
     q_T(X_i;\bt)^{R_i \Delta_i}-q_T(X_i;\bt)^{R_i
       \Delta_i}}{q_T(X_i;\bt)^{R_i \Delta_i}+\xi_{1i} \left\{\wh
       q_T(X_i;\bt)^{R_i \Delta_i}-q_T(X_i;\bt)^{R_i
         \Delta_i}\right\}} \right| \\ 
& \leq &\sup_{\bt \in \Theta} \left|n_1^{-1}\sum_{i=1}^{n_1}\frac{q_T(X_i;\bt)^{R_i \Delta_i}}{q_T(X_i;\bt)^{R_i
      \Delta_i}+\xi_{1i} \left\{\wh q_T(X_i;\bt)^{R_i
        \Delta_i}-q_T(X_i;\bt)^{R_i
        \Delta_i}\right\}}\frac{1}{q_T(X_i;\bt)^{R_i \Delta_i}}
\right| \\ && \times \sup _{\bt \in \Theta, t \in
  \mathcal{S}_T}\left|\wh q_T(t;\bt)-q_T(t;\bt)\right|,  
\ese
where $\xi_{1i} \in (0,1)$ for $i = 1,\ldots,N$. The second equation
holds due to the mean-value theorem. From Result (2) of Lemma
  \ref{lm:pgc}, $$\sup _{\bt \in \Theta, t \in \mathcal{S}_T}\left|\wh
    q_T(t;\bt)-q_T(t;\bt)\right| \povr 0. $$ 
$E[|{1}/{q_T(X;\bt)^{R \Delta}} |] < +\infty$ by Condition A3.
  Condition A2,  Result (2) and Result (8)
  of Lemma \ref{lm:pgc} imply that $$
  \sup_{x,\bt,r,\delta,\xi_1}\left|\frac{q_T(x;\bt)^{r
        \delta}}{q_T(x;\bt)^{r \delta}+\xi_1 \left\{\wh q_T(x;\bt)^{r
          \delta}-q_T(x;\bt)^{r \delta}\right\}} \right|=O_p(1),$$ 
and $$\sup _{\bt \in \Theta} \left| n_1^{-1}\sum_{i=1}^{n_1}
  \frac{1}{q_T(X_i;\bt)^{R_i \Delta_i}} - E\left\{
    \frac{1}{q_T(X;\bt)^{R \Delta}} \right\} \right| \povr 0.$$ 
  Hence, $\sup _{\bt \in \Theta}\left|\ell_{2n}(\bt)-n_1^{-1}\sum_{i=1}^{n_1}\ell_2(X_i, \Delta_i, R_i ;
    \bt)\right| \povr 0$.

Lastly, we proceed to prove $\sup _{\bt \in \Theta}\left|\ell_{3
    n}(\bt)-E\left\{\ell_3(X,\Z, \Delta,R ; \bt)\right\}\right| \povr
0$, 
\bse
&&\sup _{\bt \in \Theta}\left|\ell_{3 n}(\bt)-E\left\{\ell_3(X,\Z, \Delta,R ; \bt)\right\}\right| \\
& \leq & \sup _{\bt \in \Theta}  \left|  n_1^{-1}\sum_{i=1}^{n_1} R_i (1-\Delta_i) \left\{\log
\int_{t>X_i}\frac{q_{T\mid \Z}(t,\Z_i;\bt)}{
\wh q_T(t;\bt)}d \wh P_T(t) -\log
\int_{t>X_i}\frac{q_{T\mid \Z}(t,\Z_i;\bt)}{
q_T(t;\bt)}dP_T(t) \right\} \right |\\
&& +\sup_{\bt \in \Theta} \left| n_1^{-1}\sum_{i=1}^{n_1} R_i (1-\Delta_i)\log
\int_{t>X_i}\frac{q_{T\mid \Z}(t,\Z_i;\bt)}{
q_T(t;\bt)}dP_T(t) -  E \left\{ R (1-\Delta)\log
\int_{t>X}\frac{q_{T\mid \Z}(t,\Z;\bt)}{
q_T(t;\bt)}dP_T(t) \right\} \right|.
\ese  
The second term goes to 0 in probability according to Result (7) in
Lemma \ref{lm:pgc}.  The first term satisfies
\be
\notag && \sup _{\bt \in \Theta}  \left|  n_1^{-1}\sum_{i=1}^{n_1} R_i (1-\Delta_i)\left \{ \log
\int_{t>X_i}\frac{q_{T\mid \Z}(t,\Z_i;\bt)}{
\wh q_T(t;\bt)}d \wh P_T(t) - \log
\int_{t>X_i}\frac{q_{T\mid \Z}(t,\Z_i;\bt)}{
q_T(t;\bt)}dP_T(t) \right \} \right |\\ 
\notag &\leq & \sup_{\bt \in \Theta} \left| n_1^{-1}\sum_{i=1}^{n_1} \frac{ \left\{\int_{t>X_i}\frac{q_{T\mid \Z}(t,\Z_i;\bt)}{
q_T(t;\bt)}dP_T(t) \right\}^{R_i(1-\Delta_i)}}{
\xi_{2i} \left\{ \int_{t>X_i}\frac{q_{T\mid \Z}(t,\Z_i;\bt)}{
 \wh q_T(t;\bt)}d \wh P_T(t) \right\}^{R_i(1-\Delta_i)} + (1-\xi_{2i}) \left\{
\int_{t>X_i}\frac{q_{T\mid \Z}(t,\Z_i;\bt)}{   
q_T(t;\bt)}dP_T(t)\right\}^{R_i(1-\Delta_i)}} \right. \\ 
\notag && \left. \times \frac{1}{\left\{\int_{t>X_i}\frac{q_{T\mid \Z}(t,\Z_i;\bt)}{
q_T(t;\bt)}dP_T(t)\right\}^{R_i(1-\Delta_i)}} \right|  \times \sup _{\bt, x, \z} \left| 
\int_{t>x}\frac{q_{T\mid \Z}(t,\z;\bt)}{
\wh q_T(t;\bt)}d \wh P_T(t)-
\int_{t>x}\frac{q_{T\mid \Z}(t,\z;\bt)}{
q_T(t;\bt)}dP_T(t) \right |\\ \label{eq2.3}
\ee
where $\xi_{2i} \in (0,1)$ for $i = 1,\ldots,N$. Now,
\be
\notag &&\sup_{\bt, x, \z} \left| 
\int_{t>x}\frac{q_{T\mid \Z}(t,\z;\bt)}{
\wh q_T(t;\bt)}d \wh P_T(t)-
\int_{t>x}\frac{q_{T\mid \Z}(t,\z;\bt)}{
q_T(t;\bt)}dP_T(t) \right |\\
&\leq& \sup _{\bt, x, \z} \left[ \left| 
\int_{t>x} \left \{ \frac{q_{T\mid \Z}(t,\z;\bt)}{
\wh q_T(t;\bt)}-\frac{q_{T\mid \Z}(t,\z;\bt)}{
q_T(t;\bt)} \right\} d \wh P_T(t) \right | + \left | \int_{t>x} \frac{ q_{T\mid \Z}(t,\z;\bt)}{
q_T(t;\bt)}d \{\wh P_T(t) - P_T(t) \} \right | \right]. \label{eq2.4}
\ee

For the first term, we have
\bse
&& \sup _{\bt, x, \z} \left| 
\int_{t>x}\left\{\frac{q_{T\mid \Z}(t,\z;\bt)}{
\wh q_T(t;\bt)}-\frac{q_{T\mid \Z}(t,\z;\bt)}{
q_T(t;\bt)} \right\} d   \wh P_T(t)\right | \\
&\leq &  
\int\sup _{\bt,\z} \left| \frac{ q_{T\mid \Z}(t,\z;\bt)}{
q_T(t;\bt)\wh q_T(t;\bt)} \{\wh q_T(t;\bt) -q_T(t;\bt) \} \right | d\wh P_T(t)\\
&\leq &   
\sup_{\bt,t}|\wh q_T(t;\bt) -q_T(t;\bt)| \times \int \sup_{\bt, \z} \left| \frac{ q_{T\mid \Z}(t,\z;\bt)}{
q_T(t;\bt)\wh q_T(t;\bt)}  \right | d \wh P_T(t)\\
&\leq & \sup_{\bt,t}|\wh q_T(t;\bt) -q_T(t;\bt)| \times
\sup_{\bt,t}\left|\frac{q_T(t;\bt)}{\wh q_T(t;\bt)}\right| \times \int
 M_4(t)   d \wh P_T(t).
\ese
The last inequality holds from
Conditions A3 and A5. By Conditions A3 and A6 and applying Theorem 1.1
of \cite{stute1993strong}, $\int M_4(t)d\wh P_T(t)$ goes to
$\int M_4(t) d P_T(t)$ in probability as $n\rightarrow \infty$,
where $\int M_4(t) d P_T(t)<+\infty$.
By Result (2) in Lemma
\ref{lm:pgc}, $\sup_{\bt,t}|\wh q_T(t;\bt) -q_T(t;\bt)|=o_p(1)$ and
$\sup_{\bt,t}\left|{q_T(t;\bt)}/{\wh q_T(t;\bt)}\right|=O_p(1)$.
Hence, we have $$\sup _{\bt, x, \z} \left| 
\int_{t>x}\left \{ \frac{q_{T\mid \Z}(t,\z;\bt)}{
\wh q_T(t;\bt)}-\frac{q_{T\mid \Z}(t,\z;\bt)}{
q_T(t;\bt)}\right\} d  \wh P_T(t) \right | \povr 0.$$

  From Lemma~C.1, the second term in (\ref{eq2.4}) satisfies
  $$\sup_{\bt,x,\z}\left | \int_{t>x} \frac{ q_{T\mid
        \Z}(t,\z;\bt)}{q_T(t;\bt)}d \{\wh P_T(t) - P_T(t) \} \right |
  \povr 0.$$
Thus, the last multiplier in \eqref{eq2.3} goes to zero.
To handle the first multiplier in (\ref{eq2.3}), 
by Result (6) and Result (9) of Lemma \ref{lm:pgc}, we have $$  \sup_{x,\z,\bt,\delta,r,\xi_2}\left|\frac{ \left\{\int_{t>x}\frac{q_{T\mid \Z}(t,\z;\bt)}{
q_T(t;\bt)}dP_T(t) \right\}^{r(1-\delta)}}{
\xi_{2} \left\{ \int_{t>x}\frac{q_{T\mid \Z}(t,\z;\bt)}{
 \wh q_T(t;\bt)}d \wh P_T(t) \right\}^{r(1-\delta)} + (1-\xi_{2}) \left\{
\int_{t>x}\frac{q_{T\mid \Z}(t,\z;\bt)}{
q_T(t;\bt)}dP_T(t)\right\}^{r(1-\delta)}} \right|= O_p(1),$$ and
$$
\sup_{\bt \in \Theta} \left|n_1^{-1}\sum_{i=1}^{n_1}  \frac{1}{\left\{\int_{t>X_i}\frac{q_{T\mid \Z}(t,\Z_i;\bt)}{
q_T(t;\bt)}dP_T(t)\right\}^{R_i(1-\Delta_i)}} - E\left[ \frac{1}{\left\{\int_{t>X}\frac{q_{T\mid \Z}(t,\Z;\bt)}{
q_T(t;\bt)}dP_T(t)\right\}^{R_i(1-\Delta_i)}} \right] \right| \povr 0,
$$
and 
$$E\left[ \frac{1}{\{\int_{t>X}\frac{q_{T\mid \Z}(t,\Z;\bt)}{
q_T(t;\bt)}dP_T(t)\}^{R(1-\Delta)}} \right]$$ is finite by Condition
A4, hence the first multiplier is bounded.
Hence, $\sup_{\bt \in \Theta}|\ell_{3 n}(\bt)-E\{\ell_3(X,\Z,\Delta,R ; \bt)\}| \povr 0$.
Therefore, we have $\sup_{\bt \in \Theta}|\ell_n (\bt)-E\{\ell(X, \Z, \Delta ; \bt)\}| \povr 0$
\end{proof}

  For any function $g(\cdot, h)$, we write its $k$th Gateaux derivative with respect to $h$ at $h_1$ in the direction $h_2 - h_1$ as
$$
\left.\frac{\partial^k g\left(\cdot, h_1\right)}{\partial h^k}\left[h_2-h_1\right] \equiv \frac{\partial^k g\{\cdot, h_1+\varepsilon(h_2 - h_1)\}}{\partial \varepsilon^k}\right|_{\varepsilon=0}.
$$

\begin{Lemma}
\label{lm:2.4}

Under Conditions \ref{A2} and \ref{B2}, it follows that
\begin{enumerate}[label=(\arabic*)]
     \item the class $\mathcal{F}_{10} \equiv \left\{\z \mapsto q_{T
           \mid \Z}(t, \z ; \bt_0): t \in \mathcal{S}_T\right\}$ is
       Q-Donsker, hence $\|\wh q_T(t) - q_T(t)\|_\infty=o_p(n_2^{-1/4})$,
        \item the class $\mathcal{F}_{11} \equiv \left\{\z \mapsto
            \tfrac{\partial}{\partial \bt}q_{T \mid \Z}(t, \z ; \bt_0):
            t \in \mathcal{S}_T\right\}$ is Q-Donsker, hence
             $\|\wh q^*_T(t) -
            q^*_T(t)\|_\infty=o_p(n_2^{-1/4})$.

\end{enumerate}
\end{Lemma}

\begin{proof}

The class $\mathcal{F}_{10}$ is indexed by $t \in \mathcal{S}_T
\subset \mathcal{S}_{T0}$ where $\mathcal{S}_{T0}$ denotes a fixed big
$\ell_1$-ball that contains $\mathcal{S}_T$. Since the
$\varepsilon$-covering number of $ \mathcal{S}_{T0}$ is
$O(\varepsilon^{-1})$, $\mathcal{N}(\varepsilon,
\mathcal{S}_{T},\|\cdot\|_1)=O(\varepsilon^{-1})$. Under Condition
\ref{A2}, by \cite{van1996weak}[Theorem 2.7.11], it follows that $$ 
\mathcal{N}_{[]}(2 \|c_2 \|_{L_2(Q)} \varepsilon , \mathcal{F}_{10},
L_2(Q)) \leq \mathcal{N}(\varepsilon, \mathcal{S}_{T},\|\cdot\|_1)\leq
K_3 \varepsilon^{-1},
$$
for some $K_3>0$. Since $0<\|c_2 \|_{L_2(Q)}<+\infty$ by Condition
\ref{A2}, $\int_0^\infty \sqrt{\log \mathcal{N}_{[]}\left(\varepsilon,
    \mathcal{F}_{10}, L_2(Q)\right)} d \varepsilon<+\infty$. Hence,
$\mathcal{F}_{10}$ is Q-Donsker. By the definition of the Donsker class from \cite{van1996weak}, the empirical process $\sqrt{n_2} (\mathbb{Q}_n - Q)$ weakly converges to a tight Borel measurable element in $\ell^\infty(\mathcal{F}_{10})$, where the space $\ell^{\infty}(\mathcal{F}_{10})$ is defined as the set of all uniformly bounded, real functions on $\mathcal{F}_{10}$. Hence, we have $\|\wh q_T(t) - q_T(t)\|_\infty=O_p(n_2^{-1/2})=o_p(n_2^{-1/4})$.

The class $\mathcal{F}_{11}$ is indexed by $t \in \mathcal{S}_T
\subset \mathcal{S}_{T0}$ where $\mathcal{S}_{T0}$ denotes a fixed big
$\ell_1$-ball that contains $\mathcal{S}_T$. Since the
$\varepsilon$-covering number of $ \mathcal{S}_{T0}$ is
$O(\varepsilon^{-1})$, $\mathcal{N}(\varepsilon,
\mathcal{S}_{T},\|\cdot\|_1)=O(\varepsilon^{-1})$. From
\cite{van1996weak}[Theorem 2.7.11], we have $$ 
\mathcal{N}_{[]}(2 \|c_7 \|_{L_2(Q)}  \varepsilon, \mathcal{F}_{11},
L_2(Q)) \leq \mathcal{N}(\varepsilon, \mathcal{S}_{T},\|\cdot\|_1)\leq
K_3 \varepsilon^{-1}, 
$$
for some $K_3>0$. Since $0<\|c_7 \|_{L_2(Q)} <+\infty$ by Condition
\ref{B2}, $\int_0^\infty \sqrt{\log \mathcal{N}_{[]}\left(\varepsilon,
    \mathcal{F}_{11}, L_2(Q)\right)} d \varepsilon<+\infty$. Then,
$\mathcal{F}_{11}$ is also Q-Donsker. The empirical process $\sqrt{n_2} (\mathbb{Q}_n - Q)$ weakly converges to a tight Borel measurable element in $\ell^\infty(\mathcal{F}_{11})$. Hence, we also have $\|\wh q^*_T(t) -
            q^*_T(t)\|_\infty=O_p(n_2^{-1/2})=o_p(n_2^{-1/4})$.
\end{proof}

\begin{Lemma}
\label{lm:2.5}
Under Conditions \ref{B1} and \ref{B2}, the following holds:    
\begin{enumerate}[label=(\arabic*)]
      
        \item the class $\mathcal{F}_{12} \equiv \{t \mapsto q_{T \mid \Z}(t, \z ; \bt_0)/q_T(t): \z \in \mathcal{S}_{\Z p}\}$ satisfies
        $$
\mathcal{N}_{[]}\left(\varepsilon, \mathcal{F}_{12}, L_1(P)\right)<\infty, \quad \text { for every } \varepsilon>0,
$$
and has a measurable envelope function $M_5$ such that $$
\int_0^{\infty} \sup _{\wt P} \sqrt{\log \mathcal{N}\left(\varepsilon\|M_5\|_{L_2(\wt P)}, \mathcal{F}_{12}, L_2(\wt P)\right)} d \varepsilon<\infty,
$$
where the $\sup$ is taken over all the probability measure $\wt P$ on $\mathcal{S}_T$ such that $\|M_5\|_{L_2(\wt P)}<\infty$.

        \item the class $\mathcal{F}_{13} \equiv \left\{t  \mapsto \tfrac{\partial}{\partial \bt}q_{T \mid \Z}(t, \z ; \bt_0)/q_T(t): \z \in \mathcal{S}_{\Z p}\right\}$ satisfies
        $$
\mathcal{N}_{[]}\left(\varepsilon, \mathcal{F}_{13}, L_1(P)\right)<\infty, \quad \text { for every } \varepsilon>0,
$$
and has a measurable envelope function $M_6$ such that $$
\int_0^{\infty} \sup_{\wt P} \sqrt{\log \mathcal{N}\left(\varepsilon\|M_6\|_{L_2(\wt P)}, \mathcal{F}_{13}, L_2(\wt P)\right)} d \varepsilon<\infty,
$$
where the $\sup$ is taken over all the probability measure $\wt P$ on $\mathcal{S}_T$ such that $\|M_6\|_{L_2(\wt P)}<\infty$.

\item the class $\mathcal{F}_{14}\equiv\{t\mapsto \int q_{T\mid \Z}(t,\z;\bt_0)d \wt Q_{\Z}(\z)/q_T(t): \text{any empirical distribution function } \wt Q_\Z \text{ on } \mathcal{S}_{\Z q}\\ \text{ s.t. } \sup_{\z \in \mathcal{S}_{\Z q}}|\wt Q_\Z(\z)-Q_\Z(\z)|\leq 1/4\}$, satisfies
        $$
\mathcal{N}_{[]}\left(\varepsilon, \mathcal{F}_{14}, L_1(P)\right)<\infty, \quad \text { for every } \varepsilon>0,
$$
and has a measurable envelope function $M_7$ such that $$
\int_0^{\infty} \sup _{\wt P} \sqrt{\log \mathcal{N}\left(\varepsilon\|M_7\|_{L_2(\wt P)}, \mathcal{F}_{14}, L_2(\wt P)\right)} d \varepsilon<\infty,
$$
where the $\sup$ is taken over all the probability measure $\wt P$ on $\mathcal{S}_T$ such that $\|M_7\|_{L_2(\wt P)}<\infty$. 

\item the class $\mathcal{F}_{15}\equiv\{t\mapsto \int \tfrac{\partial}{\partial \bt}q_{T\mid \Z}(t,\z;\bt_0)d \wt Q_{\Z}(\z)/q_T(t): \text{any empirical distribution function } \wt Q_\Z \text{ on } \mathcal{S}_{\Z q}\\ \text{ s.t. } \sup_{\z \in \mathcal{S}_{\Z q}}|\wt Q_\Z(\z)-Q_\Z(\z)|\leq 1/4 \}$, satisfies
        $$
\mathcal{N}_{[]}\left(\varepsilon, \mathcal{F}_{15}, L_1(P)\right)<\infty, \quad \text { for every } \varepsilon>0,
$$
and has a measurable envelope function $M_8$ such that $$
\int_0^{\infty} \sup _{\wt P} \sqrt{\log \mathcal{N}\left(\varepsilon\|M_8\|_{L_2(\wt P)}, \mathcal{F}_{15}, L_2(\wt P)\right)} d \varepsilon<\infty,
$$
where the $\sup$ is taken over all the probability measure $\wt P$ on
$\mathcal{S}_T$ such that $\|M_8\|_{L_2(\wt P)}<\infty$. 

\end{enumerate}

\end{Lemma}

\begin{proof}

Result (1): The class $\mathcal{F}_{12}$ is indexed by $\z \in \mathcal{S}_{\Z p} \subset \mathcal{S}_{\Z p 0}$ where $\mathcal{S}_{\Z p 0}$ denotes a fixed big $\ell_1$-ball that contains $\mathcal{S}_{\Z p}$. Since the $\varepsilon$-covering number of $\mathcal{S}_{\Z p 0}$ is $O(\varepsilon^{-d_\z})$, $\mathcal{N}(\varepsilon,  \mathcal{S}_{\Z p},\|\cdot\|_1)=O(\varepsilon^{-d_\z})$.  From
\cite{van1996weak}[Theorem 2.7.11], we have for any norm $\|\cdot \|$ $$
\mathcal{N}_{[]}(2 \varepsilon \|c_5/q_T\|, \mathcal{F}_{12}, \|\cdot \|) \leq \mathcal{N}(\varepsilon,  \mathcal{S}_{\Z p},\|\cdot\|_1)\leq K_4 \varepsilon^{-d_\z},
$$
for some $K_4>0$. Since $0<\|c_5/q_T\|_{ L_1(P)}<+\infty$ by Condition
B1, $\mathcal{N}_{[]}\left(\varepsilon, \mathcal{F}_{12}, L_1(P)\right)<\infty$ for every $\varepsilon>0$. 

Note that for any norm $\|\cdot\|$, $\mathcal{N}( \varepsilon \|c_5/q_T\|, \mathcal{F}_{12}, \|\cdot \|) \leq \mathcal{N}_{[]}(2 \varepsilon \|c_5/q_T\|, \mathcal{F}_{12}, \|\cdot \|)$ and for any probability measure $\wt P$, $0<\|M_5\|_{L_2(\wt P)}<+\infty$ is equivalent to $0<\|c_5/q_T\|_{L_2(\wt P)}<+\infty$. Hence, $$
\int_0^{\infty} \sup _{\wt P} \sqrt{\log \mathcal{N}\left(\varepsilon\|M_5\|_{L_2(\wt P)}, \mathcal{F}_{12}, L_2(\wt P)\right)} d \varepsilon<\infty,
$$
where the $\sup$ is taken over all the probability measure $\wt P$ on
$\mathcal{S}_T$ such that $\|M_5\|_{L_2(\wt P)}<\infty$. 

Result (2): This proof follows the same steps as in Result (1). The class $\mathcal{F}_{13}$ is indexed by $\z \in \mathcal{S}_{\Z p}$. By Condition \ref{B2}, for any norm $\|\cdot \|$ $$
\mathcal{N}_{[]}(2 \varepsilon \|c_6/q_T\|, \mathcal{F}_{13}, \|\cdot \|) \leq \mathcal{N}(\varepsilon, \mathcal{S}_{\Z p},\|\cdot\|_1)\leq K_4 \varepsilon^{-d_\z},
$$
for some $K_4>0$. Since $0<\|c_6/q_T\|_ {L_1(P)}<+\infty$ by Condition \ref{B2}, $\mathcal{N}_{[]}\left(\varepsilon, \mathcal{F}_{13}, L_1(P)\right)<\infty$ for every $\varepsilon>0$. 

Note that for any norm $\|\cdot\|$, $\mathcal{N}( \varepsilon \|c_6/q_T\|, \mathcal{F}_{13}, \|\cdot \|) \leq \mathcal{N}_{[]}(2 \varepsilon \|c_6/q_T\|, \mathcal{F}_{13}, \|\cdot \|)$ and for any probability measure $\wt P$, $0<\|M_6\|_{L_2(\wt P)}<+\infty$ is equivalent to $0<\|c_6/q_T\|_{L_2(\wt P)}<+\infty$. Hence, $$
\int_0^{\infty} \sup _Q \sqrt{\log \mathcal{N}\left(\varepsilon\|M_6\|_{L_2(\wt P)}, \mathcal{F}_{13}, L_2(\wt P)\right)} d \varepsilon<\infty,
$$
where the $\sup$ is taken over all the probability measure $\wt P$ on $\mathcal{S}_T$ such that $\|M_6\|_{L_2(\wt P)}<\infty$.

Result (3):
First we prove $\mathcal{N}_{[]}(\varepsilon , \mathcal{F}_{14}, \|\cdot \|_{L_1(P)})<\infty$ for every $\varepsilon$.
By Condition \ref{B1}, for every $(\wt \z, \z) \in  (\mathcal{S}_{\Z p}\cup\mathcal{S}_{\Z q})^2$
$$|q_{T\mid\Z}(t,\z;\bt_0) - q_{T\mid\Z}(t,\wt \z;\bt_0)| \leq \|\z - \wt \z\|_1 c_5(t),$$
    with $0<E\{c_5(T)/q_T(T)\}< +\infty$. Let $\mathcal{F}_{16} \equiv \{t \mapsto q_{T \mid \Z}(t, \z ; \bt_0)/q_T(t): \z \in \mathcal{S}_{\Z q}\}$. From
\cite{van1996weak}[Theorem 2.7.11], we have for any norm $\|\cdot \|$ $$
\mathcal{N}_{[]}(2 \varepsilon \|c_5/q_T\|, \mathcal{F}_{16}, \|\cdot \|) \leq \mathcal{N}(\varepsilon,  \mathcal{S}_{\Z q},\|\cdot\|_1)\leq K_5  \varepsilon^{-d_\z},
$$
for some $K_5>0$. Since $0<\|c_5/q_T\|_{ L_1(P)}<+\infty$,
$\mathcal{N}_{[]}\left(\varepsilon, \mathcal{F}_{16},
  L_1(P)\right)<\infty$ for every $\varepsilon>0$.
Let $M_{\mathcal{F}_{16}}(t) \equiv \max\{1,\operatorname{diam} \mathcal{S}_{\Z
  q}\}c_5(t)/q_T(t)+q_{T \mid \Z}(t, \z_0 ;
\bt_0)/q_T(t)$
for some fixed $\z_0 \in \mathcal{S}_{\Z q}$. Then
$M_{\mathcal{F}_{16}}(t)$ is an envelope function of
$\mathcal{F}_{16}$. Since $0<\|c_5/q_T\|_{ L_1(P)}<+\infty$ by
Condition B1, $0<\|M_{\mathcal{F}_{16}}\|_{ L_1(P)}< +\infty$.

Following the proof of \cite{van1996weak}[Theorem 2.7.11], let
$\wt \z_1,\ldots,\wt \z_{K_q}$,
where $K_q\equiv\mathcal{N}(\varepsilon,
\mathcal{S}_{\Z q},\|\cdot\|_1)$, be an $\varepsilon$-net under
$\|\cdot\|_1$ for $\mathcal{S}_{\Z q}$. Hence, the compact set
$\mathcal{S}_{\Z q}$ can be  partitioned into $K_q$ disjoint subsets and each subset 
contains a single $\wt \z_i$.
These subsets are denoted as $\mathcal{S}_{\Z qi},
i = 1,\ldots, K_q$ and for any $\z \in \mathcal{S}_{\Z qi}$,
$\|\z-\wt \z_i\|_1\leq \varepsilon$. Then the brackets to cover
$\mathcal{F}_{16}$ can be constructed
as $$\left[\frac{q_{T\mid\Z}(t,\wt \z_i;\bt_0)}{q_T(t)}-\varepsilon
  \frac{c_5(t)}{q_T(t)},
  \frac{q_{T\mid\Z}(t,\wt \z_i;\bt_0)}{q_T(t)}+\varepsilon
  \frac{c_5(t)}{q_T(t)}\right], i=1,\ldots,K_q,$$ 
such that, for every $i$ and any $\z \in \mathcal{S}_{\Z qi}$, we have $$ \frac{q_{T\mid\Z}(t,\wt \z_i;\bt_0)}{q_T(t)}-\varepsilon \frac{c_5(t)}{q_T(t)} \leq \frac{q_{T \mid \Z}(t,\z;\bt_0)}{q_T(t)} \leq \frac{q_{T\mid\Z}(t,\wt \z_i;\bt_0)}{q_T(t)}+\varepsilon \frac{c_5(t)}{q_T(t)}.$$

For any  empirical distribution function $\wt Q_\Z(\z) \in \{\wt Q_\Z(\z): \sup_\z |\wt Q_\Z(\z)- Q_\Z(\z)|\leq 1/4\}$, we have 
\bse
 \int \frac{q_{T\mid\Z}(t,\z;\bt_0)}{q_T(t)}d \wt Q_\Z(\z) & \leq & \sum_{i=1}^{K_q} \int  I\{\z \in \mathcal{S}_{\Z qi}\}d \wt Q_\Z(\z) \left\{ \frac{q_{T\mid\Z}(t,\wt \z_i;\bt_0)}{q_T(t)}+\varepsilon \frac{c_5(t)}{q_T(t)}\right\} \\
&=& \sum_{i=1}^{K_q} \int  I\{\z \in \mathcal{S}_{\Z qi}\}d \wt Q_\Z(\z) \frac{q_{T\mid\Z}(t,\wt \z_i;\bt_0)}{q_T(t)}+\varepsilon \frac{c_5(t)}{q_T(t)}\\
&=& \sum_{i=1}^{K_q} \int  I\{\z \in \mathcal{S}_{\Z qi}\}d \{\wt Q_\Z(\z) - Q_\Z(\z) \}\frac{q_{T\mid\Z}(t,\wt \z_i;\bt_0)}{q_T(t)}\\
&& +\sum_{i=1}^{K_q} \int  I\{\z \in \mathcal{S}_{\Z qi}\}d Q_\Z(\z) \frac{q_{T\mid\Z}(t,\wt \z_i;\bt_0)}{q_T(t)}+\varepsilon \frac{c_5(t)}{q_T(t)},
\ese 
where $\int  I\{\z \in \mathcal{S}_{\Z qi}\}d \{\wt
Q_\Z(\z)-
Q_\Z(\z)\}$ lies in an interval $\mathcal{S}_{Q0}\equiv[-1/4,1/4]$ for
every $i \in \{1,\ldots,K_q\}$ since $\sup_{\z
  \in \mathcal{S}_{\Z q}} |\wt Q_\Z(\z)-Q_\Z(\z)|\leq
1/4$.

Similarly, we can obtain \bse
 \int \frac{q_{T\mid\Z}(t,\z;\bt_0)}{q_T(t)}d \wt Q_\Z(\z) & \geq &\sum_{i=1}^{K_q} \int  I\{\z \in \mathcal{S}_{\Z qi}\}d \{\wt Q_\Z(\z) - Q_\Z(\z) \}\frac{q_{T\mid\Z}(t,\wt \z_i;\bt_0)}{q_T(t)}\\
&&+\sum_{i=1}^{K_q} \int  I\{\z \in \mathcal{S}_{\Z qi}\}d Q_\Z(\z) \frac{q_{T\mid\Z}(t,\wt \z_i;\bt_0)}{q_T(t)}-\varepsilon \frac{c_5(t)}{q_T(t)}.
\ese 
    
Since $\mathcal{N}_{[]}(\varepsilon^{d_\z+1},\mathcal{S}_{Q0},\|\cdot\|_1)=O(\varepsilon^{-(d_\z+1)})$, we can have $\varepsilon^{d_\z+1}$-brackets for $\mathcal{S}_{Q0}$ denoted as $[C^L_{j},C^U_{j}], j=1,\ldots,\mathcal{N}_{[]}(\varepsilon^{d_\z+1},\mathcal{S}_{Q0},\|\cdot\|_1)$. Then, we have for every $\wt Q_\Z(\z)$  and for every $i=1,\ldots,K_q$, there exist some $j_i \in \{1,\ldots,\mathcal{N}_{[]}(\varepsilon^{d_\z+1},\mathcal{S}_{Q0},\|\cdot\|_1)\}$ such that \bse
 \int \frac{q_{T\mid\Z}(t,\z;\bt_0)}{q_T(t)}d \wt Q_\Z(\z) &\leq&\sum_{i=1}^{K_q} \int  I\{\z \in \mathcal{S}_{\Z qi}\}d \{\wt Q_\Z(\z) - Q_\Z(\z) \}\frac{q_{T\mid\Z}(t,\wt \z_i;\bt_0)}{q_T(t)}\\
&&+\sum_{i=1}^{K_q} \int  I\{\z \in \mathcal{S}_{\Z qi}\}d Q_\Z(\z) \frac{q_{T\mid\Z}(t,\wt \z_i;\bt_0)}{q_T(t)}+\varepsilon \frac{c_5(t)}{q_T(t)}\\
&\leq&\sum_{i=1}^{K_q} C^U_{j_i}\frac{q_{T\mid\Z}(t,\wt \z_i;\bt_0)}{q_T(t)}+\sum_{i=1}^{K_q} \int  I\{\z \in \mathcal{S}_{\Z qi}\}d Q_\Z(\z) \frac{q_{T\mid\Z}(t,\wt \z_i;\bt_0)}{q_T(t)}+\varepsilon \frac{c_5(t)}{q_T(t)},
\ese
and similarly
\bse
 \int \frac{q_{T\mid\Z}(t,\z;\bt_0)}{q_T(t)}d \wt Q_\Z(\z) &\geq&\sum_{i=1}^{K_q} \int  I\{\z \in \mathcal{S}_{\Z qi}\}d \{\wt Q_\Z(\z) - Q_\Z(\z) \}\frac{q_{T\mid\Z}(t,\wt \z_i;\bt_0)}{q_T(t)}\\
&&+\sum_{i=1}^{K_q} \int  I\{\z \in \mathcal{S}_{\Z qi}\}d Q_\Z(\z) \frac{q_{T\mid\Z}(t,\wt \z_i;\bt_0)}{q_T(t)}-\varepsilon \frac{c_5(t)}{q_T(t)}\\
&\geq&\sum_{i=1}^{K_q} C^L_{j_i}\frac{q_{T\mid\Z}(t,\wt \z_i;\bt_0)}{q_T(t)}+\sum_{i=1}^{K_q} \int  I\{\z \in \mathcal{S}_{\Z qi}\}d Q_\Z(\z) \frac{q_{T\mid\Z}(t,\wt \z_i;\bt_0)}{q_T(t)}-\varepsilon \frac{c_5(t)}{q_T(t)}.
\ese 
Then the brackets for the class $\mathcal{F}_{14}=\{t\mapsto \int
q_{T\mid \Z}(t,\z;\bt_0)d \wt Q_{\Z}(\z)/q_T(t): \wt Q_\Z \text{
  s.t. } \sup_{\z \in \mathcal{S}_{\Z q}}|\wt Q_\Z(\z)-Q_\Z(\z)|\leq
1/4\}$ can be constructed as
\bse
&&\left[\sum_{i=1}^{K_q} C^L_{j_i}\frac{q_{T\mid\Z}(t,\wt \z_i;\bt_0)}{q_T(t)}+\sum_{i=1}^{K_q} \int  I\{\z \in \mathcal{S}_{\Z qi}\}d Q_\Z(\z) \frac{q_{T\mid\Z}(t,\wt \z_i;\bt_0)}{q_T(t)}-\varepsilon \frac{c_5(t)}{q_T(t)},\right.\\&&\left.\sum_{i=1}^{K_q} C^U_{j_i}\frac{q_{T\mid\Z}(t,\wt \z_i;\bt_0)}{q_T(t)}+\sum_{i=1}^{K_q} \int  I\{\z \in \mathcal{S}_{\Z qi}\}d Q_\Z(\z) \frac{q_{T\mid\Z}(t,\wt \z_i;\bt_0)}{q_T(t)}+\varepsilon \frac{c_5(t)}{q_T(t)}\right],
\ese 
with
\bse
&&\left\|\left\{\sum_{i=1}^{K_q} C^U_{j_i}\frac{q_{T\mid\Z}(t,\wt \z_i;\bt_0)}{q_T(t)}+\sum_{i=1}^{K_q} \int  I\{\z \in \mathcal{S}_{\Z qi}\}d Q_\Z(\z) \frac{q_{T\mid\Z}(t,\wt \z_i;\bt_0)}{q_T(t)}+\varepsilon \frac{c_5(t)}{q_T(t)}\right\}\right.\\&& -\left.\left\{\sum_{i=1}^{K_q} C^L_{j_i}\frac{q_{T\mid\Z}(t,\wt \z_i;\bt_0)}{q_T(t)}+\sum_{i=1}^{K_q} \int  I\{\z \in \mathcal{S}_{\Z qi}\}d Q_\Z(\z) \frac{q_{T\mid\Z}(t,\wt \z_i;\bt_0)}{q_T(t)}-\varepsilon \frac{c_5(t)}{q_T(t)}\right\}\right\|_{L_1(P)}\\
&=& \left\|\sum_{i=1}^{K_q} (C^U_{j_i}-C^L_{j_i}) \frac{q_{T\mid\Z}(t,\wt\z_i;\bt_0)}{q_T(t)}+2\varepsilon\frac{c_5(t)}{q_T(t)} \right\|_{L_1(P)}\\
&\leq&\left\|\sum_{i=1}^{K_q} (C^U_{j_i}-C^L_{j_i}) M_{\mathcal{F}_{16}}+2\varepsilon M_{\mathcal{F}_{16}} \right\|_{L_1(P)}\\
&=& \|M_{\mathcal{F}_{16}}\|_{L_1(P)} \left|\sum_{i=1}^{K_q}
  (C^U_{j_i}-C^L_{j_i}) +2\varepsilon \right|\\
&\leq& \left\|M_{\mathcal{F}_{16}}\right\|_{L_1(P)}\left(K_q \varepsilon^{d_\z+1} + 2 \varepsilon \right)\\
&\leq& (K_5+2) \varepsilon \left\|M_{\mathcal{F}_{16}}\right\|_{L_1(P)},
\ese 
where the last inequality holds since $K_q=\mathcal{N}(\varepsilon,
  \mathcal{S}_{\Z q},\|\cdot\|_1)\leq K_5 \varepsilon^{-d_\z}$.  The
  total number of the brackets is  less or equal to $K_6^{K_q}
  \varepsilon^{-(d_\z+1)K_q}$   for some $K_6>0$. 
  
  Hence, $\mathcal{N}_{[]}((K_5+2) \varepsilon
  \left\|M_{\mathcal{F}_{16}}\right\|_{L_1(P)}, \mathcal{F}_{14},
  \|\cdot \|_{L_1(P)})\leq K_6^{K_q} \varepsilon^{-(d_\z+1)K_q}
  <\infty$ for every $\varepsilon$. Since
  $0<\|M_{\mathcal{F}_{16}}\|_{L_1(P)}< +\infty$, we can obtain
  $\mathcal{N}_{[]}( \varepsilon, \mathcal{F}_{14}, \|\cdot
  \|_{L_1(P)})<\infty$ for every $\varepsilon$.

Next we prove the second claim.
Recall $\mathcal{F}_{14}=\{t\mapsto \int q_{T\mid
  \Z}(t,\z;\bt_0)d \wt Q_{\Z}(\z)/q_T(t): \text{any empirical
  distribution }$ $\text{function} \wt Q_\Z \text{ on } \mathcal{S}_{\Z q}
\text{ s.t. } \sup_{\z \in \mathcal{S}_{\Z q}}|\wt
Q_\Z(\z)-Q_\Z(\z)|\leq 1/4\}$ and $\mathcal{F}_{16} = \{t \mapsto q_{T
  \mid \Z}(t, \z ; \bt_0)/q_T(t): \z \in \mathcal{S}_{\Z q}\}$. We
have $\mathcal{F}_{14} \subset
\overline{\operatorname{conv}}\mathcal{F}_{16}$ where
$\overline{\operatorname{conv}}\mathcal{F}_{16}$ denotes the closure
of the convex hull of $\mathcal{F}_{16}$. Note that for any
norm $\|\cdot\|$, $\mathcal{N}( \varepsilon \|c_5/q_T\|,
\mathcal{F}_{16}, \|\cdot \|) \leq \mathcal{N}_{[]}(2 \varepsilon
\|c_5/q_T\|, \mathcal{F}_{16}, \|\cdot \|)\leq K_5\varepsilon^{-d_\z}$
and  for any probability measure $\wt P$,
$0<\|M_{\mathcal{F}_{16}}\|_{L_2(\wt P)}<+\infty$ is equivalent to
$0<\|c_5/q_T\|_{L_2(\wt P)}<+\infty$. Then we have for any probability
measure $\wt P$ such that $0<\|M_{\mathcal{F}_{16}}\|_{L_2(\wt
  P)}<+\infty$, $$\mathcal{N}( \varepsilon
\|M_{\mathcal{F}_{16}}\|_{L_2(\wt P)}, \mathcal{F}_{16}, L_2(\wt P))
\leq \wt K_5\varepsilon^{-d_\z}$$ for some constant $\wt
  K_5$.
By \cite{van1996weak}[Theorem 2.6.9], it follows that for any probability measure $\wt P$ such that $0<\|M_{\mathcal{F}_{16}}\|_{L_2(\wt P)}<+\infty$,
$$\log \mathcal{N}( \varepsilon \|M_{\mathcal{F}_{16}}\|_{L_2(\wt P)}, \overline{\operatorname{conv}}\mathcal{F}_{16}, L_2(\wt P)) \leq K_7 \varepsilon^{-2d_\z/(d_\z+2)},$$
for some $K_7>0$.    Since $\mathcal{F}_{14} \subset \overline{\operatorname{conv}}\mathcal{F}_{16}$, we can obtain that $\mathcal{F}_{14}$ has a measurable envelope function $M_7(t)\equiv M_{\mathcal{F}_{16}}(t)$ such that $$
\int_0^{\infty} \sup _{\wt P} \sqrt{\log \mathcal{N}\left(\varepsilon\|M_7\|_{L_2(\wt P)}, \mathcal{F}_{14}, L_2(\wt P)\right)} d \varepsilon<\infty,
$$
where the $\sup$ is taken over all the probability measure $\wt P$ on
$\mathcal{S}_T$ such that $\|M_7\|_{L_2(\wt P)}<\infty$.

 Result (4): The proof is identical to that of Result (3), but with Condition B1 replaced by Condition B2, and functions $c_5(t)$ and $q_{T\mid \Z}(t,\z;\bt)$ replaced by functions $c_6(t)$ and $\frac{\partial}{\partial \bt}q_{T\mid \Z}(t,\z;\bt)$ respectively. Then we can obtain that $\mathcal{N}_{[]}\left(\varepsilon, \mathcal{F}_{17},
  L_1(P)\right)<\infty$ for every $\varepsilon>0$ and $\mathcal{F}_{17}$ has a measurable envelope function $M_8(t)\equiv \max\{1,\operatorname{diam} \mathcal{S}_{\Z
  q}\}c_6(t)/q_T(t)+\frac{\partial}{\partial \bt} q_{T \mid \Z}(t, \z_0 ;
\bt_0)/q_T(t)$
for some fixed $\z_0 \in \mathcal{S}_{\Z q}$ such that $$
\int_0^{\infty} \sup _{\wt P} \sqrt{\log \mathcal{N}\left(\varepsilon\|M_8\|_{L_2(\wt P)}, \mathcal{F}_{15}, L_2(\wt P)\right)} d \varepsilon<\infty,
$$
where the $\sup$ is taken over all the probability measure $\wt P$ on
$\mathcal{S}_T$ such that $\|M_8\|_{L_2(\wt P)}<\infty$.

\end{proof}

\begin{Lemma}
\label{lm:2.6}
Under Conditions \ref{A1}, \ref{B1}, \ref{B2} and \ref{B3}, we have
\begin{enumerate}[label=(\arabic*)]
\item $$
S_{0n}(x,\z,\wh q_T) - S_{0}(x,\z,q_T)= n^{-1} \sum_{i=1}^n \eta_0(X_i,\Z_i,\Delta_i,R_i;x,\z) + o_p(n_1^{-1/2}) + o_p(n_2^{-1/2}),$$
uniformly in $(x,\z)$ where 
\be
\eta_0(X,\Z,\Delta,R;x,\z)&\equiv& \frac{R}{\pi}
\eta_{0p}(X,\Delta;x,\z)+ \frac{1-R}{1-\pi} \eta_{0q}(\Z;x,\z) -
E_p\{\eta_{0p}(X,\Delta;x,\z)\}\label{eq:eta0}\\
\eta_{0p}(X,\Delta; x,\z) &\equiv& \varphi_0(X; x,\z) \gamma_0(X)\Delta + \gamma_1^{\varphi_0}(X)(1 - \Delta) - \gamma_2^{\varphi_0}(X) \label{eq:eta0p}\\
\eta_{0q}(\Z;x,\z)&\equiv& -\int \frac{I(t>x)}{q^2_T(t)} q_{T \mid \Z}(t, \z ; \bt_0)\{q_{T \mid \Z}(t,\Z;\bt_0)-q_T(t)\}d P_T(t),\label{eq:eta0q}
\ee
and $\left\| S_{0n}(x,\z,\wh q_T) - S_{0}(x,\z,q_T)  \right\|_\infty = o_p(n_1^{-1/4}) + o_p(n_2^{-1/4});$
\item 
\bse
S_{1n}(x,\z,\wh q_T) - S_{1}(x,\z,q_T)=n^{-1}\sum_{i=1}^n \eta_1(X_i,\Z_i,\Delta_i,R_i;x,\z)+ o_p(n_1^{-1/2}) + o_p(n_2^{-1/2}),
\ese
uniformly in $(x,\z)$ where 
\be
\eta_1(X,\Z,\Delta,R;x,\z)&\equiv& \frac{R}{\pi}
\eta_{1p}(X,\Delta;x,\z)+ \frac{1-R}{1-\pi} \eta_{1q}(\Z;x,\z) -
E_p\{\eta_{1p}(X,\Delta;x,\z)\}\label{eq:eta1}\\
\eta_{1p}(X,\Delta; x,\z) &\equiv& \varphi_1(X; x,\z) \gamma_0(X)\Delta + \gamma_1^{\varphi_1}(X)(1 - \Delta) - \gamma_2^{\varphi_1}(X) \label{eq:eta1p}\\
\eta_{1q}(\Z;x,\z)&\equiv& -\int  \frac{I(t>x)}{q^2_T(t)} \tfrac{\partial}{\partial \bt}q_{T \mid \Z}(t, \z ; \bt_0) \{q_{T \mid \Z}(t,\Z;\bt_0)-q_T(t)\}d P_T(t),\label{eq:eta1q}
\ee
and $\left\| S_{1n}(x,\z,\wh q_T) - S_{1}(x,\z,q_T)  \right\|_\infty = o_p(n_1^{-1/4}) + o_p(n_2^{-1/4});$

\item
$$S_{2n}(x,\z,\wh q_T, \wh q^*_T) - S_{2}(x,\z,q_T, q^*_T)= n^{-1} \sum_{i=1}^n \eta_2(X_i,\Z_i,\Delta_i,R_i;x,\z)+ o_p(n_1^{-1/2}) + o_p(n_2^{-1/2}),$$
uniformly in $(x,\z)$ where
\be
\eta_2(X,\Z,\Delta,R;x,\z)&\equiv& \frac{R}{\pi}
\eta_{2p}(X,\Delta;x,\z)+ \frac{1-R}{1-\pi} \eta_{2q}(\Z;x,\z) -
E_p\{\eta_{2p}(X,\Delta;x,\z)\}\label{eq:eta2}\\
\eta_{2p}(X,\Delta; x,\z) &\equiv& \varphi_2(X;x,\z) \gamma_0(X)\Delta_i+\gamma^{\varphi_2}_1(X)(1-\Delta_i)- \gamma^{\varphi_2}_2(X) \label{eq:eta2p}\\
\eta_{2q}(\Z;x,\z)&\equiv& \int  \frac{I(t>x)}{q^2_T(t)} q_{T \mid \Z}(t, \z ; \bt_0)\{\tfrac{\partial}{\partial \bt} q_{T \mid \Z}(t,\Z;\bt_0)-q^*_T(t)\}d P_T(t) \notag \\
&& - \int  \frac{2 I(t>x)q^*_T(t)}{q^3_T(t)} q_{T \mid \Z}(t, \z ; \bt_0) \{q_{T \mid \Z}(t,\Z;\bt_0)-q_T(t)\}d P_T(t),\label{eq:eta2q}
\ee
and $\left\| S_{2n}(x,\z,\wh q_T, \wh q^*_T) - S_{2}(x,\z,q_T, q^*_T)
\right\|_\infty = o_p(n_1^{-1/4}) + o_p(n_2^{-1/4}).$
    \end{enumerate}
  \end{Lemma}

\begin{proof}
  
\begin{enumerate}[label=(\arabic*)]

\item

Lemma \ref{lm:2.4} implies that $\|\wh q_T -
q_T\|_{\infty}=o_p(n_2^{-1/4})$
and by applying Taylor expansion, we can have 
\bse
&& S_{0n}(x,\z,\wh q_T) - S_{0}(x,\z,q_T)\\
&=& S_{0n}(x,\z,q_T) - S_{0}(x,\z,q_T) + \frac{\partial S_{0n}(x,\z,q_T)}{\partial q_T}[\wh q_T - q_T] + o_p(n_2^{-1/2})\\
&=& \int \frac{I(t>x) q_{T \mid \Z}(t, \z ; \bt_0)}{\int q_{T \mid \Z}(t, \z ; \bt_0) d Q_\Z(\z)} d \{\wh P_T(t) - P_T(t) \} +  \frac{\partial S_{0n}(x,\z,q_T)}{\partial q_T}[\wh q_T - q_T]  + o_p(n_2^{-1/2}) \\
 & \equiv & \int \varphi_0(t;x,\z) d \{\wh P_T(t) - P_T(t) \}  +  \frac{\partial S_{0n}(x,\z,q_T)}{\partial q_T}[\wh q_T - q_T]  + o_p(n_2^{-1/2}) ,
\ese
uniformly in $(x,\z)$.

First, we handle the first term.
We prove the class $\{t \mapsto \varphi_0(t;x,\z)\equiv\frac{I(t>x) q_{T \mid \Z}(t, \z ; \bt_0)}{\int q_{T \mid \Z}(t, \z ; \bt_0) d Q_\Z(\z)}: \z \in \mathcal{S}_{\Z p}, x \in \mathcal{S}_T\}$ satisfies the conditions of Lemma~C.1. From Result (1) in Lemma \ref{lm:2.5}, the class $\{t \mapsto q_{T \mid \Z}(t, \z ; \bt_0)/q_T(t): \z \in \mathcal{S}_{\Z p}\}$ satisfies the conditions of Lemma~C.1. After multiplying by $I(t>x)$, by Condition B1, conditions in Lemma~C.1 still hold for the class $\{\varphi_0\}$. Then, by applying Lemma~C.1, we obtain
\bse
 \int \varphi_0 d \wh P_T&=&n_1^{-1} \sum_{i=1}^{n_1} \{\varphi_0(X_i;x,\z) \gamma_0(X_i)\Delta_i+\gamma^{\varphi_0}_1(X_i)(1-\Delta_i)- \gamma^{\varphi_0}_2(X_i)\}+R_{n_1}(\varphi_0)\\
 &=&n^{-1} \sum_{i=1}^n \frac{R_i}{\pi}\{\varphi_0(X_i;x,\z)
 \gamma_0(X_i)\Delta_i+\gamma^{\varphi_0}_1(X_i)(1-\Delta_i)-
 \gamma^{\varphi_0}_2(X_i)\}+R_{n_1}(\varphi_0)\\
 &=&n^{-1} \sum_{i=1}^n \frac{R_i}{\pi}\eta_{0p}(\X_i,\Delta_i;x,\z) +R_{n_1}(\varphi_0),
\ese 
where $\sup _{\varphi_0}|R_{n_1}(\varphi_0)|=O({\ln ^3 n_1}/{n_1})$
and the class $\{\eta_{0p}\}$ is P-Donsker. Hence, we can write the first term as
\bse
\int \varphi_0(t;x,\z) d \{\wh P_T(t) - P_T(t) \} &=& n^{-1} \sum_{i=1}^n \frac{R_i}{\pi} \eta_{0p}(X_i,\Delta_i;x,\z) - E_p\{\eta_{0p}(X,\Delta;x,\z)\} +o_p(n_1^{-1/2}).
\ese
Next, we handle the second term. We have 
\bse
\frac{\partial S_{0n}(x,\z,q_T)}{\partial q_T}[\wh q_T - q_T]
&=& \int \frac{\partial \varphi_{0}(t,x,\z,q_T)}{\partial q_T}[\wh q_T - q_T]d \wh P_T(t)\\
&=&-n_2^{-1} \sum_{i=1}^{n_2} \int  \frac{I(t>x)}{q^2_T(t)} q_{T \mid \Z}(t, \z ; \bt_0)\{q_{T \mid \Z}(t,\Z_i;\bt_0)-q_T(t)\}d \wh P_T(t)\\
&=&  -n^{-1} \sum_{i=1}^{n} \frac{1-R_i}{1-\pi} \int  \frac{I(t>x)}{q^2_T(t)} q_{T \mid \Z}(t, \z ; \bt_0)\{q_{T \mid \Z}(t,\Z_i;\bt_0)-q_T(t)\}d P_T(t) \\
&& -\int  \frac{I(t>x)}{q_T^2(t)}q_{T \mid \Z}(t,
  \z ; \bt_0) \{\wh q_T(t)-q_T(t)\}d\{\wh P_T(t) - P_T(t)\}.
\ese

Let  $\mathcal{F}_{S_0}\equiv\{t \mapsto \frac{I(t>x)}{q_T^2(t)}q_{T \mid \Z}(t,
  \z ; \bt_0)\wt q_T(t): \z \in \mathcal{S}_{\Z p}, x \in \mathcal{S}_T, \wt q_T(t)/q_T(t) \in \mathcal{F}_{14}\}$, where $\mathcal{F}_{14}$ is defined in Lemma \ref{lm:2.5}.

From Result (1) and Result (3) in Lemma \ref{lm:2.5}, the
class $\mathcal{F}_{12}$ and $\mathcal{F}_{14}$ satisfies the Condition (b2) in Lemma
C.1. Then by Condition B4, Lemma \ref{lm:2.5} Result
  (1), Lemma \ref{lm:2.5} Result
  (3) and \cite{kosorok2007introduction}[Lemma 9.25], we
 obtain  
$$\mathcal{N}_{[]}(C_1\varepsilon,\mathcal{F}_{12}\cdot
\mathcal{F}_{14},L_1(P)) \leq \mathcal{N}_{[]}(\varepsilon
,\mathcal{F}_{12},L_1(P))\mathcal{N}_{[]}(\varepsilon
,\mathcal{F}_{14},L_1(P))<+\infty$$ 
for every $\varepsilon>0$, and $C_1$ denotes some finite constant.
Hence, the class ${\cal F}_{S_0}=I(t>x) \mathcal{F}_{12}\cdot
\mathcal{F}_{14}$ also satisfies 
$$\mathcal{N}_{[]}(\varepsilon, \mathcal{F}_{S_0},L_1(P))<+\infty.$$ 

 Next, noting that the classes $\mathcal{F}_{12}$ and
$\mathcal{F}_{14}$ are bounded from above,
for any $q^{\prime}_T(t)/q_T(t),
q^{\prime \prime}_T(t)/q_T(t) \in \mathcal{F}{14}$ and $q_{T \mid \Z}(t,
\z^{\prime};\bt_0)/q_T(t), q_{T \mid \Z}(t, \z^{\prime \prime};\bt_0)/q_T(t) \in
\mathcal{F}_{12}$, we have
\bse
&&\left|\frac{I(t>x)}{q_T^2(t)}q_{T \mid \Z}(t,
  \z^{\prime} ; \bt_0)
q^{\prime}_{T}(t) - \frac{I(t>x)}{q_T^2(t)}q_{T \mid \Z}(t,
  \z^{\prime \prime} ; \bt_0)
q^{\prime \prime}_{T}(t) \right|^2\\
&\leq& C \left\{\left|\frac{q^{\prime}_{T}(t) - q^{\prime \prime}_{T}(t)}{q_T(t)}\right|^2 + \left|\frac{q_{T \mid \Z}(t,
  \z^{\prime} ; \bt_0) - q_{T \mid \Z}(t,
  \z^{\prime \prime} ; \bt_0)}{q_T(t)}\right|^2\right\},
\ese
for some constant $C$. 
Note further that
${\cal F}_{12}$ and ${\cal F}_{14}$ have the second property of (1) and (3) of Lemma
\ref{lm:2.5}, hence by \cite{van1996weak}[Theorem 2.10.20],
the class
$\mathcal{F}_{S_0}$ also has this property, i.e.,
it satisfies the requirement on covering number in
Condition (b2) of Lemma~C.1. Combining the above results,
$\mathcal{F}_{S_0}$ satisfies Condition (b2) of Lemma~C.1.

Further taking into account of Condition \ref{B3}, the class
$\mathcal{F}_{S_0}$ satisfies all the requirements in Lemma~C.1.
In addition, 
since $\|\wh q_T-q_T\|_{\infty}=o_p(n_2^{-1/4})$,
$\frac{I(t>x)}{q_T^2(t)}q_{T \mid \Z}(t, \z ; \bt_0)\wh q_T(t)$
and  $\frac{I(t>x)}{q_T^2(t)}q_{T \mid \Z}(t, \z ; \bt_0) q_T(t)$
lies in the class $\mathcal{F}_{S_0}$ with
probability tending to 1. Let $\varphi_{S_0} \equiv  \frac{I(t>x)}{q_T^2(t)}q_{T \mid \Z}(t,
  \z ; \bt_0) \{\wh q_T(t)-q_T(t)\}$.  Then,
    $\operatorname{var}(\eta^{\varphi_{S_0}})$ goes to 0 as $n_2
    \rightarrow \infty$. Hence, by result (3) of Lemma~C.1, we have 
$$
\int \frac{I(t>x)}{q_T^2(t)} q_{T \mid \Z}(t, \z ; \bt_0)\{ \wh q_T(t)-q_T(t)\}d\{\wh P_T(t) - P_T(t)\} = o_p(n_1^{-1/2}),
$$
uniformly in $(x,\z)$. By the definition of~(\ref{eq:eta0q}), the
second term can be written as
\bse 
\frac{\partial S_{0n}(x,\z,q_T)}{\partial q_T}[\wh q_T - q_T]
&=& -n^{-1} \sum_{i=1}^{n} \frac{1-R_i}{1-\pi} \int  \frac{I(t>x)}{q^2_T(t)} q_{T \mid \Z}(t, \z ; \bt_0)\{q_{T \mid \Z}(t,\Z_i;\bt_0)-q_T(t)\}d P_T(t)\\&& +o_p(n_1^{-1/2})\\
&\equiv& n^{-1} \sum_{i=1}^{n} \frac{1-R_i}{1-\pi} \eta_{0q}(\Z_i;x,\z)+o_p(n_1^{-1/2}).
\ese

Combining the above results regarding the two terms, we get
\bse
&&S_{0n}(x,\z,\wh q_T) - S_{0}(x,\z,q_T)\\
  & = &  n^{-1} \sum_{i=1}^n \frac{R_i}{\pi} \eta_{0p}(X_i,\Delta_i;x,\z) - E_p\{\eta_{0p}(X,\Delta;x,\z)\} + n^{-1} \sum_{i=1}^{n} \frac{1-R_i}{1-\pi} \eta_{0q}(\Z_i;x,\z)+ o_p(n_1^{-1/2}) + o_p(n_2^{-1/2})\\
  &\equiv& n^{-1} \sum_{i=1}^n \left\{\frac{R_i}{\pi} \eta_{0p}(X_i,\Delta_i;x,\z)  + \frac{1-R_i}{1-\pi} \eta_{0q}(\Z_i;x,\z)\right\} -
E_p\{\eta_{0p}(X,\Delta;x,\z)\}  + o_p(n_1^{-1/2}) + o_p(n_2^{-1/2})\\
 & \equiv &  n^{-1} \sum_{i=1}^n \eta_0(X_i,\Z_i,\Delta_i,R_i;x,\z) + o_p(n_1^{-1/2}) + o_p(n_2^{-1/2})
 \ese
 uniformly in $(x,\z)$,
where $E\{\eta_0(X,\Z,\Delta,R;x,\z)\mid x,\z \}=0$,
and
\bse
 \|S_{0n}(x,\z,\wh q_T) - S_{0}(x,\z,q_T)\|_\infty= o_p(n_1^{-1/4}) +
o_p(n_2^{-1/4}),
\ese
uniformly in $(x,\z)$.

\item  Following the same path as in Result (1), applying Lemma \ref{lm:2.4} and Taylor expansion, it follows that
\bse
&& S_{1n}(x,\z,\wh q_T) - S_{1}(x,\z,q_T)\\
&=& S_{1n}(x,\z,q_T) - S_{1}(x,\z,q_T) + \frac{\partial S_{1n}(x,\z,q_T)}{\partial q_T}[\wh q_T - q_T] + o_p(n_2^{-1/2})\\
&=& \int \frac{I(t>x) \frac{\partial }{\partial \bt}q_{T \mid \Z}(t, \z ; \bt_0)}{\int q_{T \mid \Z}(t, \z ; \bt_0) d Q_\Z(\z)} d \{\wh P_T(t) - P_T(t) \} +  \frac{\partial S_{1n}(x,\z,q_T)}{\partial q_T}[\wh q_T - q_T]  + o_p(n_2^{-1/2}) \\
 & \equiv & \int \varphi_1(t;x,\z) d \{\wh P_T(t) - P_T(t) \}  +  \frac{\partial S_{1n}(x,\z,q_T)}{\partial q_T}[\wh q_T - q_T]  + o_p(n_2^{-1/2}),
\ese
uniformly in $(x,\z)$.
For the first term, we prove the class $\{t \mapsto \varphi_1(t;x,\z)\equiv \frac{I(t>x) \frac{\partial }{\partial \bt}q_{T \mid \Z}(t, \z ; \bt_0)}{\int q_{T \mid \Z}(t, \z ; \bt_0) d Q_\Z(\z)}: x \in \mathcal{S}_T, \z \in  \mathcal{S}_{\Z p}\}$ satisfies the conditions of Lemma~C.1. From Result (2) in Lemma \ref{lm:2.5}, the class $\{t \mapsto \frac{\partial}{\partial \bt}q_{T\mid \Z}(t,\z;\bt_0){ /q_T(t)}\}$ satisfies the Condition (b) of Lemma~C.1. After multiplying by $I(t>x)$, by Conditions \ref{B2} and \ref{B3}, the class $\{\varphi_1\}$ meets the conditions of Lemma~C.1. By applying Lemma~C.1, we have \bse
 \int \varphi_1 d \wh P_T &=&n_1^{-1} \sum_{i=1}^{n_1} \{\varphi_1(X_i;x,\z) \gamma_0(X_i)\Delta_i+\gamma^{\varphi_1}_1(X_i)(1-\Delta_i)- \gamma^{\varphi_1}_2(X_i)\}+R_{n_1}(\varphi_1)\\
 &=&n^{-1} \sum_{i=1}^n \frac{R_i}{\pi}\{\varphi_1(X_i;x,\z)
 \gamma_0(X_i)\Delta_i+\gamma^{\varphi_1}_1(X_i)(1-\Delta_i)-
 \gamma^{\varphi_1}_2(X_i)\}+R_{n_1}(\varphi_1)\\
 &=&n^{-1} \sum_{i=1}^n \frac{R_i}{\pi}\eta_{1p}(X_i,\Delta_i;x,\z) +R_{n_1}(\varphi_1),
\ese
where $\sup _{\varphi_1}|R_{n_1}(\varphi_1)|=O({\ln ^3 n_1}/{n_1})$
and the class $\{\eta_{1p}\}$ is Donsker. Thus, the first term can be expressed as
\bse
\int \varphi_1(t;x,\z) d \{\wh P_T(t) - P_T(t) \} &=& n^{-1} \sum_{i=1}^n \frac{R_i}{\pi} \eta_{1p}(X_i,\Delta_i;x,\z) - E_p\{\eta_{1p}(X,\Delta;x,\z)\} +o_p(n_1^{-1/2}).
\ese
Regarding the second term, it can be written as
\bse
\frac{\partial S_{1n}(x,\z,q_T)}{\partial q_T}[\wh q_T - q_T]
&=& \int \frac{\partial \varphi_{1}(t,x,\z,q_T)}{\partial q_T}[\wh q_T - q_T]d \wh P_T(t)\\
&=&-n_2^{-1} \sum_{i=1}^{n_2} \int  \frac{I(t>x)}{q^2_T(t)} \tfrac{\partial}{\partial \bt} q_{T \mid \Z}(t, \z ; \bt_0)\{q_{T \mid \Z}(t,\Z_i;\bt_0)-q_T(t)\}d \wh P_T(t)\\
&=&-n^{-1} \sum_{i=1}^n \frac{1-R_i}{1-\pi}\int  \frac{I(t>x)}{ q^2_T(t)} \tfrac{\partial}{\partial \bt} q_{T \mid \Z}(t, \z ; \bt_0)\{q_{T \mid \Z}(t,\Z_i;\bt_0)-q_T(t)\}d P_T(t) \\
&& -\int\frac{I(t>x)}{q_T^2(t)}\tfrac{\partial}{\partial \bt} q_{T \mid \Z}(t, \z ; \bt_0)\{ \wh q_T(t)-q_T(t)\}d\{\wh P_T(t) - P_T(t)\}.
\ese

Let $\mathcal{F}_{S_1}\equiv\{t \mapsto
\frac{I(t>x)}{q_T^2(t)}\frac{\partial}{\partial \bt} q_{T \mid \Z}(t,
\z ; \bt_0)\wt q_T(t): \z \in \mathcal{S}_{\Z p}, x \in \mathcal{S}_T,
\wt q_T(t)/q_T(t) \in \mathcal{F}_{14}\}$. Based on Result (2) and
Result (3) in Lemma \ref{lm:2.5}, the classes $\mathcal{F}_{13}$ and
$\mathcal{F}_{14}$ both satisfy the Condition (b2) in Lemma~C.1. Noting that $\mathcal{F}_{13}$ and
$\mathcal{F}_{14}$ are both bounded,
by \cite{kosorok2007introduction}[Lemma 9.25], it follows that 
$$\mathcal{N}_{[]}(C_3\varepsilon,\mathcal{F}_{13}\cdot \mathcal{F}_{14},L_1(P)) \leq \mathcal{N}_{[]}(\varepsilon ,\mathcal{F}_{13},L_1(P))\mathcal{N}_{[]}(\varepsilon ,\mathcal{F}_{14},L_1(P))<+\infty$$
for every $\varepsilon>0$ and $C_3$ denotes some finite constant. Hence, the class ${\cal F}_{S_1}=I(t>x)\mathcal{F}_{13}\cdot
\mathcal{F}_{14}$ also satisfies 
$$\mathcal{N}_{[]}(\varepsilon, \mathcal{F}_{S_1},L_1(P))<+\infty.$$

Since the classes $\mathcal{F}_{13}$ and  
$\mathcal{F}_{14}$ are
bounded above,
for any $q^{\prime}_T(t)/q_T(t),
q^{\prime \prime}_T(t)/q_T(t) \in \mathcal{F}_{14}$ and $\frac{\partial}{\partial \bt}q_{T \mid \Z}(t,
\z^{\prime};\bt_0)/q_T(t)$, $\frac{\partial}{\partial \bt}q_{T \mid \Z}(t, \z^{\prime \prime};\bt_0)/q_T(t) \in
\mathcal{F}_{13}$, we have
\bse
&&\left|\frac{I(t>x)}{q_T^2(t)} \frac{\partial}{\partial \bt}q_{T \mid \Z}(t,
  \z^{\prime} ; \bt_0)
q^{\prime}_{T}(t) - \frac{I(t>x)}{q_T^2(t)}\frac{\partial}{\partial \bt}q_{T \mid \Z}(t,
  \z^{\prime \prime} ; \bt_0)
q^{\prime \prime}_{T}(t) \right|^2\\
&\leq& C_5 \left\{\left|\frac{q^{\prime}_{T}(t) - q^{\prime \prime}_{T}(t)}{q_T(t)} \right|^2 + \left|\frac{\frac{\partial}{\partial \bt}q_{T \mid \Z}(t,
  \z^{\prime} ; \bt_0) - \frac{\partial}{\partial \bt}q_{T \mid \Z}(t,
  \z^{\prime \prime} ; \bt_0)}{q_T(t)}\right|^2\right\},
\ese
for some constant $C_5$. By Result (2) and Result (3) in Lemma \ref{lm:2.5} and \cite{van1996weak}[Theorem 2.10.20], the class $\mathcal{F}_{S_1}$ satisfies the requirement on covering number in Condition (b2) of Lemma~C.1. Hence, $\mathcal{F}_{S_1}$ satisfies Condition (b2) of Lemma~C.1. Combining Condition \ref{B3}, the class $\mathcal{F}_{S_1}$ meets all the requirements of Lemma~C.1. 

Since $\|\wh q_T-q_T\|_{\infty}=o_p(n_2^{-1/4})$,  $\frac{I(t>x)}{q_T^2(t)}\frac{\partial}{\partial \bt} q_{T \mid \Z}(t, \z ; \bt_0)\wh q_T(t)$ and $\frac{I(t>x)}{q_T^2(t)}\frac{\partial}{\partial \bt} q_{T \mid \Z}(t, \z ; \bt_0)q_T(t)$ lie in the class $\mathcal{F}_{S_1}$ with probability tending to 1. Let $\eta^{\varphi_{S_1}} \equiv \frac{I(t>x)}{q_T^2(t)}\frac{\partial}{\partial \bt} q_{T \mid \Z}(t, \z ; \bt_0)\{\wh q_T(t)-q_T(t)\}$. Then,
    $\operatorname{var}(\eta^{\varphi_{S_1}})$ goes to 0 as $n_2
    \rightarrow \infty$. Hence, by Result (3) in Lemma~C.1, we have 
$$\int \frac{I(t>x)}{q_T^2(t)} \tfrac{\partial}{\partial \bt} q_{T \mid \Z}(t, \z ; \bt_0)\{ \wh q_T(t)-q_T(t)\}d\{\wh P_T(t) - P_T(t)\}=o_p(n_1^{-1/2}),$$
uniformly in $(x,\z)$.  According to the definition
of~(\ref{eq:eta1q}), the second term becomes
\bse
\frac{\partial S_{1n}(x,\z,q_T)}{\partial q_T}[\wh q_T - q_T]
&=& -n^{-1} \sum_{i=1}^n \frac{1-R_i}{1-\pi}\int  \frac{I(t>x)}{ q^2_T(t)} \tfrac{\partial}{\partial \bt} q_{T \mid \Z}(t, \z ; \bt_0)\{q_{T \mid \Z}(t,\Z_i;\bt_0)-q_T(t)\}d P_T(t)\\&& +o_p(n_1^{-1/2})\\
&=& n^{-1} \sum_{i=1}^{n} \frac{1-R_i}{1-\pi} \eta_{1q}(\Z_i;x,\z)+o_p(n_1^{-1/2}).
\ese
Bringing the two terms together, it follows \bse
&& S_{1n}(x,\z,\wh q_T) - S_{1}(x,\z,q_T)\\
&=& n^{-1} \sum_{i=1}^n \frac{R_i}{\pi}\eta_{1p}(X_i,\Delta_i;x,\z) - E_p\{\eta_{1p}(X,\Delta;x,\z)\}+ n^{-1} \sum_{i=1}^{n} \frac{1-R_i}{1-\pi} \eta_{1q}(\Z_i;x,\z)+ o_p(n_1^{-1/2})+ o_p(n_2^{-1/2})\\
&\equiv& n^{-1} \sum_{i=1}^n \left\{\frac{R_i}{\pi} \eta_{1p}(X_i,\Delta_i;x,\z)  + \frac{1-R_i}{1-\pi} \eta_{1q}(\Z_i;x,\z)\right\} -
E_p\{\eta_{1p}(X,\Delta;x,\z)\}  + o_p(n_1^{-1/2}) + o_p(n_2^{-1/2})\\
&\equiv&  n^{-1} \sum_{i=1}^n \eta_1(X_i,\Z_i,\Delta_i,R_i;x,\z)+ o_p(n_1^{-1/2})+ o_p(n_2^{-1/2}),
\ese
uniformly in $(x,\z)$, where $E\{\eta_1(X_i,\Z_i,\Delta_i,R_i;x,\z)\mid x,\z\}=0$, and
\bse \|S_{1n}(x,\z,\wh q_T) - S_{1}(x,\z,q_T)\|_\infty = o_p(n_1^{-1/4})+ o_p(n_2^{-1/4}),\ese
uniformly in $(x,\z)$.

\item 
From Lemma \ref{lm:2.4}, we have $\|\wh q_T - q_T\|_{\infty}=o_p(n_2^{-1/4})$ and $\|\wh q^*_T - q^*_T\|_{\infty}=o_p(n_2^{-1/4})$. By applying Taylor expansion, it follows that
\bse
&&  S_{2n}(x,\z,\wh q_T, \wh q^*_T) - S_{2}(x,\z,q_T, q^*_T)\\
&=& S_{2n}(x,\z,q_T,q^*_T) - S_{2}(x,\z,q_T,q^*_T) + \frac{\partial S_{2n}(x,\z,q_T,q^*_T)}{\partial q_T}[\wh q_T - q_T]+\frac{\partial S_{2n}(x,\z,q_T,q^*_T)}{\partial q^*_T}[\wh q^*_T - q^*_T]\\
&& + o_p(n_2^{-1/2})\\
&=& \int \frac{I(t>x) q_{T \mid \Z}(t, \z ; \bt_0) \int \frac{\partial}{\partial \bt} q_{T\mid \Z}(t, \z ; \bt_0) d  Q_\Z(\z)}{\{\int q_{T \mid \Z}(t, \z ; \bt_0) d  Q_\Z(\z)\}^2} d \{\wh P_T(t)- P_T(t)\}+ \frac{\partial S_{2n}(x,\z,q_T,q^*_T)}{\partial q_T}[\wh q_T - q_T]\\
&& + \frac{\partial S_{2n}(x,\z,q_T,q^*_T)}{\partial q^*_T}[\wh q^*_T - q^*_T]+ o_p(n_2^{-1/2})\\
&\equiv&  \int \varphi_2(t;x,\z) d\{\wh P_T(t)-P_T(t)\}+ \frac{\partial S_{2n}(x,\z,q_T,q^*_T)}{\partial q_T}[\wh q_T - q_T]+\frac{\partial S_{2n}(x,\z,q_T,q^*_T)}{\partial q^*_T}[\wh q^*_T - q^*_T]\\
&& + o_p(n_2^{-1/2}).
\ese 
For the first term, we prove the class $\{t \mapsto \varphi_2(t;x,\z) \equiv \frac{I(t>x) q_{T \mid \Z}(t, \z ; \bt_0) \int \frac{\partial}{\partial \bt} q_{T\mid \Z}(t, \z ; \bt_0) d  Q_\Z(\z)}{\{\int q_{T \mid \Z}(t, \z ; \bt_0) d  Q_\Z(\z)\}^2}: \z \in \mathcal{S}_{\Z p}, x \in \mathcal{S}_T\}$ satisfies the conditions of Lemma~C.1. From Result (1) in Lemma 2.5 and Conditions \ref{B1} and \ref{B3}, the class $\{t \mapsto q_{T \mid \Z}(t, \z ; \bt_0)/q_T(t): \z \in \mathcal{S}_{\Z p}\}$ satisfies the conditions of Lemma~C.1. After multiplying by $ q_T^*(t)/q_T(t)$, which is bounded by Condition \ref{B4}, conditions in Lemma~C.1 still hold for the class $\{\varphi_2\}$. From Lemma~C.1, we have
\bse
\int \varphi_2(t;x,\z) d\wh P_T(t)&=&n_1^{-1} \sum_{i=1}^{n_1} \{\varphi_2(X_i;x,\z) \gamma_0(X_i)\Delta_i+\gamma^{\varphi_2}_1(X_i)(1-\Delta_i)- \gamma^{\varphi_2}_2(X_i)\}+R_{n_1}(\varphi_2)\\
 &=&n^{-1} \sum_{i=1}^n \frac{R_i}{\pi}\{\varphi_2(X_i;x,\z) \gamma_0(X_i)\Delta_i+\gamma^{\varphi_2}_1(X_i)(1-\Delta_i)- \gamma^{\varphi_2}_2(X_i)\}+R_{n_1}(\varphi_2)\\
 &\equiv& n^{-1} \sum_{i=1}^n  \frac{R_i}{\pi} \eta_{2p}(X_i,\Delta_i;x,\z)+R_{n_1}(\varphi_2),
\ese
where $\sup _{\varphi_2}|R_{n_1}(\varphi_2)|=O({\ln ^3 n_1}/{n_1})$
and the class $\{\eta_{2p}\}$ is Donsker. Hence, the first term is rewritten as
\bse
\int \varphi_2(t;x,\z) d \{\wh P_T(t) - P_T(t) \} &=& n^{-1} \sum_{i=1}^n \frac{R_i}{\pi} \eta_{2p}(X_i,\Delta_i;x,\z) - E_p\{\eta_{2p}(X,\Delta;x,\z)\} +o_p(n_1^{-1/2}).
\ese
Regarding the second term, we have
\bse
&&\frac{\partial S_{2n}(x,\z,q_T,q^*_T)}{\partial q_T}[\wh q_T - q_T]\\
&=&-n_2^{-1} \sum_{i=1}^{n_2} \int  \frac{2 I(t>x)q^*_T(t)}{q^3_T(t)} q_{T \mid \Z}(t, \z ; \bt_0) \{q_{T \mid \Z}(t,\Z_i;\bt_0)-q_T(t)\}d \wh P_T(t)\\
&=& -n^{-1} \sum_{i=1}^n \frac{1-R_i}{1-\pi} \int  \frac{2 I(t>x)q^*_T(t)}{q^3_T(t)} q_{T \mid \Z}(t, \z ; \bt_0) \{q_{T \mid \Z}(t,\Z_i;\bt_0)-q_T(t)\}d P_T(t)\\
&& -\int \frac{2I(t>x) q_T^*(t)}{q^3_T(t)}q_{T \mid \Z}(t, \z ; \bt_0)\{ \wh q_T(t)-q_T(t)\}d\{\wh P_T(t) - P_T(t)\}.
\ese

Let $\mathcal{F}_{S_{21}}\equiv \{t \mapsto \frac{I(t>x) q_{T \mid \Z}(t, \z ; \bt_0)q^*_T(t)}{\{\int q_{T \mid \Z}(t, \z ; \bt_0)d Q_\Z(\z)\}^3}\wt q_{T}(t): \z \in \mathcal{S}_{\Z p}, x \in \mathcal{S}_T, \wt q_{T}(t)/q_T(t) \in \mathcal{F}_{14}\}$. By the definitions in Result (1), the class $\mathcal{F}_{S_{21}}$ can be written as $\mathcal{F}_{S_{21}} = \frac{I(t>x)q_T^*(t)}{q_T(t)}\mathcal{F}_{12}\cdot\mathcal{F}_{14}$. We have proved in Result (1) that for every $\varepsilon>0$, $\mathcal{N}_{[]}(\varepsilon,\mathcal{F}_{12}\cdot
\mathcal{F}_{14},L_1(P)) < +\infty$. Since $q_T^*(t)/q_T(t)$ is bounded, the class $\mathcal{F}_{S_{21}}$ has $\mathcal{N}_{[]}(\varepsilon, \mathcal{F}_{S_{21}},L_1(P))<+\infty$ for every $\varepsilon>0$. 

Next, for any $q^{\prime}_T(t)/q_T(t),
q^{\prime \prime}_T(t)/q_T(t) \in \mathcal{F}_{14}$ and $q_{T \mid \Z}(t,
\z^{\prime};\bt_0)/q_T(t), q_{T \mid \Z}(t, \z^{\prime \prime};\bt_0)/q_T(t) \in
\mathcal{F}_{12}$, we have
\bse
&&\left|\frac{I(t>x) q_T^*(t)}{q_T^3(t)}q_{T \mid \Z}(t,
  \z^{\prime} ; \bt_0)
q^{\prime}_{T}(t) - \frac{I(t>x) q_T^*(t)}{q_T^3(t)}q_{T \mid \Z}(t,
  \z^{\prime \prime} ; \bt_0)
q^{\prime \prime}_{T}(t) \right|^2\\
&\leq& C_6 \left\{\left|\frac{q^{\prime}_{T}(t) - q^{\prime \prime}_{T}(t)}{q_T(t)} \right|^2 + \left|\frac{q_{T \mid \Z}(t,
  \z^{\prime} ; \bt_0) - q_{T \mid \Z}(t,
  \z^{\prime \prime} ; \bt_0)}{q_T(t)}\right|^2\right\},
\ese
for some constant $C_6$. By Result (1)  and Result (3) of Lemma
\ref{lm:2.5}, the classes $\mathcal{F}_{12}$ and 
  $\mathcal{F}_{14}$  satisfies the uniform entropy condition defined
  in \cite{van1996weak}[Section 2.5.1] 
Then by \cite{van1996weak}[Theorem 2.10.20], the class
$\mathcal{F}_{S_{21}}$ also satisfies the uniform entropy
condition. Combining the previous results, $\mathcal{F}_{S_{21}}$
satisfies Condition (b2) of Lemma~C.1. By Condition \ref{B3},
the class $\mathcal{F}_{S_{21}}$ satisfies all the requirements in
Lemma~C.1. Since $\|\wh q_T -
q_T\|_{\infty}=o_p(n_2^{-1/4})$, $\frac{I(t>x) q_T^*(t)}{q^3_T(t)}q_{T
  \mid \Z}(t, \z ; \bt_0)\wh q_T(t)$ and $\frac{I(t>x)
  q_T^*(t)}{q^3_T(t)}q_{T \mid \Z}(t, \z ; \bt_0)q_T(t)$ lie in the
class $\mathcal{F}_{S_{21}}$ with probability tending to 1. Let
$\varphi_{S_{21}} \equiv \frac{I(t>x) q_T^*(t)}{q^3_T(t)}q_{T \mid
  \Z}(t, \z ; \bt_0)\{ \wh q_T(t)-q_T(t)\}$. Then,
$\operatorname{var}(\eta^{\varphi_{S_{21}}})$ goes to 0 as $n_2 
    \rightarrow \infty$. Hence, by Result (3) of Lemma~C.1, we have 
$$
\int \frac{I(t>x)q_T^*(t)}{q^3_T(t)}q_{T \mid \Z}(t, \z ; \bt_0)\{ \wh q_T(t)-q_T(t)\}d\{\wh P_T(t) - P_T(t)\} = o_p(n_1^{-1/2}),
$$
uniformly in $(x,\z)$. Thus,
\bse
&&\frac{\partial S_{2n}(x,\z,q_T,q^*_T)}{\partial q_T}[\wh q_T - q_T]\\
&=& -n^{-1} \sum_{i=1}^n \frac{1-R_i}{1-\pi} \int  \frac{2 I(t>x)q^*_T(t)}{q^3_T(t)} q_{T \mid \Z}(t, \z ; \bt_0) \{q_{T \mid \Z}(t,\Z_i;\bt_0)-q_T(t)\}d P_T(t)
+o_p(n_1^{-1/2}).
\ese

Next, we deal with the third term. The third term can be written as 
\bse
\frac{\partial S_{2n}(x,\z,q_T,q^*_T)}{\partial q^*_T}[\wh q^*_T - q^*_T] &=& n_2^{-1} \sum_{i=1}^{n_2} \int  \frac{I(t>x) q_{T \mid \Z}(t, \z ; \bt_0)\{(1-R_i) \frac{\partial}{\partial \bt}q_{T \mid \Z}(t,\Z_i;\bt_0)-q_T^*(t)\}}{\{\int q_{T \mid \Z}(t, \z ; \bt_0) d Q_\Z(\z)\}^2}d \wh P_T(t)\\
&=& n^{-1} \sum_{i=1}^n  \frac{1-R_i}{1-\pi} \int  \frac{I(t>x)}{q^2_T(t)} q_{T \mid \Z}(t, \z ; \bt_0)\{\tfrac{\partial}{\partial \bt} q_{T \mid \Z}(t,\Z;\bt_0)-q^*_T(t)\}d \wh P_T(t) \\
&=& n^{-1} \sum_{i=1}^n \frac{1-R_i}{1-\pi} \int  \frac{I(t>x)}{q^2_T(t)} q_{T \mid \Z}(t, \z ; \bt_0)\{\tfrac{\partial}{\partial \bt} q_{T \mid \Z}(t,\Z;\bt_0)-q^*_T(t)\}d P_T(t) \\
&& +\int \frac{I(t>x)}{q_T^2(t)}q_{T \mid \Z}(t, \z ; \bt_0)\{ \wh q_T^*(t)-q_T^*(t)\}d\{\wh P_T(t) - P_T(t)\}.
\ese 
Let $\mathcal{F}_{S_{22}}\equiv\{t \mapsto \frac{I(t>x) q_{T \mid \Z}(t, \z ; \bt_0)}{\{\int q_{T \mid \Z}(t, \z ; \bt_0)d Q_\Z(\z)\}^2} \wt q_{T}^*(t): \z \in \mathcal{S}_{\Z p}, x \in \mathcal{S}_T, \wt q_{T}^*(t)/q_T(t) \in \mathcal{F}_{15}\}$. From Result (1) and Result (4) in Lemma \ref{lm:2.5}, the classes $\mathcal{F}_{12}$ and $\mathcal{F}_{15}$ satisfy Condition (b2) of Lemma~C.1. Noting that the classes $\mathcal{F}_{12}$ and $\mathcal{F}_{15}$ are both bounded, by \cite{kosorok2007introduction}[Lemma 9.25], we have $$\mathcal{N}_{[]}(C_7\varepsilon,\mathcal{F}_{12}\cdot \mathcal{F}_{15},L_1(P)) \leq \mathcal{N}_{[]}(\varepsilon ,\mathcal{F}_{12},L_1(P))\mathcal{N}_{[]}(\varepsilon ,\mathcal{F}_{15},L_1(P))<+\infty$$
for every $\varepsilon>0$ and $C_7$ denotes some finite constant. Hence, the class ${\cal F}_{S_{22}}=I(t>x) \mathcal{F}_{12}\cdot
\mathcal{F}_{15}$ has $\mathcal{N}_{[]}(\varepsilon, \mathcal{F}_{S_{22}},L_1(P))<+\infty$ for
every $\varepsilon > 0$. 

Since the classes $\mathcal{F}_{12}$ and
$\mathcal{F}_{15}$ are bounded from above, for any $q^{*\prime}_T(t)/q_T(t),
q^{*\prime \prime}_T(t)/q_T(t) \in \mathcal{F}_{15}$ and $q_{T \mid \Z}(t,
\z^{\prime};\bt_0)/q_T(t), q_{T \mid \Z}(t, \z^{\prime \prime};\bt_0)/q_T(t) \in
\mathcal{F}_{12}$, we have
\bse
&&\left|\frac{I(t>x)}{q_T^2(t)}q_{T \mid \Z}(t,
  \z^{\prime} ; \bt_0)
q^{*\prime}_{T}(t) - \frac{I(t>x)}{q_T^2(t)}q_{T \mid \Z}(t,
  \z^{\prime \prime} ; \bt_0)
q^{*\prime \prime}_{T}(t) \right|^2\\
&\leq& C_9 \left\{\left|\frac{q^{*\prime}_{T}(t) - q^{*\prime \prime}_{T}(t)}{q_T(t)} \right|^2 + \left|\frac{q_{T \mid \Z}(t,
  \z^{\prime} ; \bt_0) - q_{T \mid \Z}(t,
  \z^{\prime \prime} ; \bt_0)}{q_T(t)}\right|^2\right\},
\ese
for some constant $C_9$. Note further both ${\cal F}_{12}$ and $\mathcal{F}_{15}$ satisfy the second property of (1) of Lemma
\ref{lm:2.5}, hence by \cite{van1996weak}[Theorem 2.10.20],
the class
$\mathcal{F}_{S_{22}}$ also has this property, i.e.,
it satisfies the requirement on covering number in
Condition (b2) of Lemma~C.1. Combining the above results,
$\mathcal{F}_{S_{22}}$ satisfies Condition (b2) of Lemma~C.1. By Condition \ref{B4}, the class $\mathcal{F}_{S_{22}}$ meets the conditions of Lemma~C.1. Since $\|\wh q^*_T(t) - q^*_T(t) \|_\infty = o_p(n_2^{-1/4})$, $\frac{I(t>x)}{q_T^2(t)}q_{T \mid \Z}(t, \z ; \bt_0)\wh q_T^*(t)$ and $\frac{I(t>x)}{q_T^2(t)}q_{T \mid \Z}(t, \z ; \bt_0)q_T^*(t)$ lie in the class $\mathcal{F}_{S_{22}}$ with probability tending to 1. Let $\varphi_{S_{22}} \equiv \frac{I(t>x)}{q_T^2(t)}q_{T \mid \Z}(t, \z ; \bt_0)\{\wh q_T^*(t)-q_T^*(t)\}$. Then, $\operatorname{var}(\eta^{\varphi_{S_{22}}})$ goes to 0 as $n_2
    \rightarrow \infty$. Then we have$$
\int \frac{I(t>x)}{q_T^2(t)}q_{T \mid \Z}(t, \z ; \bt_0)\{ \wh q_T^*(t)-q_T^*(t)\}d\{\wh P_T(t) - P_T(t)\}= o_p(n_1^{-1/2}),
$$
uniformly in $(x,\z)$.
 Thus,
\bse
\frac{\partial S_{2n}(x,\z,q_T,q^*_T)}{\partial q^*_T}[\wh q^*_T - q^*_T] &=& n_2^{-1} \sum_{i=1}^{n_2} \int  \frac{I(t>x) q_{T \mid \Z}(t, \z ; \bt_0)\{(1-R_i) \frac{\partial}{\partial \bt}q_{T \mid \Z}(t,\Z_i;\bt_0)-q_T^*(t)\}}{\{\int q_{T \mid \Z}(t, \z ; \bt_0) d Q_\Z(\z)\}^2}d \wh P_T(t)\\
&=& n^{-1} \sum_{i=1}^n \frac{1-R_i}{1-\pi} \int  \frac{I(t>x)}{q^2_T(t)} q_{T \mid \Z}(t, \z ; \bt_0)\{\tfrac{\partial}{\partial \bt} q_{T \mid \Z}(t,\Z;\bt_0)-q^*_T(t)\}d P_T(t) \\
&& +o_p(n_1^{-1/2}).
\ese
Using the definition of~(\ref{eq:eta2q}), the sum of the second and third terms is given by
\bse
&&\frac{\partial S_{2n}(x,\z,q_T,q^*_T)}{\partial q_T}[\wh q_T -
q_T]+\frac{\partial S_{2n}(x,\z,q_T,q^*_T)}{\partial q^*_T}[\wh q^*_T
- q^*_T]\\
&=&
 -n^{-1} \sum_{i=1}^n \frac{1-R_i}{1-\pi} \int  \frac{2 I(t>x)q^*_T(t)}{q^3_T(t)} q_{T \mid \Z}(t, \z ; \bt_0) \{q_{T \mid \Z}(t,\Z_i;\bt_0)-q_T(t)\}d P_T(t)\\
&& + n^{-1} \sum_{i=1}^n \frac{1-R_i}{1-\pi} \int  \frac{I(t>x)}{q^2_T(t)} q_{T \mid \Z}(t, \z ; \bt_0)\{\tfrac{\partial}{\partial \bt} q_{T \mid \Z}(t,\Z;\bt_0)-q^*_T(t)\}d P_T(t) +o_p(n_1^{-1/2}) \\
&=& n^{-1} \sum_{i=1}^n\frac{1-R_i}{1-\pi} \eta_{2q}(\Z_i;x,\z) +o_p(n_1^{-1/2}).
\ese
From the above, it follows that
\bse
&&   S_{2n}(x,\z,\wh q_T, \wh q^*_T) - S_{2}(x,\z,q_T, q^*_T)\\
&=&n^{-1} \sum_{i=1}^n  \frac{R_i}{\pi} \eta_{2p}(X_i,\Delta_i;x,\z) - E_p\{\eta_{2p}(X,\Delta;x,\z)\} + n^{-1} \sum_{i=1}^n\frac{1-R_i}{1-\pi} \eta_{2q}(\Z_i;x,\z) + o_p(n_1^{-1/2})+ o_p(n_2^{-1/2})\\
&\equiv&  n^{-1} \sum_{i=1}^n \left\{\frac{R_i}{\pi} \eta_{2p}(X_i,\Delta_i;x,\z)  + \frac{1-R_i}{1-\pi} \eta_{2q}(\Z_i;x,\z)\right\} -
E_p\{\eta_{2p}(X,\Delta;x,\z)\}  + o_p(n_1^{-1/2}) + o_p(n_2^{-1/2})\\
&\equiv&n^{-1} \sum_{i=1}^n \eta_2(X_i,\Z_i,\Delta_i,R_i;x,\z) + o_p(n_1^{-1/2})+ o_p(n_2^{-1/2}),
\ese
uniformly in $(x,\z)$ and 
\bse
\left\| S_{2n}(x,\z,\wh q_T, \wh q^*_T) - S_{2}(x,\z,q_T, q^*_T)  \right\|_\infty 
=  o_p(n_1^{-1/4})+ o_p(n_2^{-1/4}),
\ese
uniformly in $(x,\z)$.

 \end{enumerate}

\end{proof}

\subsubsection{Proof of Theorem 3.1}

\begin{proof}
    The identifiability by Proposition 2.1 and Lemma 5.35 in
    \cite{van2000asymptotic} implies that $E\{\ell(X, \Z, \Delta, R ;
    \bt)\}$ has a unique maximum over $\Theta$ at $\bt_0$. Next, we prove the
    function $\bt \mapsto E\{\ell(X, \Z, \Delta, R ; \bt)\}$ is
    continuous on $\Theta$. 
 Under Condition A1 and A6, for every $(x,\z) \in \mathcal{S}_T \times (\mathcal{S}_{\Z p} \cup \mathcal{S}_{\Z q})$,  $q_{T\mid \Z}(x,\z;\bt)$, $\int q_\Z(\z) q_{T\mid \Z}(x,\z;\bt)d\z$ and $\int_x^{\infty}\frac{ p_T(t)q_{T\mid \Z}(t,\z;\bt)}{
\int q_\Z(\z)q_{T\mid \Z}(t,\z;\bt)d\z}dt$ are bounded away from 0 and
are continuous with respect to $\bt$. Then for every $(\bt, \wt \bt)
\in \Theta^2$, $|\ell(x,\z,\delta,r;\bt) - \ell(x,\z,\delta,r;\wt
\bt)| \xrightarrow{ \bt \rightarrow \wt \bt} 0$.
By Conditions A2-A4, \bse|\ell(x,\z,\delta,r;\bt) -
\ell(x,\z,\delta,r;\wt \bt)| &\leq& 
  C \left[r\delta|\log\{m_1(x,\z)\}|+r\delta|\log\{m_2(x)\}|+r(1-\delta)|\log\{m_3(x,\z)\}|\right.\\
&& +
\left.r\delta|\log\{M_1(x,\z)\}|+r\delta|\log\{M_2(x)\}|+r(1-\delta)|\log\{M_3(x,\z)\}|\right],\ese for some constant $C$,
where the right side is integrable.

By applying the Lebesgue dominated convergence theorem,
$|E\{\ell(X,\Z,\Delta,R;\bt)\} - E\{\ell(X,\Z,\Delta,R;\wt \bt)\}|
\xrightarrow{ \bt \rightarrow \wt \bt} 0$, that is, $\bt \mapsto
E\{\ell(X, \Z, \Delta,R ; \bt)\}$ is continuous on $\Theta$.  
In Lemma \ref{lem2.2}, we proved that  
    $\sup_{\bt \in \Theta}|\ell_n (\bt)-E\{\ell(X, \Z, \Delta,R ; \bt)\}| \povr 0.$
    Hence, according to \cite{van2000asymptotic}[Theorem 5.7], under \ref{A1}-\ref{A6}, $\wh \bt \povr \bt_0$
  \end{proof}
  
  \subsubsection{Proof of Lemma 3.2}
  
  The second term $\psi_{2n}(\cdot)$ in the approximated score function can be written as 
\bse
&& n^{-1}\sum_{i=1}^n \psi_{2n}(X_i,\Delta_i,R_i;\bt_0) \\
&=& n^{-1}\sum_{i=1}^n \{ \psi_{2n}(X_i,\Delta_i,R_i;\bt_0) - \psi_{2}(X_i,\Delta_i,R_i;\bt_0)\}+ n^{-1}\sum_{i=1}^n \psi_{2}(X_i,\Delta_i,R_i;\bt_0).
\ese 
From Lemma \ref{lm:2.4}, we have $\|\wh q_T(t)-q_T(t) \|_\infty=o_p(n_2^{-1/4})$ and $\|\wh q_T^*(t)-q_T^*(t) \|_\infty=o_p(n_2^{-1/4})$. Hence, by applying Taylor's expansion, we have the following uniformly in $(x,\delta,r)$
\bse
&& \psi_{2n}(x,\delta,r;\bt_0) - \psi_{2}(x,\delta,r;\bt_0)\\
&=& \frac{r\delta}{\pi}\frac{\{ \wh q_T^*(x)-q_T^*(x) \}}{q_T(x) } - \frac{r\delta}{\pi}\frac{q_T^*(x)  \{\wh q_T(x) - q_T(x) \} }{\{ q_T(x)\}^2}  +o_p(n_2^{-1/2})\\
&=&n^{-1} \sum_{i=1}^n \frac{r\delta}{\pi} \frac{1-R_i}{1-\pi}\frac{\frac{\partial}{\partial \bt}  q_{T\mid\Z}(x,\Z_i;\bt_0)q_T(x)- q_T^*(x) q_{T\mid\Z}(x,\Z_i;\bt_0) }{\{q_T(x)\}^2} +o_p(n_2^{-1/2}).
\ese 
Using the Hajek projection of U-statistics, it follows that 
\be
&& n^{-1}\sum_{i=1}^n \{\psi_{2n}(X_i,\Delta_i,R_i;\bt_0) - \psi_{2}(X_i,\Delta_i,R_i;\bt_0)\} \notag \\
 &=& n^{-2} \sum_{i=1}^n \sum_{j=1}^n \frac{R_j\Delta_j}{\pi} \frac{1-R_i}{1-\pi}\frac{\frac{\partial}{\partial \bt}  q_{T\mid\Z}(X_j,\Z_i;\bt_0)q_T(X_j)- q_T^*(X_j) q_{T\mid\Z}(X_j,\Z_i;\bt_0) }{\{q_T(X_j)\}^2} +o_p(n_2^{-1/2})\notag\\
 &=& n^{-1} \sum_{i=1}^n E \left[ \left. \frac{R\Delta}{\pi} \frac{1-R_i}{1-\pi}\frac{\frac{\partial}{\partial \bt}  q_{T\mid\Z}(X,\Z_i;\bt_0)q_T(X)- q_T^*(X) q_{T\mid\Z}(X,\Z_i;\bt_0) }{\{q_T(X)\}^2}\right| \Z_i, R_i \right]+o_p(n^{-1/2})+o_p(n_2^{-1/2})\notag\\
 &=& n^{-1} \sum_{i=1}^n \frac{1-R_i}{1-\pi} E \left[ \left. \frac{R\Delta}{\pi} \frac{\frac{\partial}{\partial \bt}  q_{T\mid\Z}(X,\Z_i;\bt_0)q_T(X)- q_T^*(X) q_{T\mid\Z}(X,\Z_i;\bt_0) }{\{q_T(X)\}^2}\right| \Z_i \right]\notag\\
 && +o_p(n_1^{-1/2})+o_p(n_2^{-1/2}). \label{eq:psi2}
\ee
Hence, we obtain
\bse
&& n^{-1}\sum_{i=1}^n \psi_{2n}(X_i,\Delta_i,R_i;\bt_0) \\
&=& n^{-1}\sum_{i=1}^n \left( \psi_{2}(X_i,\Delta_i,R_i;\bt_0)+\frac{1-R_i}{1-\pi} E \left[ \left. \frac{R\Delta}{\pi} \frac{\frac{\partial}{\partial \bt}  q_{T\mid\Z}(X,\Z_i;\bt_0)q_T(X)- q_T^*(X) q_{T\mid\Z}(X,\Z_i;\bt_0) }{\{q_T(X)\}^2}\right| \Z_i \right]\right)\\
&& +o_p(n_1^{-1/2})+o_p(n_2^{-1/2}).
\ese

  \subsubsection{Proof of Lemma 3.3}
  
  Following the same path as in the proof of Lemma 3.2, for $\psi_{3n}(x,\z,\delta,r;\bt_0)$, we have
\bse
&& n^{-1}\sum_{i=1}^n \psi_{3n}(X_i,\Z_i,\Delta_i,R_i;\bt_0)\\
&=& n^{-1}\sum_{i=1}^n \{\psi_{3n}(X_i,\Z_i,\Delta_i,R_i;\bt_0) - \psi_{3}(X_i,\Z_i,\Delta_i,R_i;\bt_0)\} + n^{-1}\sum_{i=1}^n \psi_{3}(X_i,\Z_i,\Delta_i,R_i;\bt_0).
\ese
By the definition, \bse
&& \psi_{3n}(x,\z,\delta,r;\bt_0) - \psi_{3}(x,\z,\delta,r;\bt_0)\\
&=& \frac{r(1-\delta)}{\pi} \left[ \left\{ \frac{S_{1 n}(x,\z,\wh q_T)}{S_{0 n}(x,\z,\wh q_T)} - \frac{S_{1}(x,\z, q_T)}{S_{0}(x,\z, q_T)} \right\} - \left\{\frac{S_{2n}(x,\z,\wh q_T,\wh q^{*}_T)}{S_{0 n}(x,\z,\wh q_T)} - \frac{ S_2(x,\z,q_T,q^{*}_T)}{S_{0}(x,\z, q_T)} \right\}\right].
\ese
According to Lemma \ref{lm:2.6}, $\| S_{0n}(x,\z,\wh q_T) - S_{0}(x,\z,q_T)\|_\infty = o_p(n_1^{-1/4}) + o_p(n_2^{-1/4})$, $\| S_{1n}(x,\z,\wh q_T) - S_{1}(x,\z,q_T)\|_\infty = o_p(n_1^{-1/4}) + o_p(n_2^{-1/4})$ and $\| S_{2n}(x,\z,\wh q_T, \wh q^*_T) - S_{2}(x,\z,q_T, q^*_T)\|_\infty = o_p(n_1^{-1/4}) + o_p(n_2^{-1/4})$. By Taylor's expansion, we have 
\bse
&& \frac{S_{1 n}(x,\z,\wh q_T)}{S_{0 n}(x,\z,\wh q_T)} - \frac{S_{1}(x,\z, q_T)}{S_{0}(x,\z, q_T)} \\
&=& \frac{1}{S_0(x,\z,q_T)} n^{-1} \sum_{i=1}^n \eta_1(X_i,\Z_i,\Delta_i,R_i;x,\z)- \frac{S_{1}(x,\z, q_T)}{\{S_0(x,\z,q_T)\}^2} n^{-1} \sum_{i=1}^n \eta_0(X_i,\Z_i,\Delta_i,R_i;x,\z) + o_p(n_1^{-1/2})+o_p(n_2^{-1/2}) \\
&=& n^{-1} \sum_{i=1}^n \left[ \frac{\eta_1(X_i,\Z_i,\Delta_i,R_i;x,\z)}{S_0(x,\z,q_T)} - \frac{S_{1}(x,\z, q_T) \eta_0(X_i,\Z_i,\Delta_i,R_i;x,\z) }{\{S_0(x,\z,q_T)\}^2} \right] + o_p(n_1^{-1/2})+o_p(n_2^{-1/2}),
\ese
uniformly in $(x,\z)$
and
\bse
&& \frac{S_{2 n}(x,\z,\wh q_T, \wh q_T^*)}{S_{0 n}(x,\z,\wh q_T)} - \frac{S_{2}(x,\z, q_T,q_T^*)}{S_{0}(x,\z, q_T)} \\
&=& \frac{1}{S_0(x,\z,q_T)} n^{-1} \sum_{i=1}^n \eta_2(X_i,\Z_i,\Delta_i,R_i;x,\z)- \frac{S_{2}(x,\z, q_T)}{\{S_0(x,\z,q_T)\}^2} n^{-1} \sum_{i=1}^n \eta_0(X_i,\Z_i,\Delta_i,R_i;x,\z) + o_p(n_1^{-1/2})+o_p(n_2^{-1/2}) \\
&=& n^{-1} \sum_{i=1}^n \left[ \frac{\eta_2(X_i,\Z_i,\Delta_i,R_i;x,\z)}{S_0(x,\z,q_T)} - \frac{S_{2}(x,\z, q_T,q_T^*) \eta_0(X_i,\Z_i,\Delta_i,R_i;x,\z) }{\{S_0(x,\z,q_T)\}^2} \right] + o_p(n_1^{-1/2})+o_p(n_2^{-1/2}),
\ese
uniformly in $(x,\z)$.
Hence, the following holds
\bse
&& \psi_{3n}(x,\z,\delta,r;\bt_0) - \psi_{3}(x,\z,\delta,r;\bt_0)\\
&=& \frac{r(1-\delta)}{\pi} \left[ \left\{ \frac{S_{1 n}(x,\z,\wh q_T)}{S_{0 n}(x,\z,\wh q_T)} - \frac{S_{1}(x,\z, q_T)}{S_{0}(x,\z, q_T)} \right\} - \left\{\frac{S_{2n}(x,\z,\wh q_T,\wh q^{*}_T)}{S_{0 n}(x,\z,\wh q_T)} - \frac{ S_2(x,\z,q_T,q^{*}_T)}{S_{0}(x,\z, q_T)} \right\}\right]\\
&=&n^{-1} \sum_{i=1}^n \frac{r(1-\delta)}{\pi} \left[ \frac{\eta_1(X_i,\Z_i,\Delta_i,R_i;x,\z)-\eta_2(X_i,\Z_i,\Delta_i,R_i;x,\z)}{S_0(x,\z,q_T)} \right.\\
&&\left. - \frac{\{S_{1}(x,\z, q_T)-S_{2}(x,\z, q_T,q_T^*)\} \eta_0(X_i,\Z_i,\Delta_i,R_i;x,\z) }{\{S_0(x,\z,q_T)\}^2} \right]+ o_p(n_1^{-1/2})+o_p(n_2^{-1/2}),
\ese
uniformly in $(x,\z,\delta,r)$.
Note that from formulas~(\ref{eq:eta0}), (\ref{eq:eta1}) and
  (\ref{eq:eta2}), each of $\eta_0$, $\eta_1$, and $\eta_2$ contain
  two components: one associated with the estimation of $p_T(\cdot)$,
  and the other with the estimation of $q_\Z(\cdot)$, that is, for
  $j=0,1,2$,  
\bse
\eta_j(X,\Z,\Delta,R;x,\z) &=& \frac{R}{\pi} \eta_{jp}(X,\Delta;x,\z)+ \frac{1-R}{1-\pi} \eta_{jq}(\Z;x,\z) - E_p\{\eta_{jp}(X,\Delta;x,\z)\}.
\ese Combining the decomposition of $\eta_0$, $\eta_1$, and $\eta_2$
and using the Hajek projection of U-statistics, we have 
\be
&& n^{-1}\sum_{i=1}^n \{\psi_{3n}(X_i,\Z_i,\Delta_i,R_i;\bt_0) - \psi_{3}(X_i,\Z_i,\Delta_i,R_i;\bt_0)\} \notag\\
&=& n^{-2} \sum_{i=1}^n \sum_{j=1}^n \frac{R_j(1-\Delta_j)}{\pi}
\left[
  \frac{\eta_1(X_i,\Z_i,\Delta_i,R_i;X_j,\Z_j)-\eta_2(X_i,\Z_i,\Delta_i,R_i;X_j,\Z_j)}{S_0(
      X_j,\Z_j,q_T)} \right.\notag\\
&&\left. - \frac{\{S_{1}(X_j,\Z_j, q_T)-S_{2}(X_j,\Z_j, q_T,q_T^*)\} \eta_0(X_i,\Z_i,\Delta_i,R_i;X_j,\Z_j) }{\{S_0(X_j,\Z_j,q_T)\}^2} \right] + o_p(n_1^{-1/2})+o_p(n_2^{-1/2})\notag\\
&=& n^{-1} \sum_{i=1}^n E\left( \frac{R(1-\Delta)}{\pi} \left[ \frac{\eta_1(X_i,\Z_i,\Delta_i,R_i;X,\Z)-\eta_2(X_i,\Z_i,\Delta_i,R_i;X,\Z)}{S_0(X,\Z,q_T)} \right.\right.\notag\\
&&\left. \left. \left.- \frac{\{S_{1}(X,\Z, q_T)-S_{2}(X,\Z,
        q_T,q_T^*)\} \eta_0(X_i,\Z_i,\Delta_i,R_i;X,\Z)
      }{\{S_0(X,\Z,q_T)\}^2} \right]\right|X_i,\Z_i,\Delta_i, R_i
\right)\notag \\&&+ o_p(n_1^{-1/2})+o_p(n_2^{-1/2})\notag\\
&=&n^{-1} \sum_{i=1}^n E_p\left( (1-\Delta) \left[ \frac{\eta_1(X_i,\Z_i,\Delta_i,R_i;X,\Z)-\eta_2(X_i,\Z_i,\Delta_i,R_i;X,\Z)}{S_0(X,\Z,q_T)} \right.\right.\notag\\
&&\left. \left. \left.- \frac{\{S_{1}(X,\Z, q_T)-S_{2}(X,\Z,
        q_T,q_T^*)\} \eta_0(X_i,\Z_i,\Delta_i,R_i;X,\Z)
      }{\{S_0(X,\Z,q_T)\}^2} \right]\right|X_i,\Z_i,\Delta_i, R_i
\right)\notag 
+ o_p(n_1^{-1/2})+o_p(n_2^{-1/2})\notag\\
&=&  n^{-1} \sum_{i=1}^n\left( \frac{R_i}{\pi} E_p \left\{ (1-\Delta)\frac{\eta_{1p}(X_i,\Delta_i;X,\Z) - \eta_{2p}(X_i,\Delta_i;X,\Z)}{S_0(X,\Z,q_T)}\right| X_i,\Delta_i\right\}\notag \\
 && \left. -\frac{R_i}{\pi}  E_p \left[ (1-\Delta) \frac{S_{1}(X,\Z, q_T)-S_{2}(X,\Z, q_T,q_T^*)}{\{S_0(X,\Z,q_T)\}^2}\eta_{0p}(X_i,\Delta_i;X,\Z)\right| X_i,\Delta_i  \right]\notag\\
&& +\frac{1-R_i}{1-\pi}E_p\left.\left\{(1-\Delta)\frac{\eta_{1q}(\Z_i;X,\Z)-\eta_{2q}(\z,\X,\Z)}{S_0(X,\Z,q_T)} \right | \Z_i \right\}\notag
\\&&  - \left. \frac{1-R_i}{1-\pi}E_p\left.\left[ (1-\Delta)\frac{S_{1}(X,\Z, q_T)-S_{2}(X,\Z, q_T,q_T^*)}{\{S_0(X,\Z,q_T)\}^2}\eta_{0q}(\Z_i;X,\Z) \right| \Z_i \right] \right)\notag\\
&&  - E_p\left((1-\Delta)\left[\frac{\eta_{1p}(X,\Delta;X,\Z) - \eta_{2p}(X,\Delta;X,\Z)}{S_0(X,\Z,q_T)}-\frac{S_{1}(X,\Z, q_T)-S_{2}(X,\Z, q_T,q_T^*)}{\{S_0(X,\Z,q_T)\}^2}\eta_{0p}(X,\Delta;X,\Z)\right]\right)\notag
\\&& + o_p(n_1^{-1/2})+o_p(n_2^{-1/2}).\notag
\ee
 Next, we show that the following term in~(\ref{eq:psi3}) equals \0 for any $x$ and $\delta$
\bse
E_p\left\{(1-\Delta)\frac{\eta_{1p}(x,\delta;X,\Z) - \eta_{2p}(x,\delta;X,\Z)}{S_0(X,\Z,q_T)} \right\}=\0.
\ese 
For simplicity, we denote $\S(t,\z,\bt)\equiv \tfrac{\partial}{\partial \bt_0}q_{T \mid \Z}(t, \z ; \bt_0)/q_{T \mid \Z}(t, \z ; \bt_0)$ and $E\{\S(t,\Z,\bt_0)\mid t\}=\int \S(t,\z,\bt_0)f_{\Z\mid T}(\z,t)d\z$. Therefore,
\bse
&&\frac{\eta_{1p}(x,\delta;X,\z)-\eta_{2p}(x,\delta;X,\z)}{S_0(X,\z,q_T)}\\
&=&\frac{q_\Z(\z)}{\int_{X}^\infty p_{T,\Z}(t,\z)dt}\left(\frac{[\S(x,\z,\bt_0)
-E\{\S(x,\Z,\bt_0)\mid x\}]
I(x>X) f_{\Z\mid T}(\z,x)}{q_{\Z}(\z)}\frac{\delta}{1-P_C(x-)}\right.\\
&& + \frac{\int I(t>x) I(t>X)[\S(t,\z,\bt_0)-E\{\S(t,\Z,\bt_0)\mid t\}]f_{\Z \mid T}(\z,t)p_T(t)dt}{q_\Z(\z)\{1-P_T( x)\}\{1-P_C(x)\}}(1-\delta)\\
&&\left. -\iint \frac{I(v<x, v<w)[\S(w,\z,\bt_0)-E\{\S(w,\Z,\bt_0)\mid w\}]I(w> X)f_{\Z \mid T}(\z,w) p_T(w)p_C(v)
}{q_\Z(\z)\{1-P_T(v)\}\{1-P_C(v)\}^2} dvdw\right),
\ese
and 
\bse
&&E_p \left\{ (1-\Delta)\frac{\eta_{1p}(x,\delta;X,\Z) -
      \eta_{2p}(x,\delta;X,\Z)}{S_0(X,\Z,q_T)}\mid
  x,\delta\right\}\\
&=&E_p\left\{ I(T> C)\frac{\eta_{1p}(x,\delta;C,\Z) -
      \eta_{2p}(x,\delta;C,\Z)}{S_0(C,\Z,q_T)}\mid
  x,\delta\right\}\\
&=&\iint
[\S(x,\z,\bt_0)
-E\{\S(x,\Z,\bt_0)\mid x\}]
I(x>c) \frac{f_{\Z\mid T}(\z,x)}{p_\Z(\z)}\frac{\delta}{1-P_C(x-)}p_C(c) p_\Z(\z)dc d\z\\
&& + \iint \frac{\int I(t>x) I(t>c)
[\S(t,\z,\bt_0)
-E\{\S(t,\Z,\bt_0)\mid t\}]f_{\Z \mid T}(\z,t)p_T(t)dt}{p_\Z(\z)\{1-P_T( x)\}\{1-P_C(x)\}}p_C(c) p_\Z(\z)dc d\z(1-\delta)\\
&&  -\iint \iint \frac{I(v<x, v<w)
[\S(w,\z,\bt_0)
-E\{\S(w,\Z,\bt_0)\mid w\}]I(w> c)f_{\Z \mid T}(\z,w) p_T(w)p_C(v)
}{p_\Z(\z)\{1-P_T(v)\}\{1-P_C(v)\}^2} dvdw p_C(c)p_\Z(\z)dcd\z\\
&=&\int[\S(x,\z,\bt_0)
-E\{\S(x,\Z,\bt_0)\mid x\}]
f_{\Z\mid T}(\z,x) d\z \delta\\
&& + \int \int
[\S(t,\z,\bt_0)
-E\{\S(t,\Z,\bt_0)\mid t\}]
f_{\Z \mid T}(\z,t) d\z\frac{I(t>x)P_C(t)}{\{1-P_T( x)\}\{1-P_C( x)\}}
p_T(t)dt (1-\delta)\\
&&  -\iint  \frac{\int
[\S(w,\z,\bt_0)
-E\{\S(w,\Z,\bt_0)\mid w\}]f_{\Z
    \mid T}(\z,w) d\z  p_T(w) P_C(w) I(v<x, v<w) p_C(v)
}{\{1-P_T(v)\}\{1-P_C(v)\}^2} dvdw \\
&=&\0.
\ese

Hence, we have \be
	&&n^{-1}\sum_{i=1}^n \{\psi_{3n}(X_i,\Delta_i,R_i;\bt_0)\notag\\
	&=& n^{-1}\sum_{i=1}^n\left( \psi_{3}(X_i,\Delta_i,R_i;\bt_0)\right.\left. -\frac{R_i}{\pi}  E_p \left[ (1-\Delta) \frac{S_{1}(X,\Z, q_T)-S_{2}(X,\Z, q_T,q_T^*)}{\{S_0(X,\Z,q_T)\}^2}\eta_{0p}(X_i,\Delta_i;X,\Z)\right| X_i,\Delta_i  \right]\notag \\
&& +\frac{1-R_i}{1-\pi}E_p\left.\left\{(1-\Delta)\frac{\eta_{1q}(\Z_i;X,\Z)-\eta_{2q}(\Z_i,\X,\Z)}{S_0(X,\Z,q_T)} \right | \Z_i \right\}
\notag \\&&  - \left. \frac{1-R_i}{1-\pi}E_p\left.\left[ (1-\Delta)\frac{S_{1}(X,\Z, q_T)-S_{2}(X,\Z, q_T,q_T^*)}{\{S_0(X,\Z,q_T)\}^2}\eta_{0q}(\Z_i;X,\Z) \right| \Z_i \right] \right)\notag\\
&& - E_p\left[-(1-\Delta)\frac{S_{1}(X,\Z, q_T)-S_{2}(X,\Z, q_T,q_T^*)}{\{S_0(X,\Z,q_T)\}^2}\eta_{0p}(X,\Delta;X,\Z)\right]\notag\\
	&&+ o_p(n_1^{-1/2}) +o_p(n_2^{-1/2}).\label{eq:psi3}
	\ee

\subsubsection{Proof of Theorem 3.2}

Let $n_0 \equiv \min(n_1,n_2)$. Under Conditions \ref{A1}-\ref{A6}, and \ref{B1}-\ref{B4},
    $$\sqrt{n_0}(\wh \bt - \bt_0) \dovr N(0,\bf \Sigma),$$
    as $n_0 \rightarrow \infty$, where 
   \be
       \bf \Sigma &\equiv& E_p\{\tfrac{\partial}{\partial \bt} \psi(X,\Z,\Delta;\bt_0)\}^{-1}{\bf \Sigma}_{\psi}E_p\{\tfrac{\partial}{\partial \bt}\psi(X,\Z,\Delta;\bt_0)\}^{-\rm T},\notag \\
         \bf \Sigma_{\psi} &\equiv& \min(\tfrac{1}{\pi},\tfrac{1}{1-\pi})[(1-\pi)\var_p\{\psi(X,\Z,\Delta;\bt_0)+\psi_{p_T}(X,\Delta)\}+\pi \var_q\{\psi_{q_\Z}(\Z)\}],\notag \\
        \psi(x,\z,\delta;\bt_0) &\equiv&\delta \frac{\frac{\partial}{\partial \bt} q_{T \mid \Z}(x, \z ; \bt_0)}{q_{T \mid \Z}(x, \z ; \bt_0)} - \delta \frac{q_T^*(x)}{q_T(x)} + 
(1-\delta) \frac{S_1(x,\z,q_T)-S_2(x,\z,q_T,q^{*}_T)}{S_0(x,\z,q_T)},\label{eq:psi}\\ \psi_{p_T}(x,\delta)&\equiv& 
\left. -E_p \left[ (1-\Delta) \frac{S_{1}(X,\Z, q_T)-S_{2}(X,\Z, q_T,q_T^*)}{\{S_0(X,\Z,q_T)\}^2}\eta_{0p}(x,\delta;X,\Z)\right| x,\delta  \right],\label{eq:psip} \\
\psi_{q_\Z}(\z)&\equiv&  - E_p\left[ \left. \Delta \frac{\frac{\partial}{\partial \bt}  q_{T\mid\Z}(X,\z;\bt_0)q_T(X)- q_T^*(X) q_{T\mid\Z}(X,\z;\bt_0) }{\{q_T(X)\}^2}\right| \z \right]  \notag \\ &&  + E_p\left.\left\{(1-\Delta)\frac{\eta_{1q}(\z;X,\Z)-\eta_{2q}(\z,\X,\Z)}{S_0(X,\Z,q_T)} \right | \z \right\} \notag \\&&  - E_p\left.\left[ (1-\Delta)\frac{S_{1}(X,\Z, q_T)-S_{2}(X,\Z, q_T,q_T^*)}{\{S_0(X,\Z,q_T)\}^2}\eta_{0q}(\z;X,\Z) \right| \z \right],\label{eq:psiq}
\ee
and $q_T$, $q_T^*$ $S_i$, $\eta_{ip}$ and $\eta_{iq}$ for $i=0,1,2$
are defined in Supplement~\ref{Sup:Notations},
formulas~(\ref{eq:eta0})-(\ref{eq:eta2q}) respectively.

\begin{proof}
    The score function of one single observation dividing by $\pi$ is
    \bse
\frac{r}{\pi}\psi(x,\z,\delta;\bt) &\equiv& \frac{\partial}{\partial \bt} \ell(x, \z,\delta,r; \bt)\\
&=&\frac{r \delta}{\pi} \frac{\frac{\partial}{\partial \bt} q_{T \mid \Z}(x, \z ; \bt)}{q_{T \mid \Z}(x, \z ; \bt)} - \frac{r \delta}{\pi} \frac{\int \frac{\partial}{\partial \bt}  q_{T\mid\Z}(x,\z;\bt)dQ_\Z(\z)}{\int q_{T\mid\Z}(x,\z;\bt)dQ_\Z(\z)} \\ && + 
\frac{r (1-\delta)}{\pi} \frac{\int I(t>x)\left[ \frac{\frac{\partial}{\partial \bt} q_{T \mid Z}(t, \z ; \bt)}{\int q_{T \mid \Z}(t, \z ; \bt) d Q(z)}-\frac{q_{T \mid \Z}(t, \z ; \bt) \int \frac{\partial}{\partial \bt} q_{T\mid \Z}(t, \z ; \bt) d Q_\Z(\z)}{\{\int q_{T \mid \Z}(t, \z ; \bt) d Q_\Z(\z)\}^2}\right] d P_T(t)}{\int  \frac{I(t>x) q_{T \mid \Z}(t, \z ; \bt)}{\int q_{T \mid \Z}(t, \z ; \bt) d Q_\Z(\z)} d P_T(t)}\\
&\equiv & \psi_{1}(x,\z,\delta,r;\bt ) - \psi_{ 2}(x,\delta,r;\bt) + \psi_{3}(x,\z,\delta,r;\bt ).
\ese

In practice, the estimator $\wh \bt$ satisfies $\Psi_{n}(\wh \bt)=0$
based on estimated score function given by 
\bse
\frac{r}{\pi}\psi_n(x,\z,\delta;\bt) &\equiv& \frac{r \delta}{\pi} \frac{\frac{\partial}{\partial \bt} q_{T \mid \Z}(x, \z ; \bt)}{q_{T \mid \Z}(x, \z ; \bt)} - \frac{r \delta}{\pi} \frac{\int \frac{\partial}{\partial \bt}  q_{T\mid\Z}(x,\z;\bt)d\wh Q_\Z(\z)}{\int q_{T\mid\Z}(x,\z;\bt)d\wh Q_\Z(\z)} \\ && + 
\frac{r (1-\delta)}{\pi} \frac{\int I(t>x)\left[ \frac{\frac{\partial}{\partial \bt} q_{T \mid Z}(t, \z ; \bt)}{\int q_{T \mid \Z}(t, \z ; \bt) d \wh Q_\Z(\z)}-\frac{q_{T \mid \Z}(t, \z ; \bt) \int \frac{\partial}{\partial \bt} q_{T\mid \Z}(t, \z ; \bt) d \wh Q_\Z(\z)}{\{\int q_{T \mid \Z}(t, \z ; \bt) d \wh Q_\Z(\z)\}^2}\right] d \wh P_T(t)}{\int  \frac{I(t>x) q_{T \mid \Z}(t, \z ; \bt)}{\int q_{T \mid \Z}(t, \z ; \bt) d \wh Q_\Z(\z)} d \wh P_T(t)}\\
& \equiv & \psi_{1}(x,\z,\delta,r;\bt) - \psi_{2n}(x,\delta,r;\bt) + \psi_{3n}(x,\z,\delta,r;\bt),
\ese
and 
\be
\Psi_{n}(\bt)&\equiv& n^{-1} \sum_{i=1}^{n} \frac{R_i}{\pi} \psi_{n}(X_i,\Z_i,\Delta_i;\bt)\notag \\ 
&=& n^{-1}\sum_{i=1}^n \{ \psi_{1}(X_i,\Z_i,\Delta_i,R_i;\bt) - \psi_{2n}(X_i,\Delta_i,R_i;\bt) + \psi_{3n}(X_i,\Z_i,\Delta_i,R_i;\bt)\}.\label{eq:psin} 
\ee

Then we plug (\ref{eq:psi2}) and (\ref{eq:psi3}) into (\ref{eq:psin}) so that
\bse
&&\Psi_{n}(\bt_0)\\
&=&n^{-1}\sum_{i=1}^n \{\psi_1(X_i,\Z_i,\Delta_i,R_i;\bt_0)-\psi_2(X_i,\Delta_i,R_i;\bt_0)+\psi_3(X_i,\Z_i,\Delta_i,R_i;\bt_0)\}\\
&&-  n^{-1}\sum_{i=1}^n \{\psi_{2n}(X_i,\Delta_i,R_i;\bt_0) - \psi_{2}(X_i,\Delta_i,R_i;\bt_0)\} +  n^{-1}\sum_{i=1}^n \{\psi_{3n}(X_i,\Z_i,\Delta_i,R_i;\bt_0) - \psi_{3}(X_i,\Z_i,\Delta_i,R_i;\bt_0)\}\\
&=& n^{-1}\sum_{i=1}^n\psi(X_i,\Z_i,\Delta_i,R_i;\bt_0)-  n^{-1} \sum_{i=1}^n \frac{1-R_i}{1-\pi} E \left[ \left. \frac{R\Delta}{\pi} \frac{\frac{\partial}{\partial \bt}  q_{T\mid\Z}(X,\Z_i;\bt_0)q_T(X)- q_T^*(X) q_{T\mid\Z}(X,\Z_i;\bt_0) }{\{q_T(X)\}^2}\right| \Z_i \right]\\
&& \left. -\frac{R_i}{\pi}  E_p \left[ (1-\Delta) \frac{S_{1}(X,\Z, q_T)-S_{2}(X,\Z, q_T,q_T^*)}{\{S_0(X,\Z,q_T)\}^2}\eta_{0p}(X_i,\Delta_i;X,\Z)\right| X_i,\Delta_i  \right]\notag\\
&& +\frac{1-R_i}{1-\pi}E_p\left.\left\{(1-\Delta)\frac{\eta_{1q}(\Z_i;X,\Z)-\eta_{2q}(\Z_i,\X,\Z)}{S_0(X,\Z,q_T)} \right | \Z_i \right\}\notag
\\&&  - \left. \frac{1-R_i}{1-\pi}E_p\left.\left[ (1-\Delta)\frac{S_{1}(X,\Z, q_T)-S_{2}(X,\Z, q_T,q_T^*)}{\{S_0(X,\Z,q_T)\}^2}\eta_{0q}(\Z_i;X,\Z) \right| \Z_i \right] \right)\notag\\
&& - E_p\left[-(1-\Delta)\frac{S_{1}(X,\Z, q_T)-S_{2}(X,\Z, q_T,q_T^*)}{\{S_0(X,\Z,q_T)\}^2}\eta_{0p}(X,\Delta;X,\Z)\right]\\
&=& n^{-1}\sum_{i=1}^n\psi(X_i,\Z_i,\Delta_i,R_i;\bt_0) +n^{-1}\sum_{i=1}^n \frac{R_i}{\pi}\psi_{p_T}(X_i,\Delta_i)
-E_p \{\psi_{p_T}(X,\Delta)\}\\
&&+n^{-1}\sum_{i=1}^n \frac{1-R_i}{1-\pi}\psi_{q_\Z}(Z_i) + o_p(n_1^{-1/2})+o_p(n_2^{-1/2}),
\ese
where
\bse
\psi_{p_T}(x,\delta)&\equiv&
\left. -E_p \left[ (1-\Delta) \frac{S_{1}(X,\Z, q_T)-S_{2}(X,\Z, q_T,q_T^*)}{\{S_0(X,\Z,q_T)\}^2}\eta_{0p}(x,\delta;X,\Z)\right| x,\delta  \right],\\
\psi_{q_\Z}(\z)&\equiv& - E_p\left[ \left. \Delta \frac{\frac{\partial}{\partial \bt}  q_{T\mid\Z}(X,\z;\bt_0)q_T(X)- q_T^*(X) q_{T\mid\Z}(X,\z;\bt_0) }{\{q_T(X)\}^2}\right| \z \right] \\ && + E_p\left.\left\{(1-\Delta)\frac{\eta_{1q}(\z;X,\Z)-\eta_{2q}(\z,\X,\Z)}{S_0(X,\Z,q_T)} \right | \z \right\}\\&& - E_p\left.\left[ (1-\Delta)\frac{S_{1}(X,\Z, q_T)-S_{2}(X,\Z, q_T,q_T^*)}{\{S_0(X,\Z,q_T)\}^2}\eta_{0q}(\z;X,\Z) \right| \z \right].
\ese
 Also, by the definitions of  $q_T ,q_T^*$, $S_0, S_1, S_2$ in Section
   \ref{Sup:Notations}, we can simplify $\psi(x,\z,\delta;\bt_0)$ as
\bse
 \frac{r}{\pi}\psi(x,\z,\delta;\bt_0)&=&\frac{r \delta}{\pi} \frac{\frac{\partial}{\partial \bt_0} q_{T \mid \Z}(x, \z ; \bt_0)}{q_{T \mid \Z}(x, \z ; \bt_0)} - \frac{r \delta}{\pi} \frac{\int \frac{\partial}{\partial \bt_0}  q_{T\mid\Z}(x,\z;\bt_0)dQ_\Z(\z)}{\int q_{T\mid\Z}(x,\z;\bt_0)dQ_\Z(\z)} \\ &&  + 
\frac{r (1-\delta)}{\pi} \frac{\int I(t>x)\left[ \frac{\frac{\partial}{\partial \bt_0} q_{T \mid Z}(t, \z ; \bt_0)}{\int q_{T \mid \Z}(t, \z ; \bt_0) d Q(z)}-\frac{q_{T \mid \Z}(t, \z ; \bt_0) \int \frac{\partial}{\partial \bt_0} q_{T\mid \Z}(t, \z ; \bt_0) d Q_\Z(\z)}{\{\int q_{T \mid \Z}(t, \z ; \bt_0) d Q_\Z(\z)\}^2}\right] d P_T(t)}{\int  \frac{I(t>x) q_{T \mid \Z}(t, \z ; \bt_0)}{\int q_{T \mid \Z}(t, \z ;  \bt_0) d Q_\Z(\z)} d P_T(t)}\\ &=&\frac{r \delta}{\pi} \frac{\frac{\partial}{\partial \bt} q_{T \mid \Z}(x, \z ; \bt_0)}{q_{T \mid \Z}(x, \z ; \bt_0)} - \frac{r \delta}{\pi} \frac{q_T^*(x)}{q_T(x)} + 
\frac{r (1-\delta)}{\pi} \frac{S_1(x,\z,q_T)-S_2(x,\z,q_T,q^{*}_T)}{S_0(x,\z,q_T)}.\ese
Combining the results, we obtain
\bse
\sqrt{n_0}\Psi_{n}(\bt_0)&=&n^{-1/2} \sum_{i=1}^n \sqrt{\min(\pi,1-\pi)}\frac{R_i}{\pi}\psi(X_i,\Z_i,\Delta_i;\bt_0)\\ 
&& +n^{-1/2} \sum_{i=1}^n  \sqrt{\min(\pi,1-\pi)} \left[
  \frac{R_i}{\pi}\psi_{p_T}(X_i,\Delta_i) -
  E_p\{\psi_{p_T}(X,\Delta)\}  + \frac{1-R_i}{1-\pi} \psi_{q_\Z}(\Z_i)
\right]+o_p(1)\\ 
&\dovr& N(0,\bf \Sigma_{\psi}),
\ese
where
\bse
\bf \Sigma_{\psi} &\equiv& \var\left[\sqrt{\min(\pi,1-\pi)}\left\{\frac{R}{\pi}\psi(X,\Z,\Delta;\bt_0)+ \frac{R}{\pi}\psi_{p_T}(X,\Delta)+  \frac{1-R}{1-\pi} \psi_{q_\Z}(\Z)\right\}\right].
\ese
Hence,  $\sqrt{n_0}(\wh \bt - \bt_0) \dovr N(0,\bf \Sigma),$ as $n_0
\rightarrow \infty$, where ${\bf \Sigma} \equiv E\{\frac{\partial}{\partial
  \bt}\psi(\bt_0)\}^{-1}{\bf \Sigma}_{\psi}E\{\frac{\partial}{\partial
  \bt}\psi(\bt_0)\}^{-\rm T}$.
  
Note that ${\bf \Sigma_{\psi}}$ can be further written as
\bse
{\bf \Sigma_{\psi}} &=&\var\left[\sqrt{\min(\pi,1-\pi)}\left\{\frac{R}{\pi}\psi(X,\Z,\Delta;\bt_0)+ \frac{R}{\pi}\psi_{p_T}(X,\Delta)+  \frac{1-R}{1-\pi} \psi_{q_\Z}(\Z)\right\}\right]\\
&=&\var\left(\left. E\left[\sqrt{\min(\pi,1-\pi)}\left\{ \frac{R}{\pi}\psi(X,\Z,\Delta;\bt_0)+ \frac{R}{\pi}\psi_{p_T}(X,\Delta)+  \frac{1-R}{1-\pi} \psi_{q_\Z}(\Z)\right\}\right| R\right]\right)\\
&& +E\left(\left. \var\left[\sqrt{\min(\pi,1-\pi)}\left\{\frac{R}{\pi}\psi(X,\Z,\Delta;\bt_0)+ \frac{R}{\pi}\psi_{p_T}(X,\Delta)+  \frac{1-R}{1-\pi} \psi_{q_\Z}(\Z)\right\}\right| R\right]\right)\\
&=&  \min(\pi,1-\pi)\left(\frac{1}{\pi} E_p[\{\psi(X,\Z,\Delta;\bt_0)+\psi_{p_T}(X,\Delta)\}^2]+\frac{1}{1-\pi} E_q\{\psi_{q_\Z}(\Z)^2\}\right)\\
&=& \min\left(\frac{1}{\pi},\frac{1}{1-\pi}\right)\left((1-\pi)E_p[\{\psi(X,\Z,\Delta;\bt_0)+\psi_{p_T}(X,\Delta)\}^2]+\pi E_q\{\psi_{q_\Z}(\Z)^2\}\right)\\
&=& \min\left(\frac{1}{\pi},\frac{1}{1-\pi}\right)\left[(1-\pi)\var_p\{\psi(X,\Z,\Delta;\bt_0)+\psi_{p_T}(X,\Delta)\}+\pi \var_q\{\psi_{q_\Z}(\Z)\}\right],
\ese
where the third equality holds since $E_p\{\psi(X,\Z,\Delta;\bt_0)\}=0$, $E_p\{\psi_{p_T}(X,\Delta)\}=0$ and $E_q\{\psi_{q_\Z}(\Z)\}=0$.
We summarize the resulting expressions under specific values of \( \pi \) as follows: if $\pi < 1/2$, 
$$
{\bf \Sigma_{\psi}} = \var_p\{\psi(X,\Z,\Delta;\bt_0)+\psi_{p_T}(X,\Delta)\} + \frac{\pi}{1-\pi} \var_q\{\psi_{q_\Z}(\Z)\},
$$
if $\pi = 1/2$,
$$
{\bf \Sigma_{\psi}} = \var_p\{\psi(X,\Z,\Delta;\bt_0)+ \psi_{p_T}(X,\Delta)\} + \var_q\{\psi_{q_\Z}(\Z)\},
$$
and if $\pi > 1/2$,
$$
{\bf \Sigma_{\psi}} = \frac{1-\pi}{\pi} [\var_p\{\psi(X,\Z,\Delta;\bt_0)+ \psi_{p_T}(X,\Delta)\}] + \var_q\{\psi_{q_\Z}(\Z)\}.
$$
Hence, if $\pi=0$,
$$
{\bf \Sigma_{\psi}} = \var_p\{\psi(X,\Z,\Delta;\bt_0)+ \psi_{p_T}(X,\Delta)\},
$$
and if $\pi=1$,
$$
{\bf \Sigma_{\psi}} = \var_q\{\psi_{q_\Z}(\Z)\}.$$
  
\end{proof}

\subsubsection{Proof of Corollary 3.1}

\begin{proof}

By the continuous mapping theorem and the result of Theorem 1, we have $\wh \zeta(\z) \povr \zeta(\z)$.

Furthermore, by applying the delta method and the asymptotic normality established in Theorem~3.2, we obtain the asymptotic distribution of $\wh \zeta(\z)$, that is, $$\sqrt{n_0}\{\wh \zeta(\z) -  \zeta(\z)\} \dovr N(0,\bf \Sigma_{\zeta(\z)}),$$
where
$\gamma_{\zeta(\z)} \equiv \tfrac{\partial}{\partial \bt}\int g(t) q_{T\mid \Z}(t,\z; \bt_0)dt$, and
${\bf \Sigma}_{\zeta(\z)} \equiv \gamma_{\zeta(\z)} \Sigma \gamma_{\zeta(\z)}\trans.$
\end{proof}

\subsubsection{Variance Estimation Method}

\label{ssup:varest}

The asymptotic variance consists of three components: the variances of $\psi$, $\psi_{p_T}$ and $\psi_{q_\Z}$, each involving several unknown quantities. In this subsection, we provide the details for estimating the asymptotic variance. Specifically, expectations are approximated using empirical averages, and the true value $\bt_0$ is replaced by its estimator $\wh \bt$. The remaining unknown components are estimated using the following formulas:
\bse
 \wh q_T(t) &\equiv& \int q_{T\mid \Z}(t,\z;\bt_0)d \wh Q_\Z(\z),\quad 
\wh{q}_T^*(t) \equiv \tfrac{\partial}{\partial \bt} \int q_{T\mid \Z}(t,\z;\bt_0)d \wh Q_\Z(\z),\\
 S_{0n}(x,\z,\wh q_T) &\equiv& \int \frac{I(t>x)q_{T \mid \Z}(t,\z;\bt_0)}{\wh q_T(t)}d \wh P_T(t),\quad
 S_{1n}(x,\z,\wh q_T) \equiv \int \frac{I(t>x)\tfrac{\partial}{\partial \bt} q_{T \mid \Z}(t,\z;\bt_0)}{\wh q_T(t)}d \wh P_T(t),\\
 S_{2n}(x,\z,\wh q_T,\wh q^{*}_T) &\equiv& \int \frac{I(t>x) q_{T \mid \Z}(t, \z ; \bt_0) \wh q^{*}_T(t)}{\wh q^2_T(t)} d \wh P_T(t),\\
 \wh \varphi_0(t;x, \z) &=& \frac{I(t>x) q_{T \mid \Z}(t,\z;\bt_0)}{\wh q_T(t)},\  \wh \varphi_1( t;x, \z) \equiv \frac{I(t>x)\tfrac{\partial}{\partial \bt}  q_{T \mid \Z}(t,\z;\bt_0)}{\wh q_T(t)},\\
 \wh \varphi_2(t;x, \z) &\equiv& \frac{I(t>x) q_{T \mid \Z}(t, \z ; \bt_0)\wh{q}_T^*(t) }{\wh q_T^2(t)},\\
\wh \eta_{0q}(\Z;x,\z)&\equiv& -\int  \frac{I(t>x) q_{T \mid \Z}(t, \z ; \bt_0)\{q_{T \mid \Z}(t,\Z;\bt_0)-\wh q_T(t)\}}{\wh q^2_T(t)}d \wh P_T(t),\\
\wh \eta_{1q}(\Z;x,\z)&=& -\int  \frac{I(t>x) \tfrac{\partial}{\partial \bt} q_{T \mid \Z}(t, \z ; \bt_0)\{q_{T \mid \Z}(t,\Z;\bt_0)-\wh q_T(t)\}}{\wh q^2_T(t)}d \wh P_T(t),\\
\wh \eta_{2q}(\Z;x,\z)&=&  \int  \frac{I(t>x) q_{T \mid \Z}(t, \z ; \bt_0)\{\tfrac{\partial}{\partial \bt} q_{T \mid \Z}(t,\Z;\bt_0)-\wh q^*_T(t)\}}{\wh q^2_T(t)}d \wh P_T(t) \notag \\
 &&- \int  \frac{2 I(t>x) q_{T \mid \Z}(t, \z ; \bt_0)\wh q^*_T(t) \{q_{T \mid \Z}(t,\Z;\bt_0)-\wh q_T(t)\}}{\wh q^3_T(t)}d \wh P_T(t),\\
\wh \eta_{ip}(X,\Delta; x,\z) &\equiv& \wh \varphi_i(X; x,\z) \wh \gamma_0(X)\Delta + \wh \gamma_1^{\varphi_i}(X)(1 - \Delta) - \wh \gamma_2^{\varphi_i}(X), \quad i = 0,1,2,\\
\wt{H}_n^j(z)&=&n_1^{-1} \sum_{i=1}^{n_1} I\{X_i \leq z, \delta_i=j\}, \quad j=0,1,\\
\wh \gamma_0(X) &=& \exp\left\{\int_{-\infty}^{X-} \frac{\wt H^0_n(d w)}{n_1^{-1}\sum_{i=1}^{n_1}I\{X_i\geq w\}}\right\},\\
\wh \gamma_1^{\varphi_i}(X) &=& \frac{1}{n_1^{-1}\sum_{i=1}^{n_1}I\{X_i\geq X\}} \int I{(X<w)} \wh \varphi_i(w;x,\z) \wh \gamma_0(w) \wt{H}^1_n(d w),\\
\wh \gamma_2^{\varphi_i}(X) &=& \iint \frac{I{(v<X, v<w)} \wh \varphi_i(w;x,\z) \gamma_0(w)}{(n_1^{-1}\sum_{i=1}^{n_1}I\{X_i\geq w\})^2} \wt{H}_n^0(d v) \wt{H}_n^1(d w).
\ese

\section{Tables}
\label{sup:tab}

\begin{table}[ht]
\small
\centering
\renewcommand\arraystretch{1.2}
\setlength{\tabcolsep}{12pt}
\caption{Simulation results for the proportional odds model with
  log-logistic ($\mu$,$\sigma$) baseline survival function and with $\beta_1=1$, $\beta_2=-2$, $\mu=-2$ and $\sigma=0.5$.}
\label{Tab:PO}
\begin{tabular}{cccccccc}
  \hline \hline
$n_1$ & $n_2$ & & MSE & Bias & SE & $\wh{\text{SE}}$ & CP \\ 
  \hline
250 & 500 & $\beta_1$ & 0.0096 & -0.0061 & 0.0980 & 0.0900 & 0.9280\\ 
   &  & $\beta_2$ & 0.0221 & 0.0160 & 0.1478 & 0.1353 & 0.9220\\ 
   &  & $\mu$ & 0.0104 & -0.0061 & 0.1021 & 0.1029 & 0.9500 \\ 
   &  & $\sigma$ & 0.0014 & -0.0091 & 0.0370 & 0.0376 & 0.9320 \\ \hline
  500 & 250 & $\beta_1$ & 0.0092 & -0.0025 & 0.0959 & 0.0912 & 0.9460 \\ 
   &  & $\beta_2$ & 0.0259 & 0.0121 & 0.1606 & 0.1368 & 0.9100 \\ 
   &  & $\mu$ & 0.0079 & -0.0112 & 0.0883 & 0.0830 & 0.9260 \\ 
   &  & $\sigma$ & 0.0013 & -0.0092 & 0.0344 & 0.0314 & 0.9080 \\ \hline
  500 & 500 & $\beta_1$ & 0.0060 & -0.0045 & 0.0774 & 0.0740 & 0.9400 \\ 
   &  & $\beta_2$ & 0.0141 & 0.0123 & 0.1181 & 0.1112 & 0.9480 \\ 
   &  & $\mu$ & 0.0056 & -0.0103 & 0.0742 & 0.0768 & 0.9480 \\ 
   &  & $\sigma$ & 0.0008 & -0.0050 & 0.0287 & 0.0282 & 0.9240 \\ \hline
  500 & 750 & $\beta_1$ & 0.0049 & -0.0031 & 0.0696 & 0.0671 & 0.9300 \\ 
   &  & $\beta_2$ & 0.0112 & 0.0071 & 0.1057 & 0.1010 & 0.9400 \\ 
   &  & $\mu$ & 0.0050 & -0.0049 & 0.0708 & 0.0740 & 0.9520 \\ 
   &  & $\sigma$ & 0.0008 & -0.0024 & 0.0278 & 0.0271 & 0.9400 \\ \hline
  750 & 500 & $\beta_1$ & 0.0048 & -0.0064 & 0.0688 & 0.0676 & 0.9460 \\ 
   &  & $\beta_2$ & 0.0114 & 0.0185 & 0.1051 & 0.1014 & 0.9320 \\ 
   &  & $\mu$ & 0.0045 & -0.0045 & 0.0667 & 0.0647 & 0.9380 \\ 
   &  & $\sigma$ & 0.0007 & -0.0075 & 0.0247 & 0.0241 & 0.9280 \\ \hline
  500 & 1000 & $\beta_1$ & 0.0040 & -0.0052 & 0.0630 & 0.0627 &  0.9260 \\ 
   &  & $\beta_2$ & 0.0097 & 0.0150 & 0.0976 & 0.0939 & 0.9220 \\ 
   &  & $\mu$ & 0.0057 & -0.0034 & 0.0757 & 0.0721 & 0.9440 \\ 
   &  & $\sigma$ & 0.0008 & -0.0056 & 0.0271 & 0.0262 & 0.9300 \\ \hline
  1000 & 500 & $\beta_1$ & 0.0046 & -0.0064 & 0.0677 & 0.0645 & 0.9420 \\ 
   &  & $\beta_2$ & 0.0112 & 0.0128 & 0.1051 & 0.0973 & 0.9160\\ 
   &  & $\mu$ & 0.0039 & -0.0065 & 0.0622 & 0.0586 & 0.9360 \\ 
   &  & $\sigma$ & 0.0006 & -0.0053 & 0.0230 & 0.0220 & 0.9340 \\ \hline
  1000 & 1000 & $\beta_1$ & 0.0026 & -0.0035 & 0.0509 & 0.0517 & 0.9540 \\ 
   &  & $\beta_2$ & 0.0067 & 0.0103 & 0.0813 & 0.0777 & 0.9340 \\ 
   &  & $\mu$ & 0.0031 & -0.0044 & 0.0554 & 0.0536 & 0.9420 \\ 
   &  & $\sigma$ & 0.0004 & -0.0028 & 0.0199 & 0.0197 & 0.9540 \\ 
   \hline
\hline
\end{tabular}
\begin{tablenotes}
    \item $n_1$: sample size in $\calP$; $n_2$: sample size in
      $\calQ$; MSE: empirical mean squared error; SE: empirical standard error; 
       $\wh{\text{SE}}$: mean of estimated standard errors; CP:
      empirical coverage probability of 95\% confidence intervals.
\end{tablenotes}
\end{table}

\begin{table}[ht]
\small
\centering
\renewcommand\arraystretch{1.2}
\setlength{\tabcolsep}{12pt}
\caption{Simulation results for the accelerated failure time model with log-normal baseline hazard function and with $\beta_1=1.5$, $\beta_2=-1.5$, $\mu=-1.5$ and $\sigma=0.5$}
\label{Tab:AFT}
\begin{tabular}{cccccccc}
  \hline \hline
$n_1$ & $n_2$ & & MSE & Bias & SE & $\wh{\text{SE}}$ & CP \\ 
 \hline
250 & 500 & $\beta_1$ & 0.0030 & -0.0101 & 0.0538 & 0.0555 & 0.9420 \\ 
   &  & $\beta_2$ & 0.0031 & 0.0095 & 0.0545 & 0.0556 & 0.9320 \\ 
   &  & $\mu$ & 0.0027 & -0.0036 & 0.0515 & 0.0525 & 0.9580 \\ 
   &  & $\sigma$ & 0.0008 & -0.0068 & 0.0283 & 0.0288 &  0.9160\\ \hline
  500 & 250 & $\beta_1$ & 0.0025 & -0.0111 & 0.0493 & 0.0465 &  0.9280\\ 
   &  & $\beta_2$ & 0.0027 & 0.0092 & 0.0511 & 0.0473 & 0.9320 \\ 
   &  & $\mu$ & 0.0018 & -0.0004 & 0.0429 & 0.0411 &  0.9280\\ 
   &  & $\sigma$ & 0.0006 & -0.0070 & 0.0225 & 0.0209 & 0.8900 \\ \hline
  500 & 500 & $\beta_1$ & 0.0019 & -0.0089 & 0.0421 & 0.0418 &  0.9240 \\ 
   &  & $\beta_2$ & 0.0016 & 0.0070 & 0.0397 & 0.0419 & 0.9560 \\ 
   &  & $\mu$ & 0.0015 & -0.0041 & 0.0391 & 0.0383 & 0.9440 \\ 
   &  & $\sigma$ & 0.0004 & -0.0028 & 0.0199 & 0.0203 & 0.9400 \\ \hline
  500 & 750 & $\beta_1$ & 0.0016 & -0.0058 & 0.0392 & 0.0404 & 0.9520 \\ 
   &  & $\beta_2$ & 0.0017 & 0.0069 & 0.0412 & 0.0403 & 0.9360 \\ 
   &  & $\mu$ & 0.0015 & -0.0030 & 0.0389 & 0.0374 & 0.9460 \\ 
   &  & $\sigma$ & 0.0004 & -0.0006 & 0.0202 & 0.0202 & 0.9560 \\ \hline
  750 & 500 & $\beta_1$ & 0.0013 & -0.0097 & 0.0354 & 0.0362 & 0.9420 \\ 
   &  & $\beta_2$ & 0.0014 & 0.0093 & 0.0368 & 0.0361 & 0.9380 \\ 
   &  & $\mu$ & 0.0010 & -0.0010 & 0.0321 & 0.0323 & 0.9660 \\ 
   &  & $\sigma$ & 0.0003 & -0.0034 & 0.0169 & 0.0166 & 0.9380 \\ \hline
  500 & 1000 & $\beta_1$ & 0.0017 & -0.0060 & 0.0411 & 0.0390 & 0.9320 \\ 
   &  & $\beta_2$ & 0.0017 & 0.0093 & 0.0397 & 0.0391 & 0.9520 \\ 
   &  & $\mu$ & 0.0014 & -0.0055 & 0.0372 & 0.0369 & 0.9460 \\ 
   &  & $\sigma$ & 0.0004 & -0.0004 & 0.0210 & 0.0202 & 0.9400 \\ \hline
  1000 & 500 & $\beta_1$ & 0.0012 & -0.0056 & 0.0339 & 0.0332 & 0.9420 \\ 
   &  & $\beta_2$ & 0.0012 & 0.0066 & 0.0340 & 0.0333 & 0.9480 \\ 
   &  & $\mu$ & 0.0008 & -0.0024 & 0.0277 & 0.0289 & 0.9560 \\ 
   &  & $\sigma$ & 0.0002 & -0.0031 & 0.0151 & 0.0146 &  0.9280 \\ \hline
  1000 & 1000 & $\beta_1$ & 0.0009 & -0.0057 & 0.0299 & 0.0295 & 0.9420 \\ 
   &  & $\beta_2$ & 0.0009 & 0.0046 & 0.0295 & 0.0294 & 0.9460 \\ 
   &  & $\mu$ & 0.0007 & 0.0002 & 0.0257 & 0.0268 & 0.9660 \\ 
   &  & $\sigma$ & 0.0002 & -0.0020 & 0.0139 & 0.0142 & 0.9580 \\ 
   \hline
\hline
\end{tabular}
\begin{tablenotes}
    \item $n_1$: sample size in $\calP$; $n_2$: sample size in
      $\calQ$; MSE: empirical mean squared error; SE: empirical standard error; 
       $\wh{\text{SE}}$: mean of estimated standard errors; CP:
      empirical coverage probability of 95\% confidence intervals.
\end{tablenotes}
\end{table}

\begin{table}[ht]
\small
\centering
\renewcommand\arraystretch{1.2}
\setlength{\tabcolsep}{12pt}
\caption{Simulation results for the accelerated hazards model with Weibull($\lambda$,$\gamma$)  baseline hazard function and with $\beta_1=1$, $\beta_2=-1$, $\gamma=2$ and $\lambda=1.5$.}
\label{Tab:AH}
\begin{tabular}{cccccccc}
  \hline \hline
$n_1$ & $n_2$ & & MSE & Bias & SE & $\wh{\text{SE}}$ & CP \\ 
 \hline
250 & 500 & $\beta_1$ & 0.0159 & 0.0378 & 0.1207 & 0.1412 & 0.9567 \\ 
   &  & $\beta_2$ & 0.0155 & 0.0400 & 0.1181 & 0.1418 & 0.9764 \\ 
   &  & $\gamma$ & 0.0200 & -0.0542 & 0.1308 & 0.1542 & 0.9724 \\ 
   &  & $\lambda$ & 0.0327 & 0.0732 & 0.1655 & 0.1798 & 0.9606 \\ \hline
  500 & 250 & $\beta_1$ & 0.0122 & 0.0353 & 0.1047 & 0.1147 & 0.9732 \\ 
   &  & $\beta_2$ & 0.0138 & 0.0454 & 0.1085 & 0.1152 & 0.9554 \\ 
   &  & $\gamma$ & 0.0148 & -0.0478 & 0.1120 & 0.1217 & 0.9375 \\ 
   &  & $\lambda$ & 0.0196 & 0.0356 & 0.1358 & 0.1530 & 0.9732 \\ \hline
  500 & 500 & $\beta_1$ & 0.0093 & 0.0135 & 0.0956 & 0.0989 & 0.9354 \\ 
   &  & $\beta_2$ & 0.0091 & 0.0133 & 0.0945 & 0.0989 & 0.9392 \\ 
   &  & $\gamma$ & 0.0121 & -0.0300 & 0.1060 & 0.1130 & 0.9392 \\ 
   &  & $\lambda$ & 0.0201 & 0.0585 & 0.1293 & 0.1360 & 0.9506 \\ \hline
  500 & 750 & $\beta_1$ & 0.0092 & 0.0234 & 0.0931 & 0.0944 & 0.9449 \\ 
   &  & $\beta_2$ & 0.0091 & 0.0208 & 0.0935 & 0.0951 & 0.9449 \\ 
   &  & $\gamma$ & 0.0122 & -0.0359 & 0.1045 & 0.1082 & 0.9485 \\ 
   &  & $\lambda$ & 0.0185 & 0.0476 & 0.1278 & 0.1289 & 0.9485 \\ \hline
  750 & 500 & $\beta_1$ & 0.0059 & 0.0116 & 0.0764 & 0.0843 & 0.9580 \\ 
   &  & $\beta_2$ & 0.0063 & 0.0149 & 0.0782 & 0.0849 & 0.9542 \\ 
   &  & $\gamma$ & 0.0080 & -0.0284 & 0.0847 & 0.0962 & 0.9580 \\ 
   &  & $\lambda$ & 0.0172 & 0.0389 & 0.1255 & 0.1166 & 0.9275 \\ \hline
  500 & 1000 & $\beta_1$ & 0.0092 & 0.0229 & 0.0934 & 0.0934 & 0.9368 \\ 
   &  & $\beta_2$ & 0.0086 & 0.0221 & 0.0905 & 0.0932 & 0.9480 \\ 
   &  & $\gamma$ & 0.0111 & -0.0371 & 0.0989 & 0.1068 & 0.9405 \\ 
   &  & $\lambda$ & 0.0177 & 0.0495 & 0.1236 & 0.1237 & 0.9442 \\ \hline
  1000 & 500 & $\beta_1$ & 0.0060 & 0.0093 & 0.0772 & 0.0770 & 0.9585 \\ 
   &  & $\beta_2$ & 0.0059 & 0.0146 & 0.0753 & 0.0773 & 0.9472 \\ 
   &  & $\gamma$ & 0.0068 & -0.0269 & 0.0783 & 0.0867 & 0.9623 \\ 
   &  & $\lambda$ & 0.0138 & 0.0413 & 0.1101 & 0.1078 & 0.9472 \\ \hline
  1000 & 1000 & $\beta_1$ & 0.0040 & 0.0099 & 0.0628 & 0.0673 & 0.9685 \\ 
   &  & $\beta_2$ & 0.0040 & 0.0105 & 0.0622 & 0.0672 & 0.9650 \\ 
   &  & $\gamma$ & 0.0053 & -0.0249 & 0.0688 & 0.0787 & 0.9720 \\ 
   &  & $\lambda$ & 0.0089 & 0.0315 & 0.0894 & 0.0932 & 0.9615 \\ 
   \hline
\hline
\end{tabular}
\begin{tablenotes}
    \item $n_1$: sample size in $\calP$; $n_2$: sample size in
      $\calQ$; MSE: empirical mean squared error; SE: empirical standard error; 
       $\wh{\text{SE}}$: mean of estimated standard errors; CP:
      empirical coverage probability of 95\% confidence intervals.
\end{tablenotes}
\end{table}

\begin{table}[ht]
\scriptsize
\centering
\renewcommand\arraystretch{1}
\setlength{\tabcolsep}{12pt}
\caption{Simulation results for the proportional hazards model with Weibull($\lambda$,$\gamma$)  baseline hazard function and with $\beta_1=1$, $\beta_2=1$, $\beta_3=-1$, $\beta_4=1$, $\gamma=1.5$ and $\lambda=1$.}
\label{Tab:PH4}
\begin{tabular}{cccccccc}
  \hline \hline
$n_1$ & $n_2$ & & MSE & Bias & SE & $\wh{\text{SE}}$ & CP \\ 
 \hline
250 & 500 & $\beta_1$ & 0.0108 & 0.0101 & 0.1038 & 0.1032 & 0.9360 \\ 
   &  & $\beta_2$ & 0.0105 & 0.0090 & 0.1023 & 0.1030 & 0.9600 \\ 
   &  & $\beta_3$ & 0.0133 & -0.0066 & 0.1153 & 0.1143 & 0.9440 \\ 
   &  & $\beta_4$ & 0.0286 & 0.0135 & 0.1688 & 0.1645 & 0.9500 \\ 
   &  & $\gamma$ & 0.0173 & 0.0314 & 0.1278 & 0.1300 & 0.9660 \\ 
   &  & $\lambda$ & 0.0341 & 0.0364 & 0.1813 & 0.1775 & 0.9480 \\ \hline
  500 & 250 & $\beta_1$ & 0.0057 & 0.0024 & 0.0757 & 0.0760 & 0.9520 \\ 
   &  & $\beta_2$ & 0.0058 & 0.0023 & 0.0761 & 0.0763 & 0.9500 \\ 
   &  & $\beta_3$ & 0.0073 & -0.0003 & 0.0856 & 0.0859 & 0.9460 \\ 
   &  & $\beta_4$ & 0.0145 & 0.0108 & 0.1200 & 0.1207 & 0.9540 \\ 
   &  & $\gamma$ & 0.0101 & 0.0172 & 0.0992 & 0.1004 & 0.9440 \\ 
   &  & $\lambda$ & 0.0198 & 0.0191 & 0.1395 & 0.1374 & 0.9440 \\ \hline
  500 & 500 & $\beta_1$ & 0.0056 & 0.0048 & 0.0747 & 0.0725 & 0.9360 \\ 
   &  & $\beta_2$ & 0.0053 & 0.0071 & 0.0724 & 0.0728 & 0.9560 \\ 
   &  & $\beta_3$ & 0.0070 & -0.0103 & 0.0830 & 0.0814 & 0.9480 \\ 
   &  & $\beta_4$ & 0.0123 & 0.0045 & 0.1111 & 0.1158 & 0.9580 \\ 
   &  & $\gamma$ & 0.0097 & 0.0231 & 0.0961 & 0.0930 & 0.9380 \\ 
   &  & $\lambda$ & 0.0164 & 0.0213 & 0.1264 & 0.1270 & 0.9680 \\ \hline
  500 & 750 & $\beta_1$ & 0.0050 & 0.0063 & 0.0704 & 0.0715 & 0.9460 \\ 
   &  & $\beta_2$ & 0.0047 & 0.0064 & 0.0686 & 0.0714 & 0.9780 \\ 
   &  & $\beta_3$ & 0.0065 & -0.0045 & 0.0808 & 0.0797 & 0.9400 \\ 
   &  & $\beta_4$ & 0.0128 & 0.0049 & 0.1132 & 0.1141 & 0.9620 \\ 
   &  & $\gamma$ & 0.0092 & 0.0207 & 0.0938 & 0.0898 & 0.9260 \\ 
   &  & $\lambda$ & 0.0149 & 0.0151 & 0.1212 & 0.1236 & 0.9580 \\ \hline
  750 & 500 & $\beta_1$ & 0.0037 & 0.0006 & 0.0610 & 0.0600 & 0.9420 \\ 
   &  & $\beta_2$ & 0.0038 & -0.0023 & 0.0615 & 0.0597 & 0.9440 \\ 
   &  & $\beta_3$ & 0.0051 & -0.0024 & 0.0712 & 0.0678 & 0.9340 \\ 
   &  & $\beta_4$ & 0.0100 & -0.0005 & 0.1003 & 0.0952 & 0.9400 \\ 
   &  & $\gamma$ & 0.0061 & 0.0158 & 0.0763 & 0.0780 & 0.9440 \\ 
   &  & $\lambda$ & 0.0125 & 0.0233 & 0.1093 & 0.1067 & 0.9460 \\ \hline
  500 & 1000 & $\beta_1$ & 0.0047 & 0.0069 & 0.0686 & 0.0707 & 0.9600 \\ 
   &  & $\beta_2$ & 0.0044 & 0.0055 & 0.0661 & 0.0707 & 0.9680 \\ 
   &  & $\beta_3$ & 0.0061 & -0.0057 & 0.0779 & 0.0789 & 0.9560 \\ 
   &  & $\beta_4$ & 0.0118 & 0.0118 & 0.1082 & 0.1139 & 0.9680 \\ 
   &  & $\gamma$ & 0.0073 & 0.0130 & 0.0846 & 0.0884 & 0.9620 \\ 
   &  & $\lambda$ & 0.0140 & 0.0139 & 0.1177 & 0.1212 & 0.9560 \\ \hline
  1000 & 500 & $\beta_1$ & 0.0029 & 0.0047 & 0.0534 & 0.0528 & 0.9580 \\ 
   &  & $\beta_2$ & 0.0027 & 0.0051 & 0.0515 & 0.0528 & 0.9540 \\ 
   &  & $\beta_3$ & 0.0040 & -0.0059 & 0.0628 & 0.0599 & 0.9320 \\ 
   &  & $\beta_4$ & 0.0070 & 0.0050 & 0.0836 & 0.0837 & 0.9420 \\ 
   &  & $\gamma$ & 0.0054 & 0.0181 & 0.0713 & 0.0697 & 0.9380 \\ 
   &  & $\lambda$ & 0.0082 & 0.0109 & 0.0899 & 0.0940 & 0.9660 \\ \hline
  1000 & 1000 & $\beta_1$ & 0.0024 & 0.0048 & 0.0493 & 0.0504 & 0.9580 \\ 
   &  & $\beta_2$ & 0.0023 & 0.0000 & 0.0484 & 0.0502 & 0.9660 \\ 
   &  & $\beta_3$ & 0.0030 & -0.0038 & 0.0550 & 0.0565 & 0.9540 \\ 
   &  & $\beta_4$ & 0.0057 & 0.0062 & 0.0753 & 0.0806 & 0.9560 \\ 
   &  & $\gamma$ & 0.0037 & 0.0070 & 0.0607 & 0.0641 & 0.9500 \\ 
   &  & $\lambda$ & 0.0070 & 0.0119 & 0.0828 & 0.0882 & 0.9680 \\ 
   \hline
\hline
\end{tabular}
\begin{tablenotes}
    \item $n_1$: sample size in $\calP$; $n_2$: sample size in
      $\calQ$; MSE: empirical mean squared error; SE: empirical standard error; 
       $\wh{\text{SE}}$: mean of estimated standard errors; CP:
      empirical coverage probability of 95\% confidence intervals.
\end{tablenotes}
\end{table}

\end{document}